%% file: main.tex
\documentclass[10pt]{article}
\usepackage{graphicx} 
\input{preamble}

\title{Weighted Matching in a Poly-Streaming Model\thanks{A preliminary version of this paper appeared in the \emph{European Symposium on Algorithms}, \textit{ESA~2025}.}
}
\author{
  Ahammed Ullah\textsuperscript{\dag} \quad
  S M Ferdous\textsuperscript{\ddag} \quad
  Alex Pothen\textsuperscript{\dag}
}
\date{}

\begin{document}

\maketitle
\renewcommand{\thefootnote}{\fnsymbol{footnote}}
\footnotetext[2]{Purdue University, West Lafayette, IN, USA}
\footnotetext[3]{Pacific Northwest National Laboratory, Richland, WA, USA}

\renewcommand{\thefootnote}{\arabic{footnote}}
\setcounter{footnote}{0}

\input{0_abs}

\input{1_intro}
\input{2_model}
\input{3_uma}
\input{4_numa}
\input{5_evals}
\input{6_conc}

\bibliographystyle{plainurl}
\typeout{}
\bibliography{10_ref.bib}

\clearpage

\appendix
\input{7_app1}
\input{8_app2}
\input{9_app3}

\end{document}

%% file: preamble.tex
\usepackage[margin=1in]{geometry}

\usepackage{amsmath}
\usepackage{amsthm}
\usepackage{amsfonts}
\usepackage{amssymb}

\usepackage{mathpazo}

\usepackage[T1]{fontenc}
\usepackage[utf8]{inputenc}

\usepackage[dvipsnames]{xcolor}
\usepackage[many]{tcolorbox}

\usepackage{booktabs}
\usepackage{multirow}

\usepackage{subcaption}

\usepackage[framemethod=TikZ]{mdframed}

\usepackage{hyperref}
\hypersetup{
    colorlinks=true,
    linkcolor=Blue,
    filecolor=magenta,      
    urlcolor=Sepia,
    citecolor=Sepia,
    pdftitle={Matching in a Poly-Streaming Model},
    pdfauthor={Ahammed Ullah},
    pdfpagemode=FullScreen,
    }

\usepackage{tikz}
\usetikzlibrary{positioning}
\usetikzlibrary{calc}
\usetikzlibrary{patterns}
\usetikzlibrary{decorations.pathmorphing}

\usepackage{pgfplots}

\pgfplotsset{compat=1.11,
    /pgfplots/ybar legend/.style={
    /pgfplots/legend image code/.code={%
       \draw[##1,/tikz/.cd,yshift=-0.25em]
        (0cm,0cm) rectangle (3pt,0.8em);},
   },
}

\usepackage{thmtools,thm-restate}

\newtheorem*{theorem*}{Theorem}

\newtheorem{definition}{Definition}[section]
\newtheorem{proposition}[definition]{Proposition}
\newtheorem{lemma}[definition]{Lemma}
\newtheorem{theorem}[definition]{Theorem}

\newtheorem{remark}[definition]{Remark}

\newcommand{\A}{\mathcal{A}}
\newcommand{\B}{\mathcal{B}}
\newcommand{\D}{\mathcal{D}}
\newcommand{\F}{\mathcal{F}}
\newcommand{\G}{\mathcal{P}}
\newcommand{\M}{\mathcal{M}}
\newcommand{\Mt}{\widetilde{\M}}
\newcommand{\N}{\mathcal{N}}
\newcommand{\T}{\mathcal{T}}
\newcommand{\U}{\mathcal{U}}

\newcommand{\hyref}[2]{\hyperref[#2]{#1~\ref*{#2}}}
\newcommand{\aref}[2]{\hyperref[#2]{#1}}
\newcommand{\fref}[1]{\hyref{Figure}{#1}}
\newcommand{\lemref}[1]{\hyref{Lemma}{#1}}

\newcommand{\lmax}{L_{max}}
\newcommand{\lmin}{L_{min}}

\newcommand{\br}[1]{\left(#1\right)}

\newcommand{\plog}[1]{\text{polylog }#1}

\newcommand{\bigO}[1]{\mathcal{O}\br{#1}}
\newcommand{\bigOmega}[1]{\Omega\br{#1}}
\newcommand{\bigTheta}[1]{\Theta\br{#1}}

\newcommand{\smallO}[1]{o\br{#1}}

\newcommand{\bigOtilde}[1]{\widetilde{\mathcal{O}}\br{#1}}

\newcommand{\Erdos}{Erd\H{o}s}
\newcommand{\Renyi}{R\'{e}nyi}
\renewcommand{\epsilon}{\varepsilon}

%% file: 0_abs.tex
\begin{abstract}
We introduce the \emph{poly-streaming model}, a generalization of streaming models of computation in which $k$ processors process $k$ data streams containing a total of $N$ items. 
The algorithm is allowed $\bigO{f(k)\cdot M_1}$ space, where $M_1$ is either $\smallO{N}$ or the space bound for a sequential streaming algorithm.
Processors may communicate as needed. 
Algorithms are assessed by the number of passes, per-item processing time, total runtime, space usage, communication cost, and solution quality.

We design a \emph{single-pass} algorithm in this model for approximating the \emph{maximum weight matching} (MWM) problem. Given $k$ edge streams and a parameter $\epsilon > 0$, the algorithm computes a $\br{2+\epsilon}$-approximate MWM.
We analyze its performance in a shared-memory parallel setting: for any constant $\epsilon > 0$, it runs in time $\bigOtilde{\lmax+n}$, where $n$ is the number of vertices and $\lmax$ is the maximum stream length.
It supports $\bigO{1}$ per-edge processing time using $\bigOtilde{k\cdot n}$ space.
We further generalize the design to hierarchical architectures, in which $k$ processors are partitioned into $r$ groups, each with its own shared local memory. The total intergroup communication is $\bigOtilde{r \cdot n}$ bits, while all other performance guarantees are preserved.

We evaluate the algorithm on a shared-memory system using graphs with trillions of edges. It achieves substantial speedups as $k$ increases and produces matchings with weights significantly exceeding the theoretical guarantee. On our largest test graph, it reduces runtime by nearly two orders of magnitude and memory usage by five orders of magnitude compared to an offline algorithm.
\end{abstract}

%% file: 1_intro.tex
\section{Introduction}
\label{sec:intro}

Data-intensive computations arise in data science, machine learning, and science and engineering disciplines.
These datasets are often massive, generated dynamically, and, when stored, kept in distributed formats on disks, making them amenable to processing as multiple data streams. 
The modularity of these datasets can be exploited by streaming algorithms designed for tightly-coupled shared-memory and distributed-memory multiprocessors to efficiently solve large problem instances that offline algorithms cannot handle due to their high memory requirements.
However, the design of parallel algorithms that process multiple data streams concurrently has not yet received much attention.

Current multicore shared-memory processors consist of up to a few hundred cores, organized hierarchically to share caches and memory controllers. These cores compute in parallel to achieve speedups over serial execution. 
With multiple memory controllers, I/O operations can also proceed in parallel, and this feature can be used to process multiple data streams concurrently.
These I/O capabilities and the limitations of offline algorithms motivate a model of computation, illustrated in \fref{fig:model_shared} and discussed next.

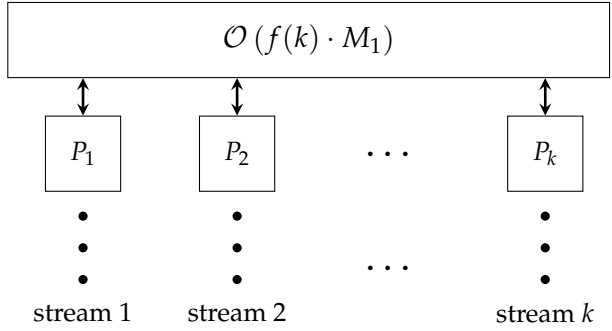
\begin{figure}
\centering
\input{Tikz/shared}
\caption{A schematic diagram of the \emph{poly-streaming model} for shared-memory parallel computers. Processors $\{P_{\ell}\}_{\ell \in [k]}$ have access to $\bigO{f(k)\cdot M_1}$ memory collectively, depicted with the rectangle connected to the processors.}
\label{fig:model_shared}
\end{figure}

The streaming model of computation allows $\smallO{N}$ space for a data stream of size $N$~\cite{alon1996space, henzinger1998computing}. For graphs, the \emph{semi-streaming model} permits $\bigO{n \cdot \plog{n}}$ space for a graph with $n$ vertices and an edge stream of arbitrary length~\cite{feigenbaum2005graph}. Building on these space-constrained models, we introduce the \emph{poly-streaming model}. The key aspects of our model are as follows.

We consider $k$ data streams that collectively contain $N$ items. An algorithm has access to $k$ (abstract) processors and is allowed $\bigO{f(k)\cdot M_1}$ total space, where $M_1$ is either $\smallO{N}$ or the space permitted to a single-stream algorithm. In each pass, each stream is assigned to one of the processors, and each processor independently reads one item at a time from its stream and processes it. Processors may communicate as needed, either via shared or remote memory access. Algorithms are assessed on several metrics: space complexity, number of passes, per-item processing time, total runtime, communication cost, and solution quality.

In the poly-streaming model, we address the problem of approximating a \emph{maximum weight matching} (MWM) in an edge-weighted graph, where the goal is to find a set of vertex-disjoint edges with maximum total weight. We design an algorithm for approximating an MWM when the graph is presented as multiple edge streams. Our design builds on the algorithm of~\cite{paz20182+} and adds support for handling multiple streams concurrently. We also generalize the design to NUMA (non-uniform memory access) multiprocessor architectures.

We summarize our contributions to the MWM problem as follows. Let $\lmax$ and $\lmin$ denote the maximum and minimum lengths of the input streams, respectively, and let $n$ denote the number of vertices in a graph $G$. For any realization of the CREW PRAM model (such as in \fref{fig:model_shared}), we have the following result.

\begin{theorem}
\label{thm_psmwm}
For any constant $\epsilon > 0$, there exists a \emph{single-pass} poly-streaming algorithm for the maximum weight matching problem that achieves a $\br{2+\epsilon}$-approximation. It admits a CREW PRAM implementation with runtime $\bigOtilde{\lmax + n}$.\footnote{$\bigOtilde{\cdot}$ hides polylogarithmic factors.} 

If $\lmin = \Omega\br{n}$, the algorithm achieves $\bigO{\log n}$ amortized per-edge processing time using $\bigOtilde{k + n}$ space. For arbitrarily balanced streams, it uses either: 
\begin{itemize}
    \item $\bigOtilde{k+n}$ space and $\bigOtilde{n}$ per-edge processing time, or
    \item $\bigOtilde{k\cdot n}$ space and $\bigO{1}$ per-edge processing time.
\end{itemize}
\end{theorem}

In NUMA architectures, memory access costs depend on a processor's proximity to the target memory. We generalize the algorithm in \hyref{Theorem}{thm_psmwm} to account for these cost differences. In particular, we show that when $k$ processors are partitioned into $r$ groups, each with its own shared local memory, the total number of global memory accesses across all groups is $\bigOtilde{r \cdot n}$. This generalization preserves all other performance guarantees from \hyref{Theorem}{thm_psmwm}, except that the $\bigOtilde{k + n}$ space bound becomes $\bigOtilde{k + r \cdot n}$. 
These results are formalized in \hyref{Theorem}{thm_psmwm_ld} 
in \hyref{Section}{sec:numa}.
This design gives a memory-efficient algorithm for the NUMA shared memory multiprocessors, on which we report empirical results.

We have evaluated our algorithm on a NUMA machine using graphs with billions to trillions of edges. For most of these graphs, our algorithm uses space that is orders of magnitude smaller than that required by offline algorithms. For example, storing the largest graph in our evaluation would require more than $91{,}600\,\mathrm{GB}$ ($\approx 90\,\mathrm{TB}$), whereas our algorithm used less than $1\,\mathrm{GB}$.
Offline matching algorithms typically require even more memory to accommodate their auxiliary data structures.

We employ approximate dual variables that correspond to a linear programming relaxation of MWM to obtain \emph{a posteriori} upper bounds on the weights of optimal matchings. These bounds allow us to compare the weight of a matching produced by our algorithm with the optimal weight. Thus, we show that our algorithm produces matchings whose weights significantly exceed the approximation guarantee.

For $k=128$, our algorithm achieves runtime speedups of $16$--$83$ across all graphs in our evaluation, on a NUMA machine with only $8$ memory controllers. This is significant scaling for a poly-streaming algorithm, given that $8$ memory controllers are not sufficient to serve the concurrent and random access requests of $128$ processors without delays. 
Nevertheless, these speedups demonstrate the effectiveness of our design, which accounts for a processor's proximity to the target memory.
A metric less influenced by memory latency suggests that the algorithm would achieve even better speedups on architectures with more efficient memory access.

Note that \hyref{Theorem}{thm_psmwm} and \hyref{Theorem}{thm_psmwm_ld} both guarantee $\bigOtilde{\lmax + n}$ runtime. This is optimal up to polylogarithmic factors when $\lmax = \bigOmega{n}$. However, by using $\bigOtilde{k \cdot n}$ space and $\bigO{1}$ per-edge processing time, we can achieve a runtime of $\bigOtilde{\lmax + n/k}$, which becomes polylogarithmic for sufficiently large $k$ (see \hyref{Appendix}{subsec:plog_time}).

\paragraph{Organization.} \hyref{Section}{sec:model} describes the details of our model. \hyref{Section}{sec:mwm} presents the design and analyses of our algorithm in \hyref{Theorem}{thm_psmwm}. In \hyref{Section}{sec:numa}, we extend the design to NUMA architectures. \hyref{Section}{sec:evals} summarizes the evaluation results. We conclude in \hyref{Section}{sec:conc} with a discussion of future research directions.

%% file: Tikz/shared.tex
\begin{tikzpicture}[>=stealth, node distance=2cm, every node/.style={rectangle, draw, minimum width=1.0cm, minimum height=1.0cm}]

    \node (mem) [minimum width=8cm] {\scalebox{1.1}{$\bigO{f(k)\cdot M_1}$}};

    \node (c1) [draw=none, below=-0.25 cm of mem, xshift=-3.0cm, minimum width=0cm, minimum height=0cm] {};
    \node (c2) [draw=none, right=1.8cm of c1, minimum width=0cm, minimum height=0cm] {};
    \node (c3) [draw=none, right=1.8cm of c2, minimum width=0cm, minimum height=0cm] {};
    \node (c4) [draw=none, right=1.8cm of c3, minimum width=0cm, minimum height=0cm] {};
    
    \node (p1) [below=0.5 cm of c1] {\scalebox{1.0}{$P_1$}};
    \node (p2) [below=0.5 cm of c2] {\scalebox{1.0}{$P_2$}};
    \node (dots) [draw=none, below=0.5 cm of c3] {\scalebox{1.5}{$\dots$}};
    \node (pk) [below=0.5 cm of c4] {\scalebox{1.0}{$P_k$}};

    \node (s11) [draw=none, below=-0.5 cm of p1, minimum width=0.5cm, minimum height=1cm] {\scalebox{3.0}{$\vdots$}};
    \node (s12) [draw=none, below=-0.5 cm of p2, minimum width=0.5cm, minimum height=1cm] {\scalebox{3.0}{$\vdots$}};
    \node (dots) [draw=none, below=0.5 cm of dots] {\scalebox{1.5}{$\dots$}};
    \node (s13) [draw=none, below=-0.5 cm of pk, minimum width=0cm, minimum height=1cm] {\scalebox{3.0}{$\vdots$}};

    \node (s31) [draw=none, below=0.0 cm of s11, minimum width=0.5cm, minimum height=0.5cm] {\scalebox{1.0}{stream 1}};
    \node (s32) [draw=none, below=0.0 cm of s12, minimum width=0.5cm, minimum height=0.5cm] {\scalebox{1.0}{stream 2}};
    \node (s33) [draw=none, below=0.0 cm of s13, minimum width=0.5cm, minimum height=0.5cm] {\scalebox{1.0}{stream $k$}};

    \draw[<->, line width=1.0pt] (p1) -- (c1);
    \draw[<->, line width=1.0pt] (p2) -- (c2);
    \draw[<->, line width=1.0pt] (pk) -- (c4);

\end{tikzpicture}

%% file: 2_model.tex
\section{The Poly-Streaming Model}
\label{sec:model}

This section elaborates on our model of computation and discusses its novelty and significance relative to existing models (\hyref{Section}{subsec:ps_sig}).

In the poly-streaming model, there are $k$ data streams containing a total of $N$ items. 
An algorithm may use $k$ processors and up to $\bigO{f(k)\cdot M_1}$ total space, where $M_1$ is either $\smallO{N}$ or the space permitted to a single-stream algorithm (as in the semi-streaming model). In each pass, each stream is assigned to a processor, and processors independently read items from their respective streams. Processing these items may require coordination. An algorithm may use the processors to perform any necessary preprocessing and post-processing. Processors may communicate as needed during preprocessing, streaming, or post-processing, via shared or remote memory. 

Note that the $\bigO{f(k)\cdot M_1}$ space constraint subsumes the $\bigO{f(k) + M_1}$ constraint. We now describe the model's components in more detail.

\paragraph{Processors.}
Each processor is an abstract unit of computation that can be emulated by a physical thread on a shared-memory or tightly coupled distributed-memory machine. \fref{fig:model_shared} illustrates a realization in which all processors directly access a shared workspace, corresponding to a shared-memory implementation.
Multiple such realizations can be connected via high-speed networks to implement the model on a distributed-memory machine.

\paragraph{Data Streams.}
The model assumes an arbitrary distribution of data across streams. Algorithms must handle arbitrary inputs with imbalanced partitioning and arbitrary item orderings. A stream may be assigned to a processor multiple times, each assignment constituting a pass. Within a pass, streams are read asynchronously, though processing individual items may require synchronization.
The parameter $k$ need not equal the number of physical input streams: physical streams may be merged or split into $k$ logical streams, which are then mapped to processors.

\paragraph{Space.} 
The bound $\bigO{f(k)\cdot M_1}$ reflects the observation that, in practice, total memory typically scales with the number of processors.
In most cases, $f(k)$ is expected to be linear in $k$, but 
superlinear growth may still be feasible, particularly for algorithms that use $\bigO{f(k)+M_1}$ space.
This formulation supports the analysis of a broad range of design choices and their associated trade-offs. It also enables a bottom-up design approach, where algorithms developed for shared-memory machines can be extended to tightly coupled distributed-memory machines.

\paragraph{Per-Item Processing Time.}
A key consideration in the poly-streaming setting is whether an algorithm can handle an influx of items arriving in quick succession, as may occur when $k^{\prime} \gg k$ physical streams are merged into $k$ logical streams.
If the algorithm cannot handle such influxes within bounded space, its correctness may be compromised.
For suitable choices of $f(\cdot)$, the $f(k)$-fold space may suffice to design bounded-space algorithms for many such scenarios.

For some design choices, the worst-case per-item processing time may not be informative. 
In such cases, under realistic assumptions, amortized or average per-item processing time may better reflect actual performance.
In particular, amortizing over the number of items per stream, rather than over the total input, can yield a more accurate estimate of this cost.

\paragraph{Runtime.} 
The runtime refers to the total time spent on preprocessing, streaming, and post-processing across all passes. It includes delays caused by contention when accessing shared or remote memory. The cost of remote memory access is assumed to be proportional to the level of contention at the target location.

The runtime of an algorithm should, in general, be dominated by a function of the maximum stream length, denoted $\lmax$. It may also depend on other parameters, such as $M_1$, which remains non-dominating for $M_1 = \bigO{\lmax}$. 
Under worst-case data distribution, $\lmax = \bigO{N}$. 
As stream lengths become more balanced, that is, as $\lmax$ approaches its lower bound $\bigTheta{N/k}$, the runtime should scale accordingly.
In such balanced settings, an algorithm may leverage the $f(k)$-fold space to match the runtime of scalable offline algorithms, for example, achieving polylogarithmic runtime for sufficiently large $k$.

\paragraph{Solution Quality.}
Poly-streaming algorithms are generally expected to admit provable bounds on solution quality, such as approximation ratios. These may be complemented by empirical performance bounds, such as a posteriori guarantees based on upper or lower bounds on the optimal. Such guarantees are particularly important, since streaming algorithms are often provably unable to compute exact solutions within a few passes. The space constraint may also facilitate the exploration of trade-offs between space and solution quality.

\paragraph{Number of Passes.}
A central goal in this model is to design single-pass algorithms. Many initial designs may require multiple passes, with single-pass algorithms emerging only after substantial algorithmic development. In some cases, multiple passes may be provably necessary to achieve objectives such as stronger approximation guarantees under tighter space bounds. Thus, the number of passes serves as a fundamental measure of algorithmic efficiency.

\paragraph{Communication.} 
The communication cost of an algorithm is defined as the total number of remote memory accesses. This abstraction excludes interconnection latency and other architecture-specific delays, as is standard in theoretical models to simplify algorithm design.

\subsection{Novelty and Significance}
\label{subsec:ps_sig}

For descriptions of existing models referenced here, see \hyref{Appendix}{sec:related_models}.

\paragraph{Parallel Computation.}
The poly-streaming model targets areas of computation beyond the reach of traditional parallel models, such as the work-depth model. 
In terms of input scale, poly-streaming algorithms are designed for datasets that offline parallel algorithms cannot handle due to their impractical memory requirements.

Another key distinction is that offline parallel algorithms assume random access to the entire input. 
In contrast, a central motivation for streaming models is to minimize expensive random accesses to massive, persistent datasets. The goal is to replace many random accesses with a small number of sequential passes, which are typically more efficient in practice.

Modern parallel file systems support concurrent, high-throughput access to data by allowing multiple simultaneous connections.
By leveraging the parallel I/O capabilities of modern shared-memory machines, poly-streaming algorithms can exploit these systems to efficiently process massive datasets while avoiding costly random accesses.

\paragraph{Distributed Computation.}
The poly-streaming model supports asynchronous communication protocols, in contrast to models that count synchronous communication rounds, such as the MPC model and the distributed streaming model. In tightly coupled shared- and distributed-memory multiprocessors, synchronous coordination is often unnecessary for designing communication\allowbreak-efficient algorithms, particularly when architectures support remote memory access. In systems based on message passing, such access can be emulated by assigning processors to mediate access to shared locations via messages.

\hyref{Appendix}{subsec:dist} sketches the design of a distributed algorithm based on asynchronous communication. This algorithm achieves optimal communication cost (up to polylogarithmic factors), supports streaming computation, and dominates comparable MPC algorithms across several metrics. Moreover, it is \emph{single-pass}, which is unlikely to be achievable under the synchronous communication protocols of existing distributed streaming models.

\paragraph{Streaming Computation.}
Traditional offline parallel models provide frameworks for optimizing time in isolation, while sequential streaming models, such as those described in \hyref{Appendix}{sec:related_models}, focus on optimizing space in isolation. The poly-streaming model offers a unified framework for optimizing both time and space jointly. Its support for asynchronous communication protocols enables the design of parallel algorithms that are not permitted in existing models, such as the distributed streaming model. \hyref{Section}{sec:mwm} and \hyref{Section}{sec:numa} present examples of such algorithms.

\paragraph{Analyzing Trade-offs.} 
A central theme in streaming literature is that space constraints often conflict with other performance metrics, such as solution quality, number of passes, and per-item processing time.
Analyzing the trade-offs between space and these metrics remains an active area of research; see \cite{assadi2024simple, feigenbaum2009graph, assadi2019coresets} for examples. Moreover, processing multiple streams concurrently may require trade-offs that do not arise in the single-stream setting; see \hyref{Section}{sec:mwm}, \hyref{Section}{sec:numa}, and \hyref{Appendix}{subsec:mwm_ds} for examples of time--space trade-offs. The space constraint in the poly-streaming model provides a unified framework for analyzing such trade-offs

\paragraph{Hierarchical Design.}
The space constraint in the poly-streaming model enables a bottom-up design approach, where algorithms developed for shared-memory systems can be extended to tightly coupled distributed\allowbreak-memory systems.
This requires algorithm designers to account for memory hierarchy in order to manage and quantify communication costs of the resulting algorithms.
\hyref{Section}{subsec:mwm_ps}, \hyref{Section}{sec:numa}, and \hyref{Appendix}{subsec:dist} collectively illustrate such a hierarchical design process.

\paragraph{Practical Relevance.}
Modern computing environments are inherently multicore, with total memory typically scaling with the number of cores.
Yet, such environments often fail to meet the space requirements of offline parallel algorithms for large problem instances. 
Conversely, sequential streaming algorithms underutilize both cores and memory, as they are not designed to exploit multicore architectures. 
These limitations warrant new paradigms of computation, as directly addressed by the poly-streaming model.

%% file: 3_uma.tex
\section{Algorithms for Uniform Memory Access Cost}
\label{sec:mwm}

In this section, we present the design and analyses of our algorithm in \hyref{Theorem}{thm_psmwm} that assumes a uniform memory access cost.

\subsection{Preliminaries}
\label{subsec:prelim}

For a graph $G=(V,E)$, let $n:=|V|$ and  $m:=|E|$ denote the number of vertices and edges, respectively. We denote an edge $e:=\{u,v\}$ by the unordered pair of its endpoints. Let $\N\br{e}$ be the set of edges in $E$ that share an endpoint with edge $e$. For a weighted graph, let $w_e$ denote the weight of edge $e$, and for any subset $A \subseteq E$, define $w(A):=\sum_{e \in A} w_e$. For $\ell \in [k]$, let $E^{\ell}$ be the set of edges received in the $\ell$th stream. Define $\lmax :=\max_{\ell \in [k]} |E^{\ell}|$ and $\lmin :=\min_{\ell \in [k]} |E^{\ell}|$.

A \emph{matching} $\M \subseteq E$ is a set of edges that do not share endpoints. A \emph{maximum weight matching} (MWM) $\M^*$ is a matching with maximum total weight; that is, $w\br{\M^{*}} \geq w\br{\M}$ for all matchings $\M \subseteq E$.

A $\rho$-approximation algorithm computes a solution whose value is within a factor $\rho$ of the \emph{optimal}. The factor $\rho$ is called the (worst-case) \emph{approximation ratio}. We assume $\rho \geq 1$ for both maximization and minimization problems. Thus, for maximization, a $\rho$-approximation guarantees a solution whose value is at least $\frac{1}{\rho}$ times the optimal.

\begin{figure}[h]
    \centering
    \begin{minipage}{0.49\textwidth}  
        \begin{mdframed}[linewidth=0.5pt, roundcorner=7pt, backgroundcolor=gray!5, frametitle={\underline{Primal LP}}, frametitlebelowskip=-4pt]
            \begin{equation*}
                \begin{array}{lll}
                \text{maximize}  & \displaystyle\sum\limits_{e \in E} w_e x_e  \\
                \text{subject to}& \displaystyle\sum\limits_{e \in \delta(u)} x_e  \leq 1,& \text{for all } u \in V\\        
                  & x_e \geq 0, &\text{for all } e \in E
                \end{array}            
            \end{equation*}        
        \end{mdframed}                    
    \end{minipage}
    \hspace{\fill}
    \begin{minipage}{0.49\textwidth}
        \begin{mdframed}[linewidth=0.5pt, roundcorner=7pt, backgroundcolor=gray!5, frametitle={\underline{Dual LP}}, frametitlebelowskip=-4pt]
            \begin{equation*}
                \begin{array}{lll}
                \text{minimize}  & \displaystyle\sum\limits_{u \in V} y_u  \\
                \text{subject to}& \displaystyle\sum\limits_{u \in e} y_u  \geq w_e,& \text{for all } e \in E\\               
                  & y_u \geq 0, &\text{for all } u \in V
                \end{array}            
            \end{equation*}        
        \end{mdframed}        
    \end{minipage}
    \caption{The linear programming (LP) relaxations of the MWM problem and its dual.}
    \label{fig:lp}
\end{figure}

We use the linear programming (LP) relaxation of the MWM problem and its dual, shown in \fref{fig:lp}.
In the primal LP, each variable $x_e$ is $1$ if edge $e$ is in the matching and $0$ otherwise. Each $y_u$ is a dual variable, and $\delta(u)$ denotes the set of edges incident on a vertex $u$.
Let $\{x_e\}_{e \in E}$ and $\{y_u\}_{u \in V}$ be feasible solutions to the primal and dual LPs, respectively. 
By weak LP duality, we have $\sum_{e \in E}w_e x_e \leq \sum_{u \in V} y_u$. 
If $\{x_e\}_{e \in E}$ is an optimal solution to the primal LP, then $w\br{\M^{*}} \leq \sum_{e \in E}w_ex_e \leq \sum_{u \in V} y_u$. 
The first inequality holds because the primal LP is a relaxation of the MWM problem.

\subsection{The Algorithm}
\label{subsec:mwm_ps}

Several semi-streaming algorithms have been developed for the MWM problem~\cite{assadi2024simple, crouch2014improved, epstein2011improved, feigenbaum2005graph,   gamlath2019weighted, ghaffari2019simplified, mcgregor2005finding, paz20182+, zelke2012weighted} (see \hyref{Section}{subsec:mwm_related} for brief descriptions of these algorithms).
In this paper, we focus exclusively on the single-pass setting in the \emph{poly-streaming} model. Our starting point is the algorithm of Paz and Schwartzman~\cite{paz20182+}, which computes a $2+\epsilon$-approximation of MWM. This is currently the best known guarantee in the single-pass setting under arbitrary or adversarial ordering of edges\footnote{No single-pass algorithm can achieve an approximation ratio better than $1+\ln 2 \approx 1.7 $; see~\cite{kapralov2021space}.}. We extend a primal-dual analysis by Ghaffari and Wajc~\cite{ghaffari2019simplified} to analyze our algorithm.

The algorithm of Paz and Schwartzman~\cite{paz20182+} proceeds as follows. Initialize an empty stack $S$ and set $\alpha_u=0$ for each vertex $u \in V$. For each edge $e=\{u,v\}$ in the edge stream, skip $e$ if $w_e < \br{1+\epsilon}\br{\alpha_u + \alpha_v}$. Otherwise, compute $g_e = w_e - \br{\alpha_u + \alpha_v}$, push $e$ onto the stack $S$, and increase both $\alpha_u$ and $\alpha_v$ by $g_e$. After processing all edges, compute a matching $\M$ greedily by popping edges from $S$.

Note that for each edge pushed onto the stack, the increment $g_e=w_e - \br{\alpha_u+\alpha_v}$ satisfies $g_e \geq \epsilon\br{\alpha_u+\alpha_v}$. This ensures that both $\alpha_u$ and $\alpha_v$ increase by a factor of $1+\epsilon$. Hence, the number of edges in the stack incident to any vertex is at most $\log_{1+\epsilon}(W) = \bigO{\frac{\log W}{\epsilon}}$, where $W$ is the (normalized) maximum edge weight. Therefore, the total number of edges in the stack is $\bigO{\frac{n\log W}{\epsilon}} =\bigO{\frac{n\log n}{\epsilon}}$.\footnote{Throughout the paper, we assume $W = \bigO{poly(n)}$. For arbitrary weights on edges, we can skip any edge whose weight is less than $\frac{\epsilon W_{max}}{2\br{1+\epsilon}n^2}$, where $W_{max}$ denotes the maximum edge weight observed so far in the stream. This ensures that the (normalized) maximum weight the algorithm sees is $\bigO{n^2/\epsilon}$, while maintaining a $2\br{1+\bigO{\epsilon}}$ approximation ratio (see~\cite{ghaffari2019simplified} for details).}

To design a poly-streaming algorithm, we begin with a simple version and then refine it. All $k$ processors share a global stack and a set of variables $\{\alpha_u\}_{u \in V}$, and each processor runs the above sequential streaming algorithm on its respective stream. To complete and adapt this setup for efficient execution across multiple streams, we must address two interrelated issues: (1) concurrent edge arrivals across streams may lead to contention for the shared stack or variables, and (2) concurrent updates to the shared variables may lead to inconsistencies in their observed values.

A natural approach to addressing these issues is to enforce a fair sequential strategy, where processors access shared resources in a round-robin order. While this ensures progress, it incurs $\bigO{k}$ per-edge processing time, which scales poorly with increasing $k$. Instead, we adopt fine-grained contention resolution that avoids global coordination by allowing processors to operate asynchronously. However, under the initial setup, this leads to $\bigOtilde{n / \epsilon}$ per-edge processing time: a processor may be blocked from accessing shared resources until the stack has accumulated its $\bigOtilde{n / \epsilon}$ potential edges. We address these limitations with the following design choices.

\begin{figure}[h]
    \centering
    \begin{parbox}{4.8in}{    
        \begin{mdframed}[linewidth=0.5pt, roundcorner=7pt, backgroundcolor=gray!5, frametitle={\underline{PS-MWM$(V, \ell, \epsilon)$}}]        
        \textcolor{teal}{/* each processor executes this algorithm concurrently */}
        \begin{enumerate}            
            \item In parallel initialize $lock_u$, and set $\alpha_u$ and $mark_u$ to $0$ for all $u \in V$ \textcolor{teal}{\\/* processor $\ell$ initializes or sets $\bigTheta{n/k}$ locks/variables */}        
            \item  $S^{\ell} \leftarrow \emptyset$ \textcolor{teal}{/* initialize an empty stack */}
            \item for each edge $e=\{u,v\}$ in $\ell${th} stream do
            \begin{enumerate}
                \item \aref{Process-Edge$(e, S^{\ell}, \epsilon$)}{fig:algo_proc_edge}
            \end{enumerate}
            \item wait for all processors to complete execution of Step~3 \textcolor{teal}{/* a barrier */}
            \item $\M^{\ell} \leftarrow$ \aref{Process-Stack$(S^{\ell})$}{fig:algo_proc_stack}
            \item return $\M^{\ell}$
        \end{enumerate}
        \end{mdframed}    
    
    }
    \end{parbox}    
    \caption{A poly-streaming matching algorithm.}
    \label{fig:algo_psmwm}
\end{figure}

For the first issue, we observe that a global ordering of edges, as used in the single-stack solution, is not necessary; local orderings within multiple stacks suffice. In particular, we can identify a subset of edges (later referred to as \emph{tight edges}) for which maintaining local orderings is sufficient to compute a $2+\epsilon$-approximate MWM.
Hence, we can localize computation using $k$ stacks, assigning one stack to each processor exclusively during the streaming phase. This design eliminates the $\bigOtilde{n / \epsilon}$ contention associated with a shared stack. 

However, contention still arises when updating the variables $\{\alpha_u\}_{u \in V}$. 
It is unclear how to resolve this contention without using additional space. 
Hence, we consider two strategies for processing edge streams that illustrate the trade-off between space and per-edge processing time. In the first, which we call the \emph{non-deferrable strategy}, the decision to include an edge in a stack is made immediately during streaming. In the second, which we call the \emph{deferrable strategy}, this decision may be deferred to post-processing. The latter strategy requires more space but achieves $\bigO{1}$ per-edge processing time.

To address the second issue, which concerns the potential for inconsistencies due to concurrent updates to the variables $\{\alpha_u\}_{u \in V}$, we observe that the variables are monotonically increasing and collectively require only $\bigOtilde{n / \epsilon}$ updates. Thus, for most edges that are not eligible for the stacks, decisions can be made by simply reading the current values of the relevant variables. However, for the $\bigOtilde{n / \epsilon}$ edges that are included in the stacks, we must update the corresponding variables. To ensure consistency of these updates, we associate a lock with each variable in $\{\alpha_u\}_{u \in V}$. We maintain $|V|$ exclusive locks and allow a variable to be updated only after acquiring its corresponding lock.\footnote{This corresponds to the concurrent-read exclusive-write (CREW) paradigm of the PRAM model.}

We now outline the non-deferrable strategy of our poly-streaming algorithm for the MWM problem (for the deferrable strategy see \hyref{Appendix}{subsec:mwm_ds}). For simplicity, we assume that if a processor attempts to release a lock it did not acquire, the operation has no effect. We also assume that any algorithmic step described with the "in parallel" construct includes an implicit barrier (or synchronization primitive) at the end, synchronizing the processors participating in that step.

\begin{figure}[h]
    \centering
    \begin{parbox}{4.6in}{    
        \begin{mdframed}[linewidth=0.5pt, roundcorner=7pt, backgroundcolor=gray!5, frametitle={\underline{Process-Edge$(e=\{u,v\}, S^{\ell}, \epsilon)$}}]
        \textcolor{teal}{/* Assumes access to global variables $\{\alpha_x\}_{x \in V}$ 
 and locks $\{lock_x\}_{x \in V}$ */}
    \begin{enumerate}
        \item if $w_e \leq (1+\epsilon)(\alpha_u +\alpha_v)$ then return
        \item repeatedly try to acquire $lock_u$ and $lock_v$ in lexicographic order of $u$ and $v$ as long as $w_e > (1+\epsilon)(\alpha_u +\alpha_v)$
        \item if $w_e > (1+\epsilon)(\alpha_u +\alpha_v)$ then
        \begin{enumerate}
            \item $g_e \leftarrow w_e - (\alpha_u +\alpha_v)$
            \item increment $\alpha_u$ and $\alpha_v$ by $g_e$
            \item add $e$ to the top of $S^{\ell}$ along with $g_e$
        \end{enumerate}
        \item release $lock_u$ and $lock_v$, and return        
    \end{enumerate}
        \end{mdframed}    
 
    }
    \end{parbox}    
    \caption{A subroutine used in algorithms \aref{PS-MWM}{fig:algo_psmwm}, \aref{PS-MWM-DS}{fig:algo_psmwm_ds}, and \aref{PS-MWM-LD}{fig:algo_psmwm_ld}.}
    \label{fig:algo_proc_edge}
\end{figure}

The non-deferrable strategy is presented in \aref{Algorithm PS-MWM}{fig:algo_psmwm}, with two subroutines used by \hyperref[fig:algo_psmwm]{PS-MWM} described in \hyperref[fig:algo_proc_edge]{Process-Edge} (\hyref{Figure}{fig:algo_proc_edge}) and \hyperref[fig:algo_proc_stack]{Process-Stack} (\hyref{Figure}{fig:algo_proc_stack}). 
In \aref{PS-MWM}{fig:algo_psmwm}, Steps~1--2 form the preprocessing phase, Steps~3--4 the streaming phase, and Step~5 the post-processing phase. Each processor $\ell \in [k]$ executes \aref{PS-MWM}{fig:algo_psmwm} asynchronously, except that all processors begin the post-processing phase simultaneously (due to Step~4) and then resume asynchronous execution. 

In the subroutine \hyperref[fig:algo_proc_edge]{Process-Edge}, Step~2 ensures that all edges are processed using the non-deferrable strategy: a processor repeatedly attempts to acquire the locks corresponding to the endpoints of an edge $e=\{u,v\}$ until it succeeds or the edge becomes ineligible for inclusion in a stack. As a result, a processor executing Step~3 has a consistent view of the variables $\alpha_u$ and $\alpha_v$. In Step~3(c), we store the \emph{gain} $g_e$ of an edge $e$ along with the edge itself for use in the post-processing phase.

\begin{figure}
    \centering
    \begin{parbox}{4.6in}{    
        \begin{mdframed}[linewidth=0.5pt, roundcorner=7pt, backgroundcolor=gray!5, frametitle={\underline{Process-Stack$(S^{\ell})$}}]
        \textcolor{teal}{/* Assumes access to global variables $\{\alpha_x\}_{x \in V}$ and $\{mark_x\}_{x \in V}$ */}
    \begin{enumerate}
        \item $\M^{\ell} \leftarrow \emptyset$                
        \item while $S^{\ell} \ne \emptyset$ do                
            \begin{enumerate}
                \item remove the top edge $e = \{u, v\}$ of $S^{\ell}$
                \item if $w_e + g_e < \alpha_u + \alpha_v$ then wait for $e$ to be a tight edge \textcolor{teal}{\\/* $e$ is a \emph{tight edge} if $w_e+g_e = \alpha_u+\alpha_v$ */}
                \item if both $mark_u$ and $mark_v$ are set to $0$ then \textcolor{teal}{\\/* no locking is needed since $e$ is a tight edge */}
                \begin{enumerate}
                    \item $\M^{\ell} \leftarrow \M^{\ell} \cup \{e\}$
                    \item set $mark_u$ and $mark_v$ to $1$            
                \end{enumerate} 
                \item decrement $\alpha_u$ and $\alpha_v$ by $g_e$
            \end{enumerate}                                
        \item return $\M^{\ell}$
    \end{enumerate}
        \end{mdframed}    
 
    }
    \end{parbox}    
    \caption{A subroutine used in algorithms \aref{PS-MWM}{fig:algo_psmwm}, \aref{PS-MWM-DS}{fig:algo_psmwm_ds}, and \aref{PS-MWM-LD}{fig:algo_psmwm_ld}.}
    \label{fig:algo_proc_stack}
\end{figure}

When all $k$ processors are ready to execute Step~5 of \hyperref[fig:algo_psmwm]{PS-MWM}, the $k$ stacks collectively contain all the edges needed to construct a $\br{2+\epsilon}$-approximate MWM, which can be obtained in several ways. In the subroutine \hyperref[fig:algo_proc_stack]{Process-Stack}, we outline a simple approach based on local edge orderings. We define an edge $e=\{u,v\}$ in a stack to be a \emph{tight edge} if $w_e + g_e = \alpha_u + \alpha_v$. Equivalently, an edge is tight if and only if all of its neighboring edges that were included after it in any stack have already been removed. Any set of tight edges can be processed concurrently, regardless of their positions in the stacks. In \hyperref[fig:algo_proc_stack]{Process-Stack}, we simultaneously process the tight edges that appear at the tops of the stacks.

\subsection{Analyses}
\label{subsec:mwm_analyses}

We now formally characterize several correctness properties of the algorithm and analyze its performance. These correctness properties include the absence of deadlock, livelock, and starvation. The performance metrics are space usage, approximation ratio, per-edge processing time, and total runtime. 

To simplify the analysis, we assume that processors operate in a quasi-synchronous manner. In particular, to analyze Step~3 of \aref{Algorithm PS-MWM}{fig:algo_psmwm}, we define an algorithmic \emph{superstep} as a unit comprising a constant number of elementary operations.

\begin{definition}[Superstep]
\label{def_ss} 
A processor takes one \emph{superstep} for an edge if it executes \hyperref[fig:algo_proc_edge]{Process-Edge} with at most one iteration of the loop in Step~2 (i.e., without contention), requiring $\bigO{1}$ elementary operations. Each additional iteration of Step~2 due to contention adds one superstep, with each such iteration also requiring $\bigO{1}$ operations.
\end{definition}

\begin{definition}[Effective Iterations]
\label{def_ei}
\emph{Effective iterations} is the maximum number of supersteps taken by any processor during the execution of Step~3 of \aref{Algorithm PS-MWM}{fig:algo_psmwm}.
\end{definition}

Note that for $k=1$, the effective iterations equals the number of edges in the stream. Using this notion, we align the supersteps of different processors and define the following directed graph.

\begin{definition}[$G^{\br{t}}$]
\label{def_ei_dag}
For the $t \mathrm{th}$ effective iteration, consider the set of edges processed across all $k$ streams. 
Let $e_{\ell}=\br{u_{\ell},v_{\ell}}$ denote the edge processed in the $\ell$th stream, where $u_{\ell}$ precedes $v_{\ell}$ in the lexicographic ordering of the vertices. If processor $\ell$ is idle in the $t\mathrm{th}$ iteration, then $e_{\ell} = \emptyset$.
Define $G^{\br{t}} := \br{V^{\br{t}},E^{\br{t}}}$, where \[E^{\br{t}} := \{e_{\ell} \mid \ell \in [k]\} \text{ and } V^{\br{t}} := \underset{\br{u_{\ell},v_{\ell}} \in E^{\br{t}}}{\bigcup} \{u_{\ell}, v_{\ell}\}.\]
\end{definition}

The following property of $G^{\br{t}}$ is straightforward to verify.
\begin{proposition}
\label{prop_ei_dag}
$G^{\br{t}}$ is a directed acyclic graph.
\end{proposition}

We show that \aref{Algorithm PS-MWM}{fig:algo_psmwm} is free from deadlock, livelock, and starvation. \emph{Deadlock} occurs when a set of processors forms a cyclic dependency, with each processor waiting for a resource held by another. \emph{Livelock} occurs when a set of processors repeatedly form such a cycle, where each processor continually acquires and releases resources without making progress. \emph{Starvation} occurs when a processor waits indefinitely for a resource because other processors repeatedly acquire it first. The following lemma shows that the streaming phase of \aref{PS-MWM}{fig:algo_psmwm} is free from deadlock, livelock, and starvation.

\begin{lemma}
\label{lemma_free}
The concurrent executions of the subroutine \aref{Process-Edge}{fig:algo_proc_edge} is free from deadlock, livelock, and starvation.
\end{lemma}
\begin{proof}
Since the variables $\{\alpha_u\}_{u \in V}$ are updated only while holding their corresponding locks, we treat the locks $\{lock_u\}_{u \in V}$ as the only shared resources in \aref{Process-Edge}{fig:algo_proc_edge}. 

Let $G^{\br{t}}$ be the graph defined in \hyref{Definition}{def_ei_dag}. By \hyref{Proposition}{prop_ei_dag}, $G^{\br{t}}$ is a directed acyclic graph (DAG), and hence each of its components is also a DAG. 

To reason about cyclic dependencies, we focus on components of $G^{\br{t}}$ involving processors executing Step~2 of \aref{Process-Edge}{fig:algo_proc_edge}. Every DAG contains at least one vertex with no outgoing edges. Thus, each such component includes an edge $e_{\ell} =\br{u_{\ell}, v_{\ell}}$ such that only processor $\ell$ requests $lock_{v_{\ell}}$. This precludes the possibility of cyclic dependencies; that is, the concurrent executions of \aref{Process-Edge}{fig:algo_proc_edge} is free from deadlock and livelock.

To show that starvation does not occur, suppose an edge appears in every effective iteration $t \in [a,b]$, that is, $e_{\ell} = \br{u_{\ell}, v_{\ell}} \in \underset{t \in [a,b]}{\cap} E^{\br{t}}$. We show that $b-a = \bigOtilde{n / \epsilon}$, which bounds the number of supersteps that processor $\ell$ may spend attempting to acquire locks for $e_{\ell}$. 

Step~2 requires one superstep per iteration, while all other steps collectively require at most one. For each $t \in (a,b]$, the component of $G^{\br{t-1}}$ containing $e_{\ell}$ has at least one vertex with no outgoing edge. This guarantees that at least one edge in that component acquires its locks and completes Step~3 during the $\br{t-1}$th effective iteration. Since Step~3 can increment the values in $\{\alpha_u\}_{u \in V}$ for at most $\bigO{n\log_{1+\epsilon} W} = \bigOtilde{n / \epsilon}$ edges over the entire execution, the number of iterations for which $e_{\ell}$ may remain blocked is also bounded by $\bigOtilde{n/ \epsilon}$.
\end{proof}

To analyze Step~5 of \aref{Algorithm PS-MWM}{fig:algo_psmwm}, we adopt the same simplification: processors are assumed to operate in a quasi-synchronous manner. Accordingly, we define $\U^{\br{t}}$ as the set of edges present in the stacks $\underset{\ell \in [k]}{\bigcup} S^{\ell}$ at the beginning of iteration $t$ of Step~2 in \hyperref[fig:algo_proc_stack]{Process-Stack}. The following definition is useful for characterizing tight edges via an equivalent notion.

\begin{definition}[Follower]
\label{def_follower}
An edge $e_j \in \U^{\br{t}}$ is a \emph{follower} of an edge $e_i \in \U^{\br{t}}$ if $e_i\cap e_j \ne \emptyset$ and $e_j$ is added to some stack $S^{j}$ after $e_i$ is added to some stack $S^i$. We denote the set of followers of an edge $e$ by $\F(e)$.
\end{definition}

The proofs of the following four lemmas are included in \hyref{Appendix}{app:lemma_free_pp}. The fourth lemma establishes that the post-processing phase of \aref{PS-MWM}{fig:algo_psmwm} is free from deadlock, livelock, and starvation.

\begin{restatable}{lemma}{LemmaTightFollower}
\label{lemma_tight_follower}
An edge $e$ is a tight edge if and only if $\F(e) = \emptyset$.
\end{restatable}

\begin{restatable}{lemma}{LemmaTightEdge}
\label{lemma_tight_edge}
Let $\T^{\br{t}}$ be the set of top edges in the stacks at the beginning of iteration $t$ of Step~2 of \hyperref[fig:algo_proc_stack]{Process-Stack}. Then $\T^{\br{t}}$ contains at least one tight edge.
\end{restatable}

\begin{restatable}{lemma}{LemmaVertexDisjoint}
\label{lemma_vertex_disjoint}
The set of tight edges in $\U^{\br{t}}$ is vertex-disjoint.
\end{restatable}

\begin{restatable}{lemma}{LemmaFreePP}
\label{lemma_free_pp}
The concurrent executions of the subroutine \aref{Process-Stack}{fig:algo_proc_stack} is free from deadlock, livelock, and starvation.
\end{restatable}

We now analyze the performance metrics of the algorithm.

\begin{lemma}
\label{lemma_algo1_space_ns} 
For any constant $\epsilon > 0$, the space complexity and per-edge processing time of \aref{Algorithm PS-MWM}{fig:algo_psmwm} are $\bigO{k+n\log n}$ and $\bigO{n \log n}$, respectively. Furthermore, for $\lmin=\bigOmega{n}$, the amortized per-edge processing time of the algorithm is $\bigO{\log n}$.
\end{lemma}
\begin{proof}
The claimed space bound follows from three components: $\bigO{n}$ space for the variables and locks, $\bigO{n \log n}$ space for the stacked edges, and $\bigO{1}$ space per processor.

The worst-case per-edge processing time follows from the second part of the proof of \lemref{lemma_free}.

Processor $\ell$ processes $|E^{\ell}|$ edges, each requiring at least one distinct effective iteration (see \hyref{Definition}{def_ei}). 
Additional iterations may arise when it repeatedly attempts to acquire locks in Step~2 of \aref{Process-Edge}{fig:algo_proc_edge}. From the second part of the proof of \lemref{lemma_free}, the total number of such additional iterations is bounded by $\bigO{ n \log n}$. This implies that to process $|E^{\ell}|$ edges, a processor $\ell$ uses $\bigO{|E^{\ell}| + n\log n}$ supersteps. Therefore, the amortized per-edge processing time is $$\bigO{\frac{|E^{\ell}|+n \log n}{|E^{\ell}|}} = \bigO{\frac{n \log n}{|E^{\ell}|}} = \bigO{\frac{n \log n}{\lmin}} = \bigO{\log n}.$$
\end{proof}

Note that the amortized per-edge processing time is computed over the edges of an individual stream, not over the total number of edges across all streams.
While both forms of amortization are meaningful for poly-streaming algorithms, our analysis is more practically relevant, as it reflects the cost incurred per edge arrival within a single stream.

\begin{lemma}
\label{lemma_algo1_time_ns}
For any constant $\epsilon> 0$, \aref{Algorithm PS-MWM}{fig:algo_psmwm} takes 
$\bigO{\lmax+n\log n}$ time.
\end{lemma}
\begin{proof}
The preprocessing phase (Steps~1--2) takes $\bigTheta{n/k}$ time.

To process $|E^{\ell}|$ edges, processor $\ell$ takes $\bigO{|E^{\ell}| + n \log n}$ supersteps (see the proof of \lemref{lemma_algo1_space_ns}). Since $|E^{\ell}| \leq \lmax$ for all $\ell \in [k]$, the time required for Step~3 is $\bigO{\lmax + n\log n}$.

At the beginning of Step~5, the total number of edges in the stacks is $\U^{\br{1}} = \bigO{n \log n}$. By \lemref{lemma_tight_edge}, iteration $t$ of \hyperref[fig:algo_proc_stack]{Process-Stack} removes at least one edge from $\U^{\br{t}}$. Hence, the time required for Step~5 is $\bigO{n\log n}$. 

The claim now follows by summing the time spent across all three phases.
\end{proof}

Now, using the characterizations of tight edges, we extend the duality-based analysis of~\cite{ghaffari2019simplified} to our algorithm. Let $\Delta_{\alpha}^e$ denote the change in $\sum_{u \in V} \alpha_u$ resulting from processing an edge $e \in E^{\ell}$ in Step~3 of \hyperref[fig:algo_proc_edge]{Process-Edge}. If an edge $e \in E^{\ell}$ is not included in a stack $S^{\ell}$ then $\Delta_\alpha^e = 0$, either because it fails the condition in Step~1 or Step~3 of \hyperref[fig:algo_proc_edge]{Process-Edge}. It follows that $\sum_{e \in \bigcup_{\ell \in [k]} E^{\ell}} \Delta_{\alpha}^e = \sum_{u\in V} \alpha_u$. For an edge $e$ that is included in some stack $S^i$, let $\G(e)$ denote the set of edges that share an endpoint with $e$ and are included in some stack $S^j$ no later than $e$ (including $e$ itself). The following two results are immediate from Observation 3.2 and Lemma 3.4 of~\cite{ghaffari2019simplified}.

\begin{proposition}
\label{prop_algo1_aux}
Any edge $e$ added to some stack $S^{\ell}$ satisfies the inequality $$ w_e \geq \sum_{e^{\prime} \in \G(e)} g_{e^{\prime}} = \frac{1}{2}\left(\sum_{e^{\prime} \in \G(e)} \Delta_{\alpha}^{e^{\prime}}\right).$$    
\end{proposition}

\begin{proposition}
\label{prop_algo1_dual}
After all processors complete Step~3 of \aref{Algorithm PS-MWM}{fig:algo_psmwm}, the variables $\{\alpha_u\}_{u \in V}$, scaled by a factor of $(1+\epsilon)$, form a feasible solution to the dual LP in \fref{fig:lp}.
\end{proposition}

\begin{lemma}
\label{lemma_algo1_approx}
Let $\M^{*}$ be a maximum weight matching in $G$. The matching $\M:= \underset{\ell \in [k]}{\bigcup} \M^{\ell}$ returned by \aref{Algorithm PS-MWM}{fig:algo_psmwm} satisfies $w(\M) \geq \frac{1}{2\left(1+\epsilon\right)} w(\M^{*})$.
\end{lemma}
\begin{proof}
We only process tight edges in \hyperref[fig:algo_proc_stack]{Process-Stack}. By \lemref{lemma_vertex_disjoint} tight edges are vertex disjoint, and hence their independent processing does not interfere with their inclusion in $\M$.

By \lemref{lemma_tight_follower}, an edge $e$ included in $\M$ must satisfy $\F(e) = \emptyset$. Consider any edge $e^{\prime} \in \G(e) \backslash \{e\}$. Since $e \in \F(e^{\prime})$, we have $\F(e^{\prime}) \ne \emptyset$, which means $e^{\prime}$ is not a tight edge before $e$ is processed. 

Thus, when $e$ is selected for inclusion in $\M$, none of the edges in $\G(e) \backslash \{e\}$ is tight. Hence, all edges of $\G(e)$ are in the stacks when we are about to process $e$. Therefore, the total gain contributed by edges in $\G(e)$ can be attributed to the weight of $e$, and by \hyref{Proposition}{prop_algo1_aux}, we have
\begin{equation*}
\begin{split}
w(\M) = \sum_{e  \in \M} w_e &\geq \frac{1}{2}\left(\sum_{e \in \M} \sum_{e^{\prime} \in \G(e)} \Delta_{\alpha}^{e^{\prime}}\right) \geq \frac{1}{2}\left(\sum_{e \in \bigcup_{\ell \in [k]} S^{\ell}} \Delta_{\alpha}^{e}\right) \\
    &= \frac{1}{2}\left(\sum_{e \in \bigcup_{\ell \in [k]} E^{\ell}} \Delta_{\alpha}^{e}\right) = \frac{1}{2} \left(\sum_{u \in V}\alpha_u\right).
\end{split}
\end{equation*}
Let $\{x_e^{*}\}_{e \in E}$ be an optimal solution to the primal LP in \fref{fig:lp}. By \hyref{Proposition}{prop_algo1_dual} and the LP duality we have
$$w(\M^{*}) \leq \sum_{e \in E}w_ex_e^{*} \leq (1+\epsilon) \left( \sum_{u \in V} \alpha_u\right) \leq 2(1+\epsilon)w(\M).$$
\end{proof}

It is straightforward to see that both strategies use only one pass over the streams (Step~4 of \hyperref[fig:algo_psmwm]{PS-MWM} and Step~5 of \hyperref[fig:algo_psmwm_ds]{PS-MWM-DS}). \hyref{Theorem}{thm_psmwm} now follows by combining the results in \lemref{lemma_algo1_space_ns}, \lemref{lemma_algo1_time_ns}, \lemref{lemma_algo1_approx},  \lemref{lemma_algo1_space_ds}, and the analysis of the \emph{deferrable strategy} sketched in \hyref{Appendix}{subsec:mwm_ds}.

%% file: 4_numa.tex
\section{Algorithms for Non-Uniform Memory Access Costs}
\label{sec:numa}

In this section, we extend the algorithm from \hyref{Section}{sec:mwm} to account for the non-uniform memory access (NUMA) costs present in real-world machines. 

In a poly-streaming algorithm, each processor may receive an arbitrary subset of the input, making it difficult to maintain memory access locality. 
Modern shared-memory machines, as illustrated in \fref{fig:model_shared}, have non-uniform memory access costs and far fewer memory controllers than processors~\cite{Milan}.
As a result, memory systems with such limitations would struggle to handle the high volume of concurrent, random memory access requests generated by poly-streaming algorithms, leading to significant delays.

\begin{figure}[h]
    \centering
    \begin{parbox}{4.8in}{    
        \begin{mdframed}[linewidth=0.5pt, roundcorner=7pt, backgroundcolor=gray!5, frametitle={\underline{PS-MWM-LD$(V, \ell, j, \epsilon)$}}]
        \begin{enumerate}
            \item In parallel initialize $lock_u$, and set $\alpha_u$ and $mark_u$ to $0$ for all $u \in V$ \textcolor{teal}{\\/* processor $\ell$ initializes or sets $\bigTheta{n/k}$ locks/variables */}    
            \item In parallel initialize $lock_u^j$, and set $\alpha_u^j$ to $0$ for all $u \in V$ \textcolor{teal}{\\/* processor $\ell$ initializes or sets $\Theta\br{n/\br{k/r}}$ locks / variables */}
            \item In parallel initialize $glock^j$ \textcolor{teal}{/* one processor initializes for group $j$ */}
            \item  $S^{\ell} \leftarrow \emptyset$ \textcolor{teal}{/* initialize an empty stack */}
            \item for each edge $e=\{u,v\}$ in $\ell${th} stream do
            \begin{enumerate}
                \item \aref{Process-Edge-LD$(e, S^{\ell}, \epsilon$)}{fig:algo_proc_edge_ld}
            \end{enumerate}
            \item wait for all processors to complete execution of Step~4 \textcolor{teal}{/* a barrier */}
            \item $\M^{\ell} \leftarrow$ \aref{Process-Stack$(S^{\ell})$}{fig:algo_proc_stack}
            \item return $\M^{\ell}$
        \end{enumerate}
        \end{mdframed}    
    
    }
    \end{parbox}    
    \caption{A generalization of \aref{Algorithm PS-MWM}{fig:algo_psmwm} 
 using local dual variables.}
    \label{fig:algo_psmwm_ld}
\end{figure}

We now describe a generalization of the algorithm from \hyref{Section}{sec:mwm} that localizes a significant portion of each processor's memory access to its near memory. This generalization applies to both edge-processing strategies introduced in \hyref{Section}{subsec:mwm_ps}. We focus on the \emph{non-deferrable strategy}. (The deferrable strategy generalizes in the same way, following the same relationship between the two strategies as in the specialized case.)

The runtime of \aref{Process-Edge}{fig:algo_proc_edge} is dominated by accesses to the dual variables $\{\alpha_u\}_{u \in V}$. By assigning a dedicated stack to each processor, we have substantially localized accesses associated with edges in that stack. However, since a large fraction of edges is typically not included in the stacks, the runtime remains dominated by accesses to dual variables associated with these discarded edges. We therefore describe an algorithm that localizes these accesses to memory near the processor.

To localize accesses to the dual variables $\{\alpha_u\}_{u \in V}$, we observe that these variables increase monotonically during the streaming phase. 
This observation motivates a design in which a subset of processors maintains local copies of the variables and can discard a substantial number of edges without synchronizing with the global copy.
When a processor includes an edge in its stack, it increments the corresponding dual variables in the global copy by the gain of the edge and synchronizes its local copy accordingly.
As a result, some local copies may lag behind the global copy, but they can be synchronized when needed.

A general scheme for allocating dual variables is as follows. The set of $k$ processors is partitioned into $r$ groups.
For simplicity, we assume that $k$ is a multiple of $r$, so each group contains exactly $k/r$ processors. 
For $r>1$, in addition to a global copy of dual variables, we maintain $r$ \emph{local copies} $\{\alpha_u^j\}_{u \in V}$, one for each $j \in [r]$. 
Group $j$ consists of the processors $\{\ell \in [k] \mid \lfloor \ell / (k/r) \rfloor = j\}$, and uses $\{\alpha_u^j\}_{u \in V}$ as its local copy of the dual variables.
\aref{Algorithm PS-MWM}{fig:algo_psmwm} corresponds to the special case $r = 1$, where all processors operate using only the global copy of the dual variables.

\aref{Algorithm PS-MWM-LD}{fig:algo_psmwm_ld}, along with its subroutine \hyperref[fig:algo_proc_edge_ld]{Process-Edge-LD}, incorporates local dual variables in addition to the global ones.
In Step~2, processors in each group $j \in [r]$ collectively initialize their group's local copies of dual variables and locks, followed by initializing a group lock in Step~3.
All other steps of the algorithm are identical to those in \hyperref[fig:algo_psmwm]{PS-MWM}.

\begin{figure}
    \centering
    \begin{parbox}{5.4in}{    
        \begin{mdframed}[linewidth=0.5pt, roundcorner=7pt, backgroundcolor=gray!5, frametitle={\underline{Process-Edge-LD$(e=\{u,v\}, S^{\ell}, \epsilon)$}}]
        \textcolor{teal}{/* Assumes access to $\{\alpha_x\}_{x \in V}$, $\{lock_x\}_{x \in V}$, $\{\alpha_x^j\}_{x \in V}$, $\{lock_x^j\}_{x \in V}$, and $glock^j$ */}
    \begin{enumerate}
        \item if $w_e \leq (1+\epsilon)(\alpha_u^j +\alpha_v^j)$ then return
        \item repeatedly try to acquire $lock_u^j$ and $lock_v^j$ in lexicographic order of $u$ and $v$ as long as $w_e > (1+\epsilon)(\alpha_u^j +\alpha_v^j)$        
        \item if $w_e \leq (1+\epsilon)(\alpha_u^j +\alpha_v^j)$ then release $lock_u^j$ and $lock_v^j$, and return      
        \item repeatedly try to acquire $glock^j$
        \item \aref{Process-Edge$(e, S^{\ell}, \epsilon$)}{fig:algo_proc_edge}
        \item $\alpha_u^j \leftarrow \alpha_u$ and $\alpha_v^j \leftarrow \alpha_v$ \textcolor{teal}{/* synchronization of local and global dual variables */}
        \item release $lock_u^j$, $lock_v^j$, $glock^j$ and return
    \end{enumerate}
        \end{mdframed}    
 
    }
    \end{parbox}    
    \caption{A subroutine used in \aref{Algorithm PS-MWM-LD}{fig:algo_psmwm_ld}.}
    \label{fig:algo_proc_edge_ld}
\end{figure}

In the subroutine \hyperref[fig:algo_proc_edge_ld]{Process-Edge-LD}, Step~5 implements the non-deferrable strategy. 
Steps~1--3 and Step~6 enforce the localization of access to dual variables.
Steps~2--3 ensure that, at any given time, each global dual variable is accessed by at most one processor per group; we refer to this processor as the \emph{delegate} of the group for that variable.
Thus, a processor executing Steps~4--6 serves as a delegate of its group for the corresponding dual variables during that time. 
In Step~6, after completing updates to the global variables, the delegate synchronizes its group's local copy in $\bigO{1}$ time.
As a result, the waiting time on a local variable in Step~2 is bounded by the total time spent by the corresponding delegates, up to constant factors.

The delegates in each group handle vertex-disjoint edges, so concurrent executions of Step~6 would have been safe. However, the lock in Step~4 ensures that at most one delegate per group executes Step~5 of \hyperref[fig:algo_proc_edge]{Process-Edge}. Regardless of these design choices, the behavior of delegates executing Step~5 concurrently mirrors that of processors competing for exclusive access to global dual variables in \hyperref[fig:algo_psmwm]{PS-MWM}.

The following lemma highlights the benefit of using \aref{Algorithm PS-MWM-LD}{fig:algo_psmwm_ld}; a proof is included in \hyref{Appendix}{subsec:thm_hd}.

\begin{restatable}{lemma}{LemmaLD}
\label{lemma_psmwm_ld}
For any constant $\epsilon > 0$, in the streaming phase of \aref{Algorithm PS-MWM-LD}{fig:algo_psmwm_ld}, processors in all $r$ groups collectively access global variables a total of $\bigO{r\cdot n \log n }$ times.
\end{restatable}

In contrast to the bound in \lemref{lemma_psmwm_ld}, the streaming phase of \aref{Algorithm PS-MWM}{fig:algo_psmwm} accesses global variables $\bigOmega{m}$ times or up to $\bigO{m +  k\cdot n \log n}$ times. 

\aref{Algorithm PS-MWM-LD}{fig:algo_psmwm_ld}, together with the generalization of the deferrable strategy, leads to the following result (a proof is included in \hyref{Appendix}{subsec:thm_hd}).

\begin{theorem}
\label{thm_psmwm_ld}
Let $k$ processors be partitioned into $r$ groups, each with its own shared local memory. 

For any constant $\epsilon > 0$, there exists a \emph{single-pass} poly-streaming algorithm for the maximum weight matching problem that achieves a $\br{2+\epsilon}$-approximation. It admits a CREW PRAM implementation with runtime $\bigOtilde{\lmax + n}$. 

If $\lmin = \Omega\br{n}$, the algorithm achieves $\bigO{\log n}$ amortized per-edge processing time using $\bigOtilde{k + r\cdot n}$ space. For arbitrarily balanced streams, it uses either:
\begin{itemize}
    \item $\bigOtilde{k+r\cdot n}$ space and $\bigOtilde{n}$ per-edge processing time, or
    \item $\bigOtilde{k\cdot n}$ space and $\bigO{1}$ per-edge processing time.
\end{itemize}
The processors collectively access the global memory $\bigOtilde{r\cdot n}$ times.
\end{theorem}

%% file: 5_evals.tex
\section{Empirical Evaluation}
\label{sec:evals}

This section summarizes our evaluation results for \aref{Algorithm PS-MWM-LD}{fig:algo_psmwm_ld}. 
Detailed datasets, experimental setup, and additional comparisons (including with \aref{PS-MWM}{fig:algo_psmwm}) are included in \hyref{Appendix}{sec:app_evals}. Our code will be made available at \url{https://github.com/ahammed-ullah/algodyssey}.

\subsection{Datasets}

\begin{table}[h]
    \centering
    \caption{Summary of datasets. Each collection contains eight graphs (details are included in \hyref{Appendix}{subsec:app_datasets}).}
    \begin{tabular}{ll}
        \toprule        
        Graph Collection & \# of Edges (in billions) \\
        \midrule
        The \emph{SSW graphs} & $1.36-127.4$ \\
        The \emph{BA graphs} & $4.64-550.1$ \\
        The \emph{ER graphs} & $256-4096$ \\
        The \emph{UA-dv graphs} & $275.2-550.1$ \\
        The \emph{UA graphs} & $8.93-1100$ \\
        The \emph{ER-dv graphs} & $32-4096$ \\        
        \bottomrule
    \end{tabular}
    \label{tab:datasets_summary}
\end{table}

Table~\ref{tab:datasets_summary} summarizes our datasets. Each collection consists of eight graphs, with edge counts ranging from one billion to four trillion.
To the best of our knowledge, these represent some of the largest graphs for which matchings have been reported in the literature.
Exact and approximate offline MWM algorithms (see~\cite{pothen2019approximation}) would exceed available memory on the larger graphs.
The first class (SSW) consists of six of the largest graphs from the SuiteSparse Matrix collection~\cite{davis2011university} and two from the Web Data Commons~\cite{meusel2015graph}, which includes the largest publicly available graph dataset. Other classes include synthetic graphs
generated from the Barab{\'a}si--Albert (BA), Uniform Attachment (UA), and \Erdos--\Renyi \ (ER) models~\cite{albert2002statistical, erdos1960evolution, pekoz2013total}.

\subsection{Experimental Setup}
We ran all experiments on a community cluster called Negishi~\cite{McCartney2014}, where each node has an AMD Milan processor with 128 cores running at 2.2 GHz, 256--1024~GB of memory, and the Rocky Linux 8 operating system version~8.8. 
The cores are organized in a hierarchy: groups of eight cores constitute a core complex that share an L3 cache.
Eight core complexes form a socket, and they share four dual-channel memory controllers; two sockets constitute a Milan node~\cite{Milan}. Memory access within a socket is approximately three times faster than across sockets. 

We implemented the algorithms in \texttt{C++} and compiled the code using the \textit{gcc} compiler (version 12.2.0) with the \texttt{-O3} optimization flag. For shared-memory parallelism, we used the OpenMP library (version $4.5$). All experiments used $\epsilon = 1e-6$. Reported values are the average over five runs. \hyref{Appendix}{subsec:app_setup} contains additional details of the experimental setup, including the generation of edge streams.

\subsection{Space}
\label{subsec:evals_space}

\fref{fig:space} summarizes the space usage of our algorithm. For $k=1$, the algorithm of Paz and Schwartzman~\cite{paz20182+}, we store one copy of the dual variables, stack, and matching. For $k>1$, our algorithm stores $r+1$ copies of the dual variables (global and local), stacks, matching, and locks. We choose the values of $r$ based on the system architecture and the number of streams (see \hyref{Appendix}{subsec:app_space} for details).

\input{Tikz/space/data}

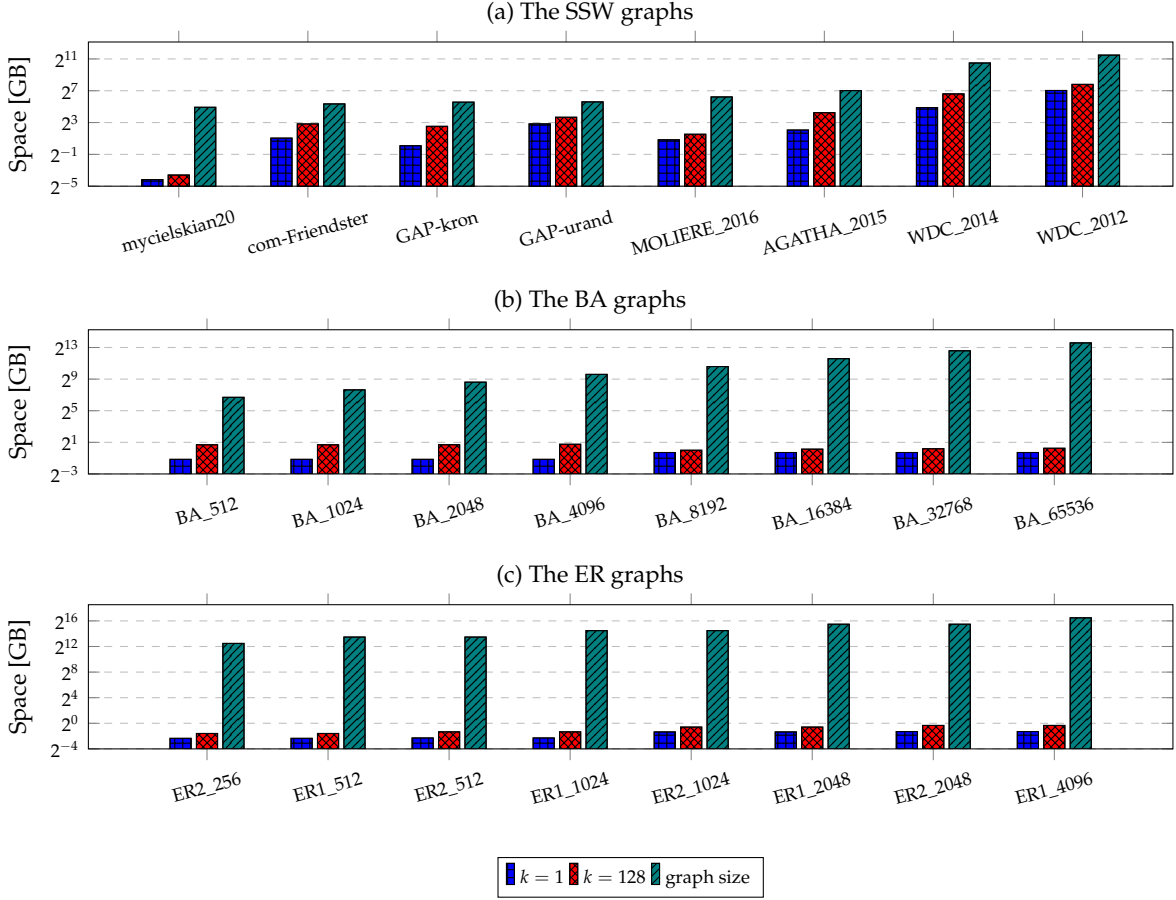
\begin{figure}[h]
    \centering
    \begin{subfigure}[t]{1.0\linewidth}
    \centering
    \caption{The SSW graphs}
    \input{Tikz/space/space_ssw}
    \end{subfigure}
    \begin{subfigure}[t]{1.0\linewidth}
    \centering
    \caption{{The BA graphs}}
    \input{Tikz/space/space_ba}
    \end{subfigure}
    \begin{subfigure}[t]{1.0\linewidth}
    \centering
    \caption{The ER graphs}
    \input{Tikz/space/space_er}
    \end{subfigure}
    \caption{Memory used by the algorithm and the corresponding \emph{graph size} (space needed to store the entire graph in CSR format). Note that the $y$-axes are in a logarithmic scale.}
    \label{fig:space}
\end{figure}

The maximum space used by our algorithm is $223\,\mathrm{GB}$, for the web graph WDC\_2012.
In comparison, storing this graph in compressed sparse row (CSR) format would require over $2800\,\mathrm{GB}$. 
Storing the largest graph in our datasets (ER1\_4096) in CSR would require more than $91,600\,\mathrm{GB}$ ($89.45\,\mathrm{TB}$), for which our algorithm used less than $0.8\,\mathrm{GB}$.

\subsection{Solution Quality}
\label{subsec:evals_quality}

\textbf{min-OPT percent.}  In \hyref{Appendix}{subsec:app_duals}, we describe different ways to get \emph{a posteriori} upper bounds on the weight of a MWM $w\br{M^{*}}$, using the values of the dual variables.  
Let $Y_{min}$ denote the minimum value of these upper bounds. If $\M$ is a matching in the graph returned by any algorithm, then we have $\frac{w\br{\M}}{w\br{\M^{*}}} \geq \frac{w\br{\M}}{Y_{min}}$. Hence, $\frac{w(\M)}{Y_{min}} \times 100$ gives a lower bound on the percentage of the maximum weight $w\br{\M^{*}}$ obtained by $\M$. We use \emph{min-OPT percent} to denote the fraction $\frac{w(\M)}{Y_{min}} \times 100$.

\input{Tikz/quality/data}

\begin{figure}[h]
    \begin{subfigure}[t]{1.0\linewidth}
    \centering    
    \caption{The SSW graphs}
    \input{Tikz/quality/sq_ssw}
    \end{subfigure}    
    \begin{subfigure}[t]{1.0\linewidth}
    \centering
    \caption{The BA graphs}
    \input{Tikz/quality/sq_ba}
    \end{subfigure}
    \begin{subfigure}[t]{1.0\linewidth}
    \centering
    \caption{The ER graphs}
    \input{Tikz/quality/sq_er}
    \end{subfigure}
    \caption{Comparisons of \emph{min-OPT percent} obtained by different algorithms. \emph{ALG-d} denotes the best results from four dual update rules described in \hyref{Appendix}{subsec:app_duals}, and \emph{ALG-s} denotes the algorithm of Feigenbaum et al.~\cite{feigenbaum2005graph}.}
    \label{fig:sq}
\end{figure}
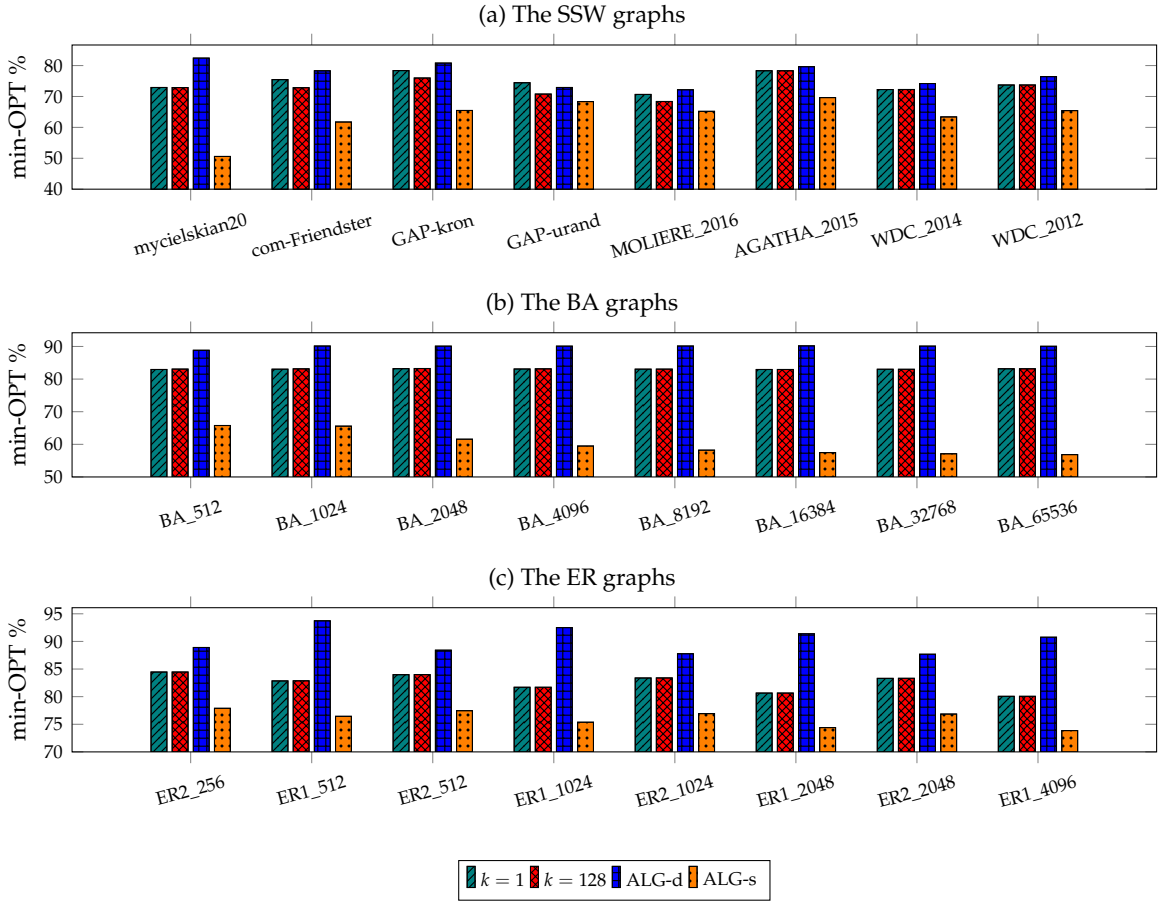

\fref{fig:sq} shows min-OPT percent obtained by different algorithms. In \hyref{Appendix}{subsec:app_duals}, we describe four dual update rules as alternatives to the default rule used in Steps~3(a)--(b) of \hyperref[fig:algo_proc_edge]{Process-Edge}. The values under $k=1$ and $k=128$ use the default rule, and the values under \emph{ALG-d} use the best result among the four new dual update rules.
For perspective, we include min-OPT percent obtained by the sequential 6-approximate streaming algorithm of Feigenbaum et al.~\cite{feigenbaum2005graph}, denoted \emph{ALG-s}.

The results under $k=1$ and $k=128$ show that, in terms of solution quality, our poly-streaming algorithm is on par with the single-stream algorithm of~\cite{paz20182+}. The values under \emph{ALG-d} indicate further potential improvements using alternative dual update rules. The comparison with \emph{ALG-s} supports our choice of the algorithm from~\cite{paz20182+} over other simple algorithms, such as that of~\cite{feigenbaum2005graph}. \hyref{Appendix}{subsec:app_quality} contains comparisons with an offline algorithm and details on the dual update rules.

\subsection{Runtime}
\label{subsec:evals_runtime}

We report runtime-based speedups, computed as the total time across all three phases of \aref{PS-MWM-LD}{fig:algo_psmwm_ld}
(preprocessing, streaming, and post-processing). \fref{fig:speedup} shows these speedups. 
For $k=128$, we have speedups of 16--60, 37--73, and 68--83 for the SSW graphs, the BA graphs, and the ER graphs, respectively. 

\input{Tikz/time/data}

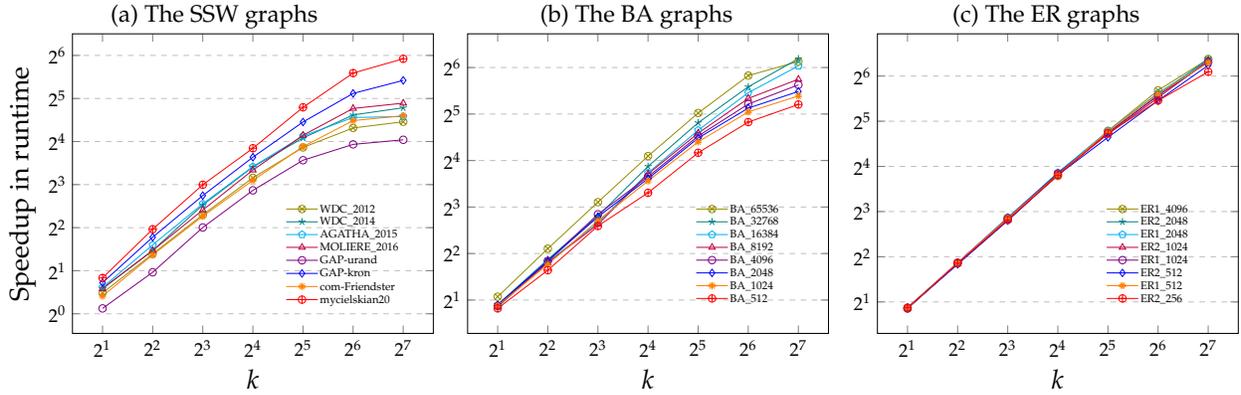
\begin{figure}[h]
\centering
\begin{subfigure}[t]{0.34\textwidth}
    \centering
    \caption{The SSW graphs}
    \input{Tikz/time/sp_ssw}    
\end{subfigure}
\hfill
\begin{subfigure}[t]{0.32\textwidth}
    \centering
    \caption{The BA graphs}
    \input{Tikz/time/sp_ba}    
\end{subfigure}
\hfill
\begin{subfigure}[t]{0.32\textwidth}
    \centering
    \caption{The ER graphs}
    \input{Tikz/time/sp_er}    
\end{subfigure}
\caption{Speedup in runtime vs.\ $k$. Note that both axes are on a logarithmic scale.}
\label{fig:speedup}
\end{figure}

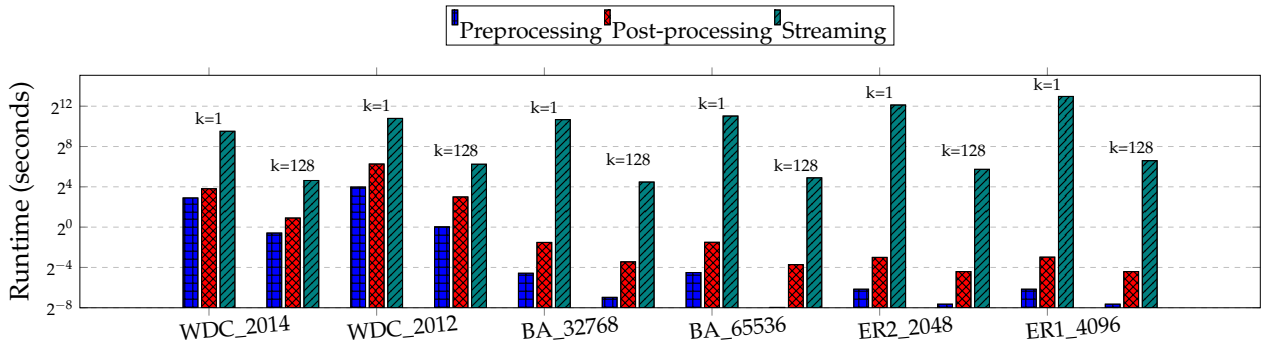
\begin{figure}[h]
\centering
\centering
\input{Tikz/time/breakdown}
\caption{Breakdown of runtime into three phases for $k=1$ and $k=128$. Note that the $y$-axis is in a logarithmic scale.}
\label{fig:br}
\end{figure}

Due to the significant memory bottlenecks (discussed in \hyref{Section}{sec:numa}), we also report speedups w.r.t. \emph{effective iterations} (\hyref{Definition}{def_ei}), which are less affected by such bottlenecks. The speedup w.r.t. effective iterations is the ratio of the metric for one stream to that for $k$ streams.
Now for $k=128$, we obtain speedups of 112--127, 121--127, and 124--128 for the SSW graphs, the BA graphs, and the ER graphs, respectively. 
These results indicate that shared variable access incurs no noticeable contention among processors. As a result, we expect even better runtime improvements on systems with more memory controllers or better support for remote memory access.

\fref{fig:br} shows the runtimes for different graphs, decomposed into three phases, for $k=1$ and $k=128$. The plots report the absolute time savings achieved by processing multiple streams concurrently. For $k=1$ and $k=128$, the geometric means of the runtimes for these graphs are over $2350$ seconds and under $45$ seconds, respectively. For the largest graph (ER1\_4096), single-stream processing took over $8000$ seconds, whereas poly-stream processing reduced the time to under $100$ seconds.

%% file: Tikz/space/data.tex
\pgfplotstableread{
Graph    k1     k128    gsize
mycielskian20	5.54E-02	8.26E-02	3.03E+01
com-Friendster	2.06E+00	7.10E+00	4.06E+01
GAP-kron	1.06E+00	5.76E+00	4.74E+01
GAP-urand	7.13E+00	1.27E+01	4.85E+01
MOLIERE\_2016	1.78E+00	2.89E+00	7.47E+01
AGATHA\_2015	4.17E+00	1.89E+01	1.30E+02
WDC\_2014	2.95E+01	9.73E+01	1.44E+03
WDC\_2012	1.31E+02	2.23E+02	2.86E+03
}\spacessw

\pgfplotstableread{
Graph    k1     k128    gsize
BA\_512	0.45	1.61	103.71
BA\_1024	0.45	1.61	199.71
BA\_2048	0.45	1.62	391.71
BA\_4096	0.45	1.69	775.71
BA\_8192	0.82	1.00	1543.71
BA\_16384	0.82	1.10	3079.71
BA\_32768	0.82	1.14	6151.71
BA\_65536	0.82	1.19	12295.71
}\spaceba

\pgfplotstableread{
Graph    k1     k128    gsize
ER2\_256	1.96E-01	3.32E-01	5.72E+03
ER1\_512	1.96E-01	3.32E-01	1.14E+04
ER2\_512	2.04E-01	3.96E-01	1.14E+04
ER1\_1024	2.04E-01	3.96E-01	2.29E+04
ER2\_1024	3.92E-01	6.64E-01	2.29E+04
ER1\_2048	3.92E-01	6.64E-01	4.58E+04
ER2\_2048	4.09E-01	7.93E-01	4.58E+04
ER1\_4096	4.09E-01	7.93E-01	9.16E+04
}\spaceer

\pgfplotstableread{
Graph    k1     k128    gsize
UA\_512\_537	52.50390625	128.90625	6153.681357
UA\_1024\_537	52.50390625	131.9335938	12297.68136
UA\_1024\_268	26.25	65.86914063	6152.681357
UA\_2048\_268	26.25	71.86523438	12296.68136
UA\_2048\_134	13.12304688	32.25585938	6152.181357
UA\_4096\_134	13.12304688	36.53320313	12296.18136
UA\_4096\_67	6.5625	15.81054688	6151.931357
UA\_8192\_67	6.5625	18.26171875	12295.93136
}\spaceuadv

\pgfplotstableread{
Graph    k1     k128    gsize
UA\_1024	0.45	1.65	199.71
UA\_2048	0.45	1.70	391.71
UA\_4096	0.82	1.76	775.71
UA\_8192	0.82	1.91	1543.71
UA\_16384	0.88	2.04	3079.71
UA\_32768	0.88	2.17	6151.71
UA\_65536	1.63	2.36	12295.71
UA\_131072	1.63	2.57	24583.71
}\spaceua

\pgfplotstableread{
Graph    k1     k128    gsize
ER128\_32	0.41	0.79	715.27
ER64\_64	0.41	0.79	1430.52
ER32\_128	0.41	0.79	2861.03
ER16\_256	0.41	0.79	5722.06
ER8\_512	0.41	0.79	11444.09
ER4\_1024	0.41	0.79	22888.20
ER2\_2048	4.09E-01	7.93E-01	4.58E+04
ER1\_4096	4.09E-01	7.93E-01	9.16E+04
}\spaceerdv

%% file: Tikz/space/space_ssw.tex
\begin{tikzpicture}
\begin{semilogyaxis}[        
    ymin=0.03125,
    width=16cm,
    height=3.5cm,
    symbolic x coords={mycielskian20,com-Friendster,GAP-kron, GAP-urand, MOLIERE\_2016, AGATHA\_2015, WDC\_2014, WDC\_2012},
    ylabel={Space [GB]},
    ylabel style={font=\small},
    ybar, bar width=8pt,
    ymajorgrids=true,
    grid style=dashed,        
    log basis y={2},
    log origin=infty,
    point meta=explicit,
    x tick label style={rotate=15},
    x tick label style={font=\scriptsize},
    y tick label style={font=\scriptsize},
    ytick distance=16
    ]

\addplot [fill=blue, postaction={pattern=grid}] table [meta=k1] \spacessw;

\addplot [fill=red, postaction={pattern=crosshatch}] table [y=k128, meta=k128] \spacessw;

\addplot [fill=teal, postaction={pattern=north east lines}] table [y=gsize, meta=gsize] \spacessw;
\end{semilogyaxis}
\end{tikzpicture}

%% file: Tikz/space/space_ba.tex
\begin{tikzpicture}
\begin{axis}[ 
    ymin=0.125,
    width=16cm,
    height=3.5cm,    
    symbolic x coords={BA\_512, BA\_1024, BA\_2048, BA\_4096, BA\_8192, BA\_16384, BA\_32768, BA\_65536},
    ylabel={Space [GB]},
    ylabel style={font=\small},
    enlargelimits=0.14,
    enlarge y limits=upper,    
    ybar, bar width=8pt,
    ymajorgrids=true,
    grid style=dashed,
    ymode=log,
    log basis y={2},
    log origin=infty,
    point meta=explicit,
    x tick label style={rotate=15},
    x tick label style={font=\scriptsize},
    y tick label style={font=\scriptsize},
    ytick distance=16
    ]

\addplot [fill=blue, postaction={pattern=grid}] table [meta=k1] \spaceba;

\addplot [fill=red, postaction={pattern=crosshatch}] table [y=k128, meta=k128] \spaceba;

\addplot [fill=teal, postaction={pattern=north east lines}] table [y=gsize, meta=gsize] \spaceba;
\end{axis}
\end{tikzpicture}

%% file: Tikz/space/space_er.tex
\begin{tikzpicture}
\begin{axis}[     
    ymin=0.0625,
    width=16cm,
    height=3.5cm,
    symbolic x coords={ER2\_256, ER1\_512, ER2\_512, ER1\_1024, ER2\_1024, ER1\_2048, ER2\_2048, ER1\_4096},
    ylabel={Space [GB]},
    ylabel style={font=\small},
    enlargelimits=0.14,
    enlarge y limits=upper,    
    ybar, bar width=8pt,
    ymajorgrids=true,
    grid style=dashed,
    ymode=log,
    log basis y={2},
    log origin=infty,
    legend style={    
        font=\scriptsize,
        at={(0.5,-0.75)},
        anchor=north,
        legend columns=-1
        },
    ymode=log,
    log basis y={2},    
    point meta=explicit,
    x tick label style={rotate=15},
    x tick label style={font=\scriptsize},
    y tick label style={font=\scriptsize},
    ytick distance=16
    ]

\addplot [fill=blue, postaction={pattern=grid}] table [meta=k1] \spaceer;

\addplot [fill=red, postaction={pattern=crosshatch}] table [y=k128, meta=k128] \spaceer;

\addplot [fill=teal, postaction={pattern=north east lines}] table [y=gsize, meta=gsize] \spaceer;
\legend{$k=1$, $k=128$, graph size}
\end{axis}
\end{tikzpicture}

%% file: Tikz/quality/data.tex
\pgfplotstableread{
Graph    k1     k128    algdu    algss
mycielskian20	72.91	72.89	82.42	50.59
com-Friendster	75.46	72.84	78.32	61.74
GAP-kron	78.38	75.98	80.81	65.47
GAP-urand	74.45	70.78	72.91	68.35
MOLIERE\_2016	70.66	68.38	72.17	65.19
AGATHA\_2015	78.33	78.33	79.64	69.64
WDC\_2014	72.25	72.26	74.19	63.39
WDC\_2012	73.72	73.72	76.43	65.43
}\sqssw

\pgfplotstableread{
Graph    k1     k128    algdu    algss
BA\_512	82.94	83.06	88.83	65.77
BA\_1024	83.03	83.14	90.13	65.57
BA\_2048	83.19	83.24	90.09	61.57
BA\_4096	83.09	83.14	90.09	59.48
BA\_8192	83.04	83.05	90.13	58.21
BA\_16384	82.94	82.94	90.18	57.47
BA\_32768	83.01	83.00	90.10	57.11
BA\_65536	83.16	83.16	90.05	56.85
}\sqba

\pgfplotstableread{
Graph    k1     k128    algdu    algss
ER2\_256	84.46	84.46	88.90	77.89
ER1\_512	82.85	82.87	93.73	76.44
ER2\_512	83.98	83.99	88.41	77.46
ER1\_1024	81.70	81.70	92.50	75.37
ER2\_1024	83.38	83.39	87.77	76.93
ER1\_2048	80.65	80.65	91.38	74.40
ER2\_2048	83.31	83.32	87.68	76.85
ER1\_4096	80.06	80.06	90.76	73.86
}\sqer

\pgfplotstableread{
Graph    k1     k128    algdu    algss
UA\_4096\_67	86.47	86.45	90.51	76.35
UA\_8192\_67	89.30	89.28	93.03	77.48
UA\_2048\_134	86.37	86.35	90.37	76.24
UA\_4096\_134	89.21	89.20	92.90	77.41
UA\_1024\_268	86.24	86.23	90.17	76.10
UA\_2048\_268	89.09	89.08	92.70	77.33
UA\_512\_537	86.04	86.03	89.86	75.89
UA\_1024\_537	88.90	88.90	92.39	77.20
}\squadv

\pgfplotstableread{
Graph    k1     k128    algdu    algss
UA\_1024	69.12	69.02	71.23	67.88
UA\_2048	69.72	69.60	72.58	67.34
UA\_4096	74.23	74.09	77.79	70.17
UA\_8192	78.93	78.79	82.94	72.84
UA\_16384	83.22	83.08	87.32	75.00
UA\_32768	86.73	86.58	90.65	76.50
UA\_65536	89.39	89.25	93.01	77.46
UA\_131072	91.34	91.19	95.24	78.01
}\squa

\pgfplotstableread{
Graph    k1     k128    algdu    algss
ER128\_32	85.40	85.40	90.10	78.39
ER64\_64	85.39	85.39	90.14	78.47
ER32\_128	85.27	85.28	89.92	78.49
ER16\_256	85.17	85.17	89.71	78.48
ER8\_512	85.03	85.03	89.51	78.40
ER4\_1024	84.62	84.61	89.07	78.05
ER2\_2048	83.31	83.32	87.68	76.85
ER1\_4096	80.06	80.06	90.76	73.86
}\sqerdv

\pgfplotstableread{
Graph    relaxed     algdu    greedy
mycielskian20	72.91	82.42	62.21
com-Friendster	75.46	78.32	71.83
GAP-kron	78.38	80.81	71.11
GAP-urand	74.45	72.91	83.87
MOLIERE\_2016	70.66	72.17	73.69
AGATHA\_2015	78.33	79.64	78.45
WDC\_2014	72.25	74.19	75.06
BA\_512	82.94	88.83	82.18
}\sqgreedy

\pgfplotstableread{
Graph    relaxed     algdu    greedy
BA\_1024	83.03	90.13	82.08
BA\_2048	83.19	90.09	81.99
BA\_4096	83.09	90.09	82.04
UA\_1024	69.12	71.23	88.04
UA\_2048	69.72	72.58	86.86
UA\_4096	74.23	77.79	89.99
ER128\_32	85.40	90.10	99.21
ER64\_64	85.39	90.14	99.30
}\sqgreedytwo

%% file: Tikz/quality/sq_ssw.tex
\begin{tikzpicture}
\begin{axis}[
    ymin=40,    
    width=16cm,
    height=3.5cm,    
    ylabel={min-OPT \%},
    ylabel style={font=\small},
    enlargelimits=0.14,
    enlarge y limits=upper,     
    ybar, bar width=6pt,
    point meta=explicit,
    symbolic x coords={mycielskian20,com-Friendster,GAP-kron, GAP-urand, MOLIERE\_2016, AGATHA\_2015, WDC\_2014, WDC\_2012},
    x tick label style={rotate=15},
    x tick label style={font=\scriptsize},
    y tick label style={font=\scriptsize},
    ytick distance=10,
    nodes near coords={}    
]
\addplot [fill=teal, postaction={pattern=north east lines}] table [meta=k1] \sqssw;
\addplot [fill=red, postaction={pattern=crosshatch}] table [y=k128, meta=k128] \sqssw;
\addplot [fill=blue, postaction={pattern=grid}] table [y=algdu, meta=algdu] \sqssw;
\addplot [fill=orange, postaction={pattern=dots}] table [y=algss, meta=algss] \sqssw;
\end{axis}
\end{tikzpicture}

%% file: Tikz/quality/sq_ba.tex
\begin{tikzpicture}
\begin{axis}[
    ymin=50,    
    width=16cm,
    height=3.5cm,    
    ylabel={min-OPT \%},
    ylabel style={font=\small},
    enlargelimits=0.14,
    enlarge y limits=upper,    
    ybar, bar width=6pt,
    point meta=explicit,
    symbolic x coords={BA\_512, BA\_1024, BA\_2048, BA\_4096, BA\_8192, BA\_16384, BA\_32768, BA\_65536},
    x tick label style={rotate=15},    
    x tick label style={font=\scriptsize},
    y tick label style={font=\scriptsize},
    ytick distance=10,
    nodes near coords={}    
]
\addplot [fill=teal, postaction={pattern=north east lines}] table [meta=k1] \sqba;
\addplot [fill=red, postaction={pattern=crosshatch}] table [y=k128, meta=k128] \sqba;
\addplot [fill=blue, postaction={pattern=grid}] table [y=algdu, meta=algdu] \sqba;
\addplot [fill=orange, postaction={pattern=dots}] table [y=algss, meta=algss] \sqba;
\end{axis}
\end{tikzpicture}    

%% file: Tikz/quality/sq_er.tex
\begin{tikzpicture}
\begin{axis}[
    ymin=70,    
    width=16cm,
    height=3.5cm,    
    ylabel={min-OPT \%},
    ylabel style={font=\small},
    enlargelimits=0.14,
    enlarge y limits=upper,
    legend style={
        font=\scriptsize,
        at={(0.5,-0.75)},
        anchor=north,
        legend columns=-1        
    },
    ybar, bar width=6pt,    
    point meta=explicit,
    symbolic x coords={ER2\_256, ER1\_512, ER2\_512, ER1\_1024, ER2\_1024, ER1\_2048, ER2\_2048, ER1\_4096},
    x tick label style={rotate=15},
    x tick label style={font=\scriptsize},
    y tick label style={font=\scriptsize},
    ytick distance=5,
    nodes near coords={}    
]
\addplot [fill=teal, postaction={pattern=north east lines}] table [meta=k1] \sqer;
\addplot [fill=red, postaction={pattern=crosshatch}] table [y=k128, meta=k128] \sqer;
\addplot [fill=blue, postaction={pattern=grid}] table [y=algdu, meta=algdu] \sqer;
\addplot [fill=orange, postaction={pattern=dots}] table [y=algss, meta=algss] \sqer;
\legend{$k=1$, $k=128$, ALG-d, ALG-s}
\end{axis}
\end{tikzpicture}

%% file: Tikz/time/data.tex
\pgfplotstableread{
nstreams	WDC_2012	WDC_2014	AGATHA_2015	MOLIERE_2016	GAP-urand	GAP-kron	com-Friendster	mycielskian20
2   1.40	1.50	1.53	1.50	1.09	1.65	1.33	1.78
4	2.62	2.77	3.03	2.77	1.95	3.43	2.58	3.90
8	4.92	5.76	5.87	5.31	4.01	6.70	4.81	7.98
16	8.92	10.62	10.69	10.10	7.29	12.43	8.56	14.36
32	14.54	16.93	17.47	17.69	11.84	21.91	14.73	27.74
64	19.91	24.54	23.56	27.25	15.28	34.71	22.40	48.15
128	21.98	27.55	23.92	29.66	16.45	42.79	24.24	60.60
}\spssw

\pgfplotstableread{
nstreams	BA_512	BA_1024	BA_2048	BA_4096	BA_8192	BA_16384	BA_32768	BA_65536
2	1.764538135	1.838877568	1.861317248	1.837362208	1.816441691	1.80689213	1.86663501	2.101772768
4	3.118153136	3.398170095	3.576265931	3.579917227	3.493548367	3.525587887	3.662239125	4.301465796
8	6.032028652	6.526008116	6.916952921	7.173856738	6.173779185	6.313016985	6.994328639	8.599814969
16	9.888362599	11.79521855	12.20611644	12.56148312	13.08467544	13.29271401	14.65953238	17.09675288
32	17.95465429	21.14155571	22.35031535	23.03779334	24.05614538	25.1396215	28.0881132	32.46613166
64	28.37148763	33.02619297	34.92737467	37.11957632	40.35907891	43.93667375	47.99380702	56.71297942
128	36.86942832	41.85920689	44.77139207	49.40456762	53.72263907	65.52984241	72.88033714	69.99235671
}\spba

\pgfplotstableread{
nstreams	G1_4096	G1_2048	G1_1024	G1_512	G2_2048	G2_1024	G2_512	G2_256
2	1.81	1.84	1.81	1.82	1.81	1.82	1.82	1.83
4	3.63	3.64	3.65	3.64	3.60	3.60	3.56	3.62
8	7.17	7.32	7.24	7.20	7.10	7.00	6.96	7.03
16	13.82	14.51	14.35	14.04	14.57	14.18	14.02	13.99
32	27.60	27.55	26.79	27.37	26.70	26.26	24.97	26.51
64	51.40	49.31	47.53	48.69	47.47	45.91	43.93	43.88
128	83.02	80.77	79.18	78.58	83.35	80.99	75.39	68.27
}\sper

\pgfplotstableread{
Graph 	pr	pp	st
WDC2014(k=1)	7.446	14.040	730.50
WDC2014(k=128)	0.666	1.873	24.76
WDC2012(k=1)	15.760	76.870	1771.00
WDC2012(k=128)	1.021	7.981	75.77
BA32768(k=1)	0.042	0.350	1627.00
BA32768(k=128)	0.008	0.092	22.23
BA65536(k=1)	0.044	0.353	2084.00
BA65536(k=128)	0.004	0.076	29.70
G22048(k=1)	0.014	0.125	4450.00
G22048(k=128)	0.005	0.047	53.34
G14096(k=1)	0.014	0.127	8026.00
G14096(k=128)	0.005	0.047	96.63
}\algps

\pgfplotstableread{
nstreams	UA_1024	UA_2048	UA_4096	UA_8192	UA_16384	UA_32768	UA_65536	UA_131072
2	1.799137903	1.638746828	1.771284651	1.795953851	1.776863801	1.834900432	1.83398584	1.840414297
4	3.254696176	3.182771847	3.532182554	3.62793607	3.502054517	3.481872503	3.548136044	3.505232186
8	6.832407357	7.306770775	7.09735622	7.22764959	7.022663628	7.082627688	7.069240421	7.244429249
16	12.0889398	13.12233503	13.18204977	13.67463419	13.34065527	13.39265629	13.35300012	13.30424829
32	21.97187636	22.67719401	23.09121318	24.55645678	24.19757881	24.90598015	24.57613431	24.84460638
64	32.67880281	35.5959083	36.75097539	39.84290174	39.50267386	42.15066309	42.19763335	42.26690888
128	39.43632285	43.83467901	44.26424548	48.17236451	52.57337372	60.8192427	61.29093175	60.51724865
}\spua

\pgfplotstableread{
nstreams	G1_4096	G2_2048	G4_1024	G8_512	G16_256	G32_128	G64_64	G128_32
2	1.81	1.81	1.89	1.88	1.90	1.95	1.87	1.97
4	3.63	3.60	3.77	3.75	3.78	3.79	3.69	3.75
8	7.17	7.10	7.45	7.25	7.12	7.32	6.92	6.81
16	13.82	14.57	14.93	14.21	13.97	13.67	12.45	11.79
32	27.60	26.70	26.88	25.98	25.16	24.09	19.24	19.47
64	51.40	47.47	46.09	42.75	40.80	35.52	23.77	23.04
128	83.02	83.35	81.44	76.00	60.79	44.82	41.69	39.47
}\sperdv

\pgfplotstableread{
nstreams	WDC_2014	MOLIERE_2016	GAP-kron	com-Friendster	WDC_2014(NLD)	MOLIERE_2016(NLD)	GAP-kron(NLD)	com-Friendster(NLD)
2	1.50	1.50	1.65	1.33	1.28	1.17	1.24	1.24
4	2.77	2.77	3.43	2.58	2.56	2.62	2.71	2.55
8	5.76	5.31	6.70	4.81	5.15	5.15	5.80	5.20
16	10.62	10.10	12.43	8.56	9.69	10.13	11.57	9.90
32	16.93	17.69	21.91	14.73	16.05	17.39	20.20	14.73
64	24.54	27.25	34.71	22.40	21.42	17.56	22.41	14.67
128	27.55	29.66	42.79	24.24	17.18	15.01	17.45	12.28
}\spNLD

\pgfplotstableread{
nstreams	G1_4096	G1_2048	G1_1024	G1_512	G1_4096(NLD)	G1_2048(NLD)	G1_1024(NLD)	G1_512(NLD)
2	1.81	1.84	1.81	1.82	1.85	1.97	1.97	1.86
4	3.63	3.64	3.65	3.64	4.01	3.68	3.41	3.42
8	7.17	7.32	7.24	7.20	4.35	3.19	2.90	3.17
16	13.82	14.51	14.35	14.04	8.73	8.16	7.64	8.47
32	27.60	27.55	26.79	27.37	17.56	17.18	15.47	16.92
64	51.40	49.31	47.53	48.69	34.09	36.10	35.36	32.57
128	83.02	80.77	79.18	78.58	54.28	59.73	65.08	70.53
}\spGNLD

\pgfplotstableread{
ngroups	UA_4096_67	UA_8192_67	UA_2048_134	UA_4096_134	UA_1024_268	UA_2048_268	UA_512_537	UA_1024_537
1	6.734334896	7.691096099	5.297270831	5.756258061	5.067922492	5.168585273	5.330605148	5.302769422
2	13.17329652	14.7111806	11.38442417	12.2178572	11.5723789	12.19683096	9.10498183	9.187980433
4	22.43448498	24.64610882	21.50129567	22.41344512	13.19748491	14.6028268	8.128426296	8.634042667
8	24.39917724	26.51108221	21.99505559	24.31207871	11.19315355	12.04024324	13.9669025	15.01222241
16	27.87008498	29.98735453	21.03357929	23.16229053	17.63768843	20.7984804	16.08816471	19.61772611
}\spuaxy

%% file: Tikz/time/sp_ssw.tex
\begin{tikzpicture}[scale=0.7]
\begin{axis}[        
    xlabel={$k$},    
    xlabel style={font=\Large},
    ylabel={Speedup in runtime},
    ylabel style={font=\Large},
    ymajorgrids=true,
    grid style=dashed,
    legend style={        
        font=\large,
        nodes={scale=0.5, transform shape},
        cells={anchor=west},
        at={(0.75,0.45)},
        anchor=north,
        draw=none,
        fill=none},
    xmode=log,
    ymode=log,
    log basis y={2},
    log basis x={2},
    ]

\addplot [mark=otimes,olive] table [x={nstreams}, y={WDC_2012}] \spssw;

\addplot [mark=star,teal] table [x={nstreams}, y={WDC_2014}] \spssw;

\addplot [mark=pentagon,cyan] table [x={nstreams}, y={AGATHA_2015}] \spssw;

\addplot [mark=triangle,purple] table [x={nstreams}, y={MOLIERE_2016}] \spssw;

\addplot [mark=o,violet] table [x={nstreams}, y={GAP-urand}] \spssw;

\addplot [mark=diamond,blue] table [x={nstreams}, y={GAP-kron}] \spssw;

\addplot [mark=10-pointed star,orange] table [x={nstreams}, y={com-Friendster}] \spssw;

\addplot [mark=oplus,red] table [x={nstreams}, y={mycielskian20}] \spssw;

\legend{WDC\_2012, WDC\_2014, AGATHA\_2015, MOLIERE\_2016, GAP-urand, GAP-kron, com-Friendster, mycielskian20};

\end{axis}
\end{tikzpicture}

%% file: Tikz/time/sp_ba.tex
\begin{tikzpicture}[scale=0.7]
\begin{axis}[        
    xlabel={$k$},
    xlabel style={font=\Large},
    ylabel={},    
    ymajorgrids=true,
    grid style=dashed,
    legend style={       
        font=\large,
        nodes={scale=0.5, transform shape},
        cells={anchor=west},
        at={(0.75,0.45)},
        anchor=north,
        draw=none,
        fill=none},
    xmode=log,
    ymode=log,
    log basis y={2},
    log basis x={2},
    ]

\addplot [mark=otimes,olive] table [x={nstreams}, y={BA_65536}] \spba;

\addplot [mark=star,teal] table [x={nstreams}, y={BA_32768}] \spba;

\addplot [mark=pentagon,cyan] table [x={nstreams}, y={BA_16384}] \spba;

\addplot [mark=triangle,purple] table [x={nstreams}, y={BA_8192}] \spba;

\addplot [mark=o,violet] table [x={nstreams}, y={BA_4096}] \spba;

\addplot [mark=diamond,blue] table [x={nstreams}, y={BA_2048}] \spba;

\addplot [mark=10-pointed star,orange] table [x={nstreams}, y={BA_1024}] \spba;

\addplot [mark=oplus,red] table [x={nstreams}, y={BA_512}] \spba;

\legend{BA\_65536, BA\_32768, BA\_16384, BA\_8192, BA\_4096, BA\_2048, BA\_1024, BA\_512};

\end{axis}
\end{tikzpicture}

%% file: Tikz/time/sp_er.tex
\begin{tikzpicture}[scale=0.7]
\begin{axis}[        
    xlabel={$k$},
    xlabel style={font=\Large},
    ylabel={},
    ymajorgrids=true,
    grid style=dashed,
    legend style={        
        font=\large,
        nodes={scale=0.5, transform shape},
        cells={anchor=west},
        at={(0.75,0.45)},
        anchor=north,
        draw=none,
        fill=none},
    xmode=log,
    ymode=log,
    log basis y={2},
    log basis x={2},
    ]

\addplot [mark=otimes,olive] table [x={nstreams}, y={G1_4096}] \sper;

\addplot [mark=star,teal] table [x={nstreams}, y={G2_2048}] \sper;

\addplot [mark=pentagon,cyan] table [x={nstreams}, y={G1_2048}] \sper;

\addplot [mark=triangle,purple] table [x={nstreams}, y={G2_1024}] \sper;

\addplot [mark=o,violet] table [x={nstreams}, y={G1_1024}] \sper;

\addplot [mark=diamond,blue] table [x={nstreams}, y={G2_512}] \sper;

\addplot [mark=10-pointed star,orange] table [x={nstreams}, y={G1_512}] \sper;

\addplot [mark=oplus,red] table [x={nstreams}, y={G2_256}] \sper;

\legend{ER1\_4096, ER2\_2048, ER1\_2048, ER2\_1024, ER1\_1024, ER2\_512, ER1\_512, ER2\_256};

\end{axis}
\end{tikzpicture}

%% file: Tikz/time/breakdown.tex
\begin{tikzpicture}[scale=0.7]
\begin{axis}[
    ymin=0.00390625,
    width=24cm,
    height=6cm,    
    ylabel={Runtime (seconds)},
    ylabel style={font=\Large},
    enlargelimits=0.14,
    enlarge y limits=upper,    
    ybar,
    bar width=8pt,
    enlargelimits=0.14,
    enlarge y limits=upper,
    ymajorgrids=true,
    grid style=dashed,
    ymode=log,
    log basis y={2},
    log origin=infty,
    legend style={    
        font=\Huge,
        nodes={scale=0.5, transform shape},
        cells={anchor=north},
        at={(0.5,1.3)},
        anchor=north,
        legend columns=3,        
        fill=none},    
    point meta=explicit,
    symbolic x coords={WDC2014(k=1),WDC2014(k=128), WDC2012(k=1), WDC2012(k=128), BA32768(k=1), BA32768(k=128), BA65536(k=1),BA65536(k=128), G22048(k=1), G22048(k=128), G14096(k=1), G14096(k=128)},
    x tick label style={rotate=5},    
    x tick label style={font=\large},
    y tick label style={font=\small},
    xtick={WDC2014(k=1),WDC2012(k=1), BA32768(k=1), BA65536(k=1), G22048(k=1),G14096(k=1)},
    xticklabels={
      WDC\_2014, WDC\_2012, BA\_32768, BA\_65536, ER2\_2048, ER1\_4096},
    x tick label style={xshift=3ex},
    ytick distance=16,
    nodes near coords={}    
]
\addplot [fill=blue, postaction={pattern=grid}] table [meta=pr] \algps;
\addplot [fill=red, postaction={pattern=crosshatch}] table [y=pp, meta=pp] \algps;
\addplot [fill=teal, postaction={pattern=north east lines}] table [y=st, meta=st] \algps;

\node[anchor=south, font=\small] at (axis cs:{WDC2014(k=1)},760) {k=1};
\node[anchor=south, font=\small] at (axis cs:{WDC2014(k=128)},30) {k=128};
\node[anchor=south, font=\small] at (axis cs:{WDC2012(k=1)},1850) {k=1};
\node[anchor=south, font=\small, xshift=-0.5ex] at (axis cs:{WDC2012(k=128)},85) {k=128};

\node[anchor=south, font=\small] at (axis cs:{BA32768(k=1)},1650) {k=1};
\node[anchor=south, font=\small] at (axis cs:{BA32768(k=128)},45) {k=128};
\node[anchor=south, font=\small] at (axis cs:{BA65536(k=1)},2100) {k=1};
\node[anchor=south, font=\small] at (axis cs:{BA65536(k=128)},30) {k=128};

\node[anchor=south, font=\small] at (axis cs:{G22048(k=1)},4550) {k=1};
\node[anchor=south, font=\small] at (axis cs:{G22048(k=128)},85) {k=128};
\node[anchor=south, font=\small] at (axis cs:{G14096(k=1)},8100) {k=1};
\node[anchor=south, font=\small] at (axis cs:{G14096(k=128)},100) {k=128};
\legend{Preprocessing, Post-processing, Streaming};
\end{axis}
\end{tikzpicture}

%% file: 6_conc.tex
\section{Conclusion}
\label{sec:conc}
While numerous studies have focused on optimizing either time (in parallel computing) or space (in streaming algorithms) in isolation, the poly-streaming model offers a \emph{practically relevant paradigm} for jointly optimizing both. 
It fills a gap by providing a formal framework for analyzing algorithmic design choices and their associated time--space trade-offs.
Our study of matchings illustrates the practical relevance of this paradigm in supporting diverse design choices and enabling principled analysis of their trade-offs.

The simplicity of our matching algorithm and its generalization reflects our choice to adopt the design of~\cite{paz20182+}.
We believe this principle will inspire the development of other poly-streaming algorithms.
To this end, we note that~\cite{paz20182+} has also motivated simple algorithms for related problems, such as matchings with submodular objectives~\cite{levin2021streaming}, $b$-matchings~\cite{huang2024semi}, and collections of disjoint matchings~\cite{ferdous2024semi}.

Our study focuses on computing matchings in single-pass, shared-memory settings. The same framework may also be effective in multi-pass and distributed-memory settings. These directions are discussed in \hyref{Appendix}{subsec:dist} and \hyref{Appendix}{subsec:app_duals}.

%% file: 7_app1.tex
\section{Related Models of Computation}
\label{sec:related_models}

\paragraph{The Streaming Model \cite{alon1996space, henzinger1998computing}.}
In the streaming model of computation, a sequence of data items is fed to an algorithm as a data stream. 
The algorithm reads and processes one item at a time from the stream. 
For a data stream of $N$ items, drawn from a universe $\{1,\ldots,M\}$, a streaming algorithm is allowed to use $\smallO{\min{ \{N,M\} } }$ space. 
The ultimate goal is to compute a solution using $\bigO{\log N + \log M }$ space. 
A streaming algorithm may be restricted to a single pass over the stream (\emph{single-pass}) or allowed multiple passes (\emph{multi-pass}).

\paragraph{The Semi-Streaming Model \cite{feigenbaum2005graph}.}
General graph problems are intractable if the space available is $\smallO{n}$, where $n$ denotes the number of vertices. 
To address this, \cite{feigenbaum2005graph} introduced the semi-streaming model. 
In this model, an algorithm has sequential access to the edges of a graph (i.e., a stream of edges) and is allowed to use $\bigO{n \cdot \plog{n} }$ space.\footnote{$\bigOmega{n \log n}$ bits are needed just to store $n$ integers.} Similar to a streaming algorithm, a semi-streaming algorithm is either \emph{single-pass} or \emph{multi-pass}.

\paragraph{The Distributed Streaming Model \cite{gibbons2001estimating}.}
In the distributed streaming model \cite{gibbons2001estimating}, $t$ processors (or parties) process $t$ data streams independently and generate summaries of their respective streams. 
These summaries are then sent simultaneously to a central referee (processor), who estimates a global function over the union of the streams. 
Each processor is allowed to read its own data stream only once and must operate using space sublinear in the size of the stream.
Communication is modeled using the \emph{simultaneous one-round} communication complexity model. 

A related model is the \emph{continuous distributed monitoring model} \cite{cormode2013continuous}, where the goal is to observe $t$ distributed streams using $t$ sites (processors) and continuously maintain a global function over the streams using a central coordinator (processor).

\paragraph{The Work-Depth Model \cite{jaja1992introduction, karp1988survey}.} 
The (shared-memory parallel) work-depth model is used to analyze algorithms designed for \emph{PRAM}-like shared-memory machines. 
An algorithm in this model is evaluated using two measures: \emph{work}, the total computation performed, and \emph{depth}, the length of the longest chain of sequential dependencies.
This is an offline model of parallel computation.

\paragraph{The MPC Model \cite{andoni2014parallel, beame2017communication, goodrich2011sorting, karloff2010model}.} 
The massively parallel computation (MPC) model was introduced in \cite{karloff2010model}, and refined in subsequent work \cite{andoni2014parallel, beame2017communication, goodrich2011sorting}. 
In its general setting, the model consists of $N$ data items distributed across $M$ machines, each with $S$ bits of local memory. The primary interest lies in the regime where $S=N^{\alpha}$ for some $\alpha \in (0, 1)$, and $N=\bigO{M \cdot S}$, meaning no single machine can store the entire input, but the collective memory is sufficient to store all data items.
Computation proceeds in synchronous rounds. 
In each round, machines perform local computation on their data. 
At the end of the round, machines may communicate, subject to the constraint that each machine receives only as much data as fits in its $S$ bits of memory. 
Algorithms in the MPC model are primarily assessed on the number of communication rounds required to solve a problem.

%% file: 8_app2.tex
\section{Deferred Proofs and Techniques for Matching}
\label{sec:app_mwm}

\subsection{Related Algorithms}
\label{subsec:mwm_related}
The first exact algorithm for the \emph{maximum weight matching} (MWM) problem is due to Edmonds~\cite{edmonds1965paths}, and the fastest known algorithm, with runtime $\bigOtilde{m \sqrt{n}}$, is due to Gabow and Tarjan~\cite{gabow1991faster}. For offline approximation algorithms in sequential and parallel settings, we refer the reader to~\cite{assadi2019distributed, huang20231, pothen2019approximation} and the references therein. For MPC algorithms and multi-pass streaming algorithms, see~\cite{assadi2024simple, assadi2019distributed, gamlath2019weighted} and the references therein. We next summarize the literature on single-pass streaming algorithms.

The first streaming algorithm for the MWM problem is due to Feigenbaum et al.~\cite{feigenbaum2005graph}, achieving a 6-approximation guarantee. The algorithm maintains an initially empty matching. When an edge $e=\{u,v\}$ arrives, it examines the edges in the current matching that are incident to $u$ or $v$, and computes their total weight. If the weight of $e$ is at most twice this total, the edge is ignored; otherwise, the incident edges are removed from the matching and $e$ is inserted. McGregor~\cite{mcgregor2005finding} improved the approximation ratio from $6$ to $5.828$ by replacing the threshold 2 with $1+\gamma$, and selecting the optimal value of $\gamma$.

Zelke~\cite{zelke2012weighted}, Epstein et al.~\cite{epstein2011improved}, and Crouch and Stubbs~\cite{crouch2014improved} used different techniques to improve the approximation ratio to $5.85$, $4.91 + \epsilon$, and $4 + \epsilon$, respectively. Paz and Schwartzman~\cite{paz20182+} applied the local-ratio technique~\cite{bar2004local} to obtain a $2+\epsilon$-approximation algorithm. We extend their design to the \emph{poly-streaming} setting; see \hyref{Section}{subsec:mwm_ps} for details. Ghaffari and Wajc~\cite{ghaffari2019simplified} extended this algorithm to achieve optimal space complexity in the \emph{semi-streaming model}.

The single-pass streaming algorithms described above make no assumptions about the ordering of edges; they are designed to work under arbitrary or adversarial ordering. Assadi and Behnezhad~\cite{assadi2021beating} presented a single-pass algorithm that achieves better than a $1.5$-approximation for the \emph{maximum matching} problem when edges arrive in random order. In contrast, Kapralov~\cite{kapralov2021space} showed that, under adversarial ordering, no single-pass algorithm can achieve better than a $1+\ln 2 \approx 1.7$-approximation, even for unweighted bipartite graphs.

\subsection{Proof of Lemma~\ref{lemma_tight_follower}--\ref{lemma_free_pp}}
\label{app:lemma_free_pp}

Unlike \aref{Process-Edge}{fig:algo_proc_edge}, the subroutine \hyperref[fig:algo_proc_stack]{Process-Stack} does not use locks. To justify this, we formally verify that its steps can be executed asynchronously. \aref{Process-Stack}{fig:algo_proc_stack} makes progress by identifying tight edges. In the proof of \lemref{lemma_free}, we relied on a characterization of the lexicographic ordering of vertices (\hyref{Proposition}{prop_ei_dag}). Analogously, we now establish a few useful characterizations of tight edges.

\begin{lemma}
\label{lemma_tight}
Let $e \in \U^{\br{t}}$ be an edge contained in a stack $S^{\ell}$, and suppose no edge in $\N(e)$ was added to any stack after $e$. Then $e$ is a tight edge.
\end{lemma}
\begin{proof}
Consider the execution of Step~3 in \hyperref[fig:algo_proc_edge]{Process-Edge}, by which the edge $e=\{u,v\}$ was included in the stack $S^{\ell}$. For any vertex $x\in V$, let $\alpha_x^{\text{old}}$ and $\alpha_x^{\text{new}}$ denote the values of $\alpha_x$ before and after Step~3(b), respectively. 

By Step~3(a), we have $\alpha_{u}^{\text{old}} + \alpha_{v}^{\text{old}} = w_e - g_e$, and  by Step~3(b),
\[\alpha_{u}^{\text{new}}+\alpha_{v}^{\text{new}} = \alpha_{u}^{\text{old}}+ \alpha_{v}^{\text{old}}+ 2g_e = w_e -g_e + 2g_e =w_e + g_e.\] 
By assumption, no edge in $\N(e)$ was included in any stack after $e$, so the values of $\alpha_u$ and $\alpha_v$ remain unchanged until $e$ is processed. Hence, when $e$ is processed, we have $\alpha_{u}+\alpha_{v} = w_e + g_e$, which proves the claim.
\end{proof}

\begin{remark}
\label{rem_follower1}
For any edge $e=\{u,v\}$ added to some stack $S^i$, the locks of $u$ and $v$ are held exclusively by processor $i$, ensuring that no other edge in $\N(e)$ can be added to any stack $S^j$ with $j \ne i$ while $e$ is being pushed onto $S^i$. Therefore, the follower relationship is always asymmetric: if $e_j \in \F(e_i)$, then $e_i \not\in \F(e_j)$.
\end{remark}

\begin{remark}
\label{rem_follower2}
If an edge $e_j$ is included in some stack $S^j$, after another edge $e_i$ has been included in some stack $S^i$, then $e_i$ cannot be a follower of $e_j$. Note that this does imply that $e_j$ is a follower of $e_i$.
\end{remark}

\LemmaTightFollower*
\begin{proof}
\textbf{Only if.}
Suppose $e=\{u,v\}$ is a tight edge; that is, $w_e + g_e = \alpha_u + \alpha_v$. For contradiction, assume $\F(e) \ne \emptyset$, and let $e_j \in \F(e)$.
By \hyref{Definition}{def_follower}, we must have $e_j \cap e \ne \emptyset$ and $e_j$ was included in some stack after $e$. 

As shown in the proof of \lemref{lemma_tight}, the inclusion of $e$ implies that immediately after Step~3(b) of \hyperref[fig:algo_proc_edge]{Process-Edge}, we had $\alpha_u+\alpha_v = w_e + g_e$.

Since $e_j$ shares an endpoint with $e$ and was included after $e$, its gain $g_{e_j}> 0$ must have been added to at least one of $\alpha_u$ or $\alpha_v$.\footnote{If $e_j$ is parallel to $e$ then both $\alpha_u$ and $\alpha_v$ would have been incremented by $g_{e_j}$.} 
Therefore, after the inclusion of $e_j$, we must have $\alpha_u+\alpha_v > w_e + g_e$, contradicting the assumption that $e$ is a tight edge. Thus, $\F(e) = \emptyset$.

\medskip

\noindent\textbf{If.} 
Suppose $\F(e) = \emptyset$. If $e$ is the last edge from $\N(e) \cup \{e\}$ to be added to any stack, then by \lemref{lemma_tight}, $e$ is a tight edge. 

Otherwise, $\F(e) \ne \emptyset$ at the beginning of the execution of \hyperref[fig:algo_proc_stack]{Process-Stack},  but becomes empty in some iteration of Step~2. 

Every edge $e_j = \{x,y\} \in \F(e)$ must have been removed in Step~2(a), with its gain $g_{e_j}$ subtracted from $\alpha_x$ and $\alpha_y$ in Step~2(d). These updates precisely reverse the effect of $e_j$'s inclusion on $\alpha_u$ and $\alpha_v$. Consequently, once all followers of $e$ have been processed, the values of $\alpha_u$ and $\alpha_v$ are exactly as they would be if none of the followers of $e$ had ever been added to any stack. By \lemref{lemma_tight}, it follows that $e$ is a tight edge.
\end{proof}

\begin{lemma}
\label{lemma_follower}
Let $e_i,e_j \in \U^{\br{t}}$ be the top edges of stacks $S^i$ and $S^j$, respectively. If $S^j \cap \F(e_i) \neq \emptyset$, then $S^i \cap \F(e_j) = \emptyset$.
\end{lemma}
\begin{proof}
Assume $S^j \cap \F(e_i) \neq \emptyset$; that is, $S^j$ contains a follower of $e_i$. 
Let $e_{\ell} \in  S^j \cap \F(e_i)$ be such an edge. 
Since $e_j$ is the top edge of $S^j$, $e_{\ell}$ must have been added to $S^j$ no later than $e_j$.

By \hyref{Definition}{def_follower}, $e_{\ell} \in \F(e_i)$ implies that $e_{\ell}$ was included in $S^j$ strictly after $e_i$ was included in $S^i$. 
But $e_i$ is the top edge of $S^i$, so all edges in $S^i$ were added no later than $e_i$. It follows that $e_{\ell}$ was added strictly after all edges of $S^i$. 
Since $e_{\ell}$ was added no later than $e_j$, $e_{j}$ was also added strictly after all edges of $S^i$. Therefore, no edge in $S^i$ can be a follower of $e_j$; that is, $S^i \cap \F(e_j) = \emptyset$.
\end{proof}

\LemmaTightEdge*
\begin{proof}
Consider a directed graph $G_{\T}^{\br{t}}=\br{\T^{\br{t}}, E_{\T}^{\br{t}}}$, where for each ordered pair $\br{e_i, e_j}$ of distinct edges in $\T^{\br{t}}$, we include a directed edge $(e_i, e_j)$ in $E_{\T}^{\br{t}}$ if stack $S_j$ contains a follower of $e_i$ (i.e., $S_j \cap \F(e_i) \ne \emptyset$). By \lemref{lemma_follower}, if $(e_i, e_j) \in E_{\T}^{\br{t}}$, then $(e_j, e_i) \not\in E_{\T}^{\br{t}}$.

We claim that $G_{\T}^{\br{t}}$ is a directed acyclic graph (DAG). Suppose, for contradiction, that it contains a directed cycle $C=\langle e_{i_1}, e_{i_2}, ..., e_{i_j}, e_{i_1} \rangle$, with $j \geq 3$ by \lemref{lemma_follower}. 
For each $q=1 \dots j-1$, the edge $\br{e_{i_q}, e_{i_{q+1}}} \in E_{\T}^{\br{t}}$, so $S^{i_{q+1}}$ contains a follower of $e_{i_q}$. 
By \hyref{Definition}{def_follower}, this implies that $e_{i_{q+1}}$ was included strictly after $e_{i_q}$ was included. Chaining these inequalities, we deduce that $e_{i_j}$ was included strictly after $e_{i_1}$ was included. But the existence of the cycle $C$ implies that $e_{i_1}$ is a follower of $e_{i_j}$; that is, $e_{i_1}$ was included strictly after $e_{i_j}$, a contradiction. Hence, $G_{\T}^{\br{t}}$ is a DAG.

Since any DAG contains at least one vertex without an outgoing edge, $G_{\T}^{\br{t}}$ contains some vertex $e_i \in \T^{\br{t}}$ such that $(e_i, e_j) \not\in E_{\T}^{\br{t}}$ for all $e_j \ne e_i$, that is, $S_j \cap \F(e_i) = \emptyset$ for all $j \ne i$. Since $e_i$ is the top edge of stack $S_i$, we also have $S^i \cap \F(e_i) = \emptyset$. Hence, $\F(e_i) = \emptyset$, and by \lemref{lemma_tight_follower}, $e_i$ is a tight edge.
\end{proof}

\LemmaVertexDisjoint*
\begin{proof}
Consider two distinct tight edges $e_i, e_j \in \U^{\br{t}}$. Suppose, for contradiction, that $e_i \cap e_j \ne \emptyset$. 
Then by \hyref{Definition}{def_follower}, either $ e_j \in \F\br{e_i}$ or $ e_i \in \F\br{e_j}$. But since $e_i$ and $e_j$ are tight, \lemref{lemma_tight_follower} gives $\F\br{e_i} = \F\br{e_j} = \emptyset$, a contradiction.
\end{proof}

\lemref{lemma_vertex_disjoint} justifies our claim that tight edges can be processed asynchronously (and they do not even need to reside at the top of the stacks for processing). We can now complete the proof of \lemref{lemma_free_pp}.

\LemmaFreePP*
\begin{proof}
The claim follows directly from \lemref{lemma_tight_edge} and its proof. The lemma constructs a dependency graph over the set of top edges in each iteration of Step~2 of \aref{Process-Stack}{fig:algo_proc_stack}, and shows that this graph is acyclic. This rules out cyclic dependencies, and guarantees the absence of both deadlock and livelock. 

To establish the absence of starvation, note that \lemref{lemma_tight_edge} guarantees that each iteration removes at least one edge from the stacks. Since the total number of stacked edges is at most $\bigOtilde{n / \epsilon}$, any processor may be blocked at Step~2(b) for at most $\bigOtilde{n / \epsilon}$ iterations.
\end{proof}

\subsection{The Deferrable Strategy}
\label{subsec:mwm_ds}

\begin{figure}[h]
    \centering
    \begin{parbox}{5.4in}{    
        \begin{mdframed}[linewidth=0.5pt, roundcorner=7pt, backgroundcolor=gray!5, frametitle={\underline{PS-MWM-DS$(V, \ell, \epsilon)$}}]
        \begin{enumerate}
        \item In parallel initialize $lock_u$, and set $\alpha_u$ and $mark_u$ to $0$ for all $u \in V$ \textcolor{teal}{\\/* processor $\ell$ initializes or sets $\Theta(n/k)$ locks or variables */}        
        \item set $S^{\ell} \leftarrow \emptyset$ \textcolor{teal}{/* initialize an empty stack */}
        \item set $R^{\ell} \leftarrow \emptyset$ \textcolor{teal}{/* for storing edges unresolved in the streaming phase */}
        \item for each edge $e=\{u,v\}$ in $\ell${th} stream do
        \begin{enumerate}
            \item \aref{Process-Edge-DS$(e, S^{\ell}, R^{\ell}, \epsilon)$}{fig:algo_proc_edge_ds}   \textcolor{teal}{/* process an edge in the streaming phase */}
        \end{enumerate}        
        \item for each edge $e=\{u,v\} \in R^{\ell}$ do 
        \begin{enumerate}
            \item \aref{Process-Edge$(e, S^{\ell}, \epsilon)$}{fig:algo_proc_edge}  \textcolor{teal}{/* process an edge in the post-processing phase */}
        \end{enumerate}
        \item wait for all processors to complete execution of Step~5 \textcolor{teal}{/* a barrier */}
        \item $\M^{\ell} \leftarrow$ \aref{Process-Stack$(S^{\ell})$}{fig:algo_proc_stack}
        \item return $\M^{\ell}$
    \end{enumerate}
        \end{mdframed}    
    
    }
    \end{parbox}    
    \caption{Modifications of \aref{Algorithm PS-MWM}{fig:algo_psmwm} for the \emph{deferrable strategy}.}
    \label{fig:algo_psmwm_ds}
\end{figure}

\begin{figure}[h]
    \centering
    \begin{parbox}{5.6in}{    
        \begin{mdframed}[linewidth=0.5pt, roundcorner=7pt, backgroundcolor=gray!5, frametitle={\underline{Process-Edge-DS$(e=\{u,v\}, S^{\ell}, R^{\ell}, \epsilon)$}}]
        \textcolor{teal}{/* Assumes access to global variables $\{\alpha_x\}_{x \in V}$ 
 and locks $\{lock_x\}_{x \in V}$ */}
    \begin{enumerate}
        \item if $w_e \leq (1+\epsilon)(\alpha_u +\alpha_v)$ then return
        \item In $\bigO{1}$ attempts, try to acquire $lock_u$ and $lock_v$ in lexicographic order of $u$ and $v$
        \item if Step~2 fails to acquire the locks then
        \begin{enumerate}
            \item if $w_e \leq (1+\epsilon)(\alpha_u +\alpha_v)$ then return
            \item else include $e$ in $R^{\ell}$, and return \textcolor{teal}{\\/* defers decision to the post-processing phase */}
        \end{enumerate}        
        \item if $w_e > (1+\epsilon)(\alpha_u +\alpha_v)$ then
        \begin{enumerate}
            \item $g_e \leftarrow w_e - (\alpha_u +\alpha_v)$
            \item increment $\alpha_u$ and $\alpha_v$ by $g_e$
            \item add $e$ to the top of $S^{\ell}$ along with $g_e$
        \end{enumerate}
        \item release $lock_u$ and $lock_v$, and return
    \end{enumerate}
        \end{mdframed}    
    }
    \end{parbox}    
    \caption{A subroutine used in \aref{Algorithm PS-MWM-DS}{fig:algo_psmwm_ds}.}
    \label{fig:algo_proc_edge_ds}
\end{figure}

Modifications to \aref{PS-MWM}{fig:algo_psmwm} for the \emph{deferrable strategy} are outlined in \aref{Algorithm PS-MWM-DS}{fig:algo_psmwm_ds}. The algorithm invokes a new subroutine, \hyperref[fig:algo_proc_edge_ds]{Process-Edge-DS}, in Step~4, and continues to use the subroutines \hyperref[fig:algo_proc_edge]{Process-Edge} and \hyperref[fig:algo_proc_stack]{Process-Stack} in Steps~5 and~7, respectively. In \hyperref[fig:algo_psmwm_ds]{PS-MWM-DS}, Steps~1--3 constitute the preprocessing phase, Step~4 is the streaming phase, and Steps~5--7 form the post-processing phase. 

Each processor $\ell \in [k]$ executes the algorithm asynchronously, except for a synchronization barrier at the start of Step~7 (via Step~6). During the streaming phase (Step~4), processors may defer processing certain edges by placing them in the data structures $\{R^{\ell}\}_{\ell \in [k]}$. These deferred edges, if any, are then processed in the post-processing phase (Step~5).

\hyperref[fig:algo_psmwm_ds]{PS-MWM-DS} invokes the subroutine \hyperref[fig:algo_proc_edge_ds]{Process-Edge-DS} only during the streaming phase. For each edge in the streams, this subroutine takes $\bigO{1}$ time. In Step~2 of \hyperref[fig:algo_proc_edge_ds]{Process-Edge-DS}, the algorithm attempts to acquire the locks corresponding to the edge's endpoints within $\bigO{1}$ tries. 
If this step fails and the edge $e$ remains eligible for inclusion in the stack, then the algorithm defers it by placing $e$ into $R^{\ell}$ (Step~3(b) of \hyperref[fig:algo_proc_edge_ds]{Process-Edge-DS}), to be handled in the post-processing phase.

At the beginning of the post-processing phase (Step~5 of \hyperref[fig:algo_psmwm_ds]{PS-MWM-DS}), each processor $\ell$ processes its deferred edges, if any, by treating the contents of $R^{\ell}$ as an edge stream.

Using the analysis of the non-deferrable strategy, we now sketch an analysis of the deferrable strategy (\aref{PS-MWM-DS}{fig:algo_psmwm_ds}).

We extend the definition of superstep (\hyref{Definition}{def_ss}), by treating the substeps of Step~3 in \hyperref[fig:algo_proc_edge_ds]{Process-Edge-DS} as one superstep, and grouping all other steps in \hyperref[fig:algo_proc_edge_ds]{Process-Edge-DS} into another superstep. We refer to this extended notion as a \emph{ds-superstep}, and the corresponding dependency graph by $G_{D}^{\br{t}}$.

Adapting \lemref{lemma_free} to $G_{D}^{\br{t}}$ and applying \lemref{lemma_free_pp} then shows that \aref{PS-MWM-DS}{fig:algo_psmwm_ds} is free from deadlock, livelock, and starvation.

\begin{lemma}
\label{lemma_algo1_space_ds} 
For any constant $\epsilon > 0$, space complexity and per-edge processing time of \aref{Algorithm PS-MWM-DS}{fig:algo_psmwm_ds} are $\bigO{k\cdot n\log n}$ and $\bigO{1}$, respectively.
\end{lemma}
\begin{proof}
During the streaming phase (Step~4 of \hyperref[fig:algo_psmwm_ds]{PS-MWM-DS}), each edge is processed by the subroutine \hyperref[fig:algo_proc_edge_ds]{Process-Edge-DS}, which takes $\bigO{1}$ time per-edge. Hence, the per-edge processing time is $\bigO{1}$.

To bound the space usage, consider the $t$th effective iteration and the corresponding dependency graph $G_D^{\br{t}}$. Suppose processor $\ell$ executes Step~4 of \hyperref[fig:algo_proc_edge_ds]{Process-Edge-DS} or Step~3 of \hyperref[fig:algo_proc_edge]{Process-Edge} for some edge $e_{\ell}$. 
Then processor $\ell$ participates in a component of $G_D^{\br{t}}$, containing $e_{\ell}$, with at most $k-1$ other processors. 
As a result, at most $k-1$ processors may execute Step~3(b) of \hyperref[fig:algo_proc_edge_ds]{Process-Edge-DS} during the $\br{t+1}$th effective iteration, each contributing at most one edge to the set $\{R^{\ell}\}_{\ell \in [k]}$.

Across all processors, this deferral occurs in at most $\bigO{n \log n}$ iterations, corresponding to the total number of edges added to the stacks. Therefore, the total number of edges stored across all $R^{\ell}$ is $\bigO{k\cdot n\log n }$.
\end{proof}

The streaming phase (Step~5 of \hyperref[fig:algo_psmwm_ds]{PS-MWM-DS}) takes $\bigO{\lmax}$ time. \aref{PS-MWM}{fig:algo_psmwm} is identical to Steps~1-3 and Steps~5-8 of \hyperref[fig:algo_psmwm_ds]{PS-MWM-DS}, if we treat the $\{R^{\ell}\}_{\ell \in [k]}$ data structures as edge streams. 

Since $|R^{\ell}| \leq |E^{\ell}| \leq \lmax$ for all $\ell \in [k]$, \hyperref[fig:algo_psmwm_ds]{PS-MWM-DS} achieves the runtime bound stated in \lemref{lemma_algo1_time_ns}.

The substeps of Step~3 in \hyperref[fig:algo_proc_edge]{Process-Edge} are identical to those of Step~4 in \hyperref[fig:algo_proc_edge_ds]{Process-Edge-DS}. 
Therefore, once all processors complete Step~5 of \aref{PS-MWM-DS}{fig:algo_psmwm_ds}, the variables $\{\alpha_u\}_{u \in V}$, scaled by $(1+\epsilon)$, form a feasible solution to the dual LP in \fref{fig:lp}, as in \hyref{Proposition}{prop_algo1_dual}.
It follows that \aref{PS-MWM-DS}{fig:algo_psmwm_ds} achieves the approximation ratio stated in \lemref{lemma_algo1_approx}.

\subsection{Polylogarithmic Runtime}
\label{subsec:plog_time}

The runtime of \aref{Algorithm PS-MWM-DS}{fig:algo_psmwm_ds} is $\bigOtilde{\lmax + n}$, which is optimal up to polylogarithmic factors when $\lmax = \bigOmega{n}$. If $\lmax =\smallO{n}$, one may ask whether this can be improved to $\bigOtilde{\lmax + n/k}$. In such cases, the total number of edges satisfies $m=\smallO{k \cdot n}$, so any offline algorithm with $\bigOtilde{k \cdot n}$ space would suffice. However, in streaming settings, the value of $\lmax$ may not be known a priori, making such instances indistinguishable from the general case. We now outline a modification of \aref{PS-MWM-DS}{fig:algo_psmwm_ds} that runs in $\bigOtilde{\lmax + n/k}$ time, yielding a polylogarithmic runtime for sufficiently large~$k$.

The preprocessing and streaming phases of \aref{PS-MWM-DS}{fig:algo_psmwm_ds} are fully parallelizable: for sufficiently large $k$, Steps~1--4 can be completed in $\bigO{1}$ time.
As shown in \hyref{Appendix}{subsec:mwm_ds}, the edge set \[\D :=\underset{\ell \in [k]}{\bigcup} \{R^{\ell} \cup S^{\ell}\}\] contains the edges of a $\br{2+\epsilon}$-approximate MWM in the input graph. 
Therefore, it suffices to design an algorithm that computes a near-optimal MWM on $\D$ in $\bigOtilde{\lmax+n/k}$ time.
Although it is possible to achieve this without loss in approximation, for example by using the algorithm in Corollary~1.2 of~\cite{huang20231}), the design of such algorithms is intricate.
Instead, we present a simpler algorithm that runs in $\bigOtilde{\lmax+n/k}$ time and incurs a factor-of-two loss in the approximation guarantee.

\begin{figure}[h]
    \centering
    \begin{parbox}{5.0in}{    
        \begin{mdframed}[linewidth=0.5pt, roundcorner=7pt, backgroundcolor=gray!5, frametitle={\underline{PS-MWM-PR$(V, \ell, \epsilon)$}}]
        \begin{enumerate}
        \item In parallel initialize $lock_u$, set $\alpha_u$ to $0$, and $M(u)$ to $\emptyset$ for all $u \in V$       
        \item set $S^{\ell}, R^{\ell} \leftarrow \emptyset$
        \item for each edge $e=\{u,v\}$ in $\ell${th} stream do
        \begin{enumerate}
            \item \aref{Process-Edge-DS$(e, S^{\ell}, R^{\ell}, \epsilon)$}{fig:algo_proc_edge_ds}
        \end{enumerate}        
        \item $R^{\ell} \leftarrow R^{\ell} \cup S^{\ell}$
        \item  for $t=1$ to $8 \ln\frac{2}{\epsilon}$ do
        \begin{enumerate}
            \item let $R^{\br{\ell, t}} := R^{\ell} \setminus \{e=\{u,v\} \in R^{\ell} \mid M(u)=v \text{ and } M(v)=u\}$
            \item for each edge $e \in  R^{\br{\ell, t}} $, compute its associated weight \[w^{\prime}_e := w_e - w_{e_u} - w_{e_v},\] where $e_x := \{x, M(x)\}$ and $w_{e_x} = 0$ if $M(x)=\emptyset$
            \item let $R^{\br{\ell, t}}$ now denote the set of edges $e$ with weights $w^{\prime}_e$
            \item $\M^{\br{\ell,t}} \leftarrow$ \aref{Reduce-To-Maximal$(V, \ell, \epsilon, R^{\br{\ell, t}})$}{fig:algo_mwm_reduce}
            \item for each edge $\{u,v\} \in \M^{\br{\ell,t}}$ do
            \begin{enumerate}                
                \item  \aref{Augment-Matching$(u,v)$}{fig:algo_mwm_augment}              
            \end{enumerate}                        
        \end{enumerate}
        \item return $\M^{\ell} := \{e=\{u,v\} \in R^{\ell} \mid M(u)=v \text{ and } M(v)=u\}$
    \end{enumerate}
        \end{mdframed}    
    
    }
    \end{parbox}    
    \caption{A modification of \aref{Algorithm PS-MWM-DS}{fig:algo_psmwm_ds}.}
    \label{fig:algo_psmwm_pr}
\end{figure}

\begin{figure}[h]
    \centering
    \begin{parbox}{5.0in}{    
        \begin{mdframed}[linewidth=0.5pt, roundcorner=7pt, backgroundcolor=gray!5, frametitle={\underline{Augment-Matching$(u, v)$}}]
        \textcolor{teal}{/* Assumes access to global variables $\{M(z)\}_{z \in V}$ 
 and locks $\{lock_z\}_{z \in V}$ */}
        \begin{enumerate}
            \item for each $x \in \{u, v\}$ do
            \begin{enumerate}
                \item let $y := M(x)$
                \item if $y = \emptyset$ then set $M(x) \gets \{u,v\} \setminus \{x\}$ and skip to the next $x$
                \item acquire $lock_x$ and $lock_y$ in lexicographic order of $x$ and $y$
                \item if $M(y) = x$ then set $M(y) \gets \emptyset$
                \item set $M(x) \gets \{u,v\} \setminus \{x\}$
                \item release $lock_x$ and $lock_y$
            \end{enumerate}
        \end{enumerate}
        \end{mdframed}
    }
    \end{parbox}
    \caption{A subroutine used in \aref{PS-MWM-PR}{fig:algo_psmwm_pr}.}
    \label{fig:algo_mwm_augment}
\end{figure}

\begin{figure}[h]
    \centering
    \begin{parbox}{5.2in}{    
        \begin{mdframed}[linewidth=0.5pt, roundcorner=7pt, backgroundcolor=gray!5, frametitle={\underline{Reduce-To-Maximal$(V, \ell, \epsilon, \A^{\ell})$}}]
        \begin{enumerate}
        \item In parallel, compute $W := \max\{w_e \mid e \in \bigcup_{\ell \in [k]} \A^{\ell}\}$
        \item In parallel, set $mark_u$ to $0$ for all $u \in V$
        \item Set $\Mt^{\ell} \leftarrow \emptyset$        
        \item for $r:= \lfloor \log_{1+\epsilon} W \rfloor$ down to $1$ do
        \begin{enumerate}
            \item Let $\B^{\ell} := \{e \in \A^{\ell} \mid w_e \geq \br{1+\epsilon}^r\}$
            \item In parallel, compute a maximal matching $\Mt^r$ in $G^r := \br{V, \bigcup_{\ell \in [k] } \B^{\ell}}$, and let $\Mt^{\br{\ell,r}} := \Mt^r \cap \B^{\ell}$
            \item for each $e=\{u,v\} \in \Mt^{\br{\ell,r}}$, if both $mark_u$ and $mark_v$ are set to $0$, then include  $e$ in $\Mt^{\ell}$ and set both $mark_u$ and $mark_v$ to $1$
        \end{enumerate}                
        \item return $\Mt^{\ell}$
    \end{enumerate}
        \end{mdframed}    
    
    }
    \end{parbox}    
    \caption{A subroutine used in \aref{PS-MWM-PR}{fig:algo_psmwm_pr} that approximates MWM via maximal matching.}
    \label{fig:algo_mwm_reduce}
\end{figure}

Modifications to \aref{PS-MWM-DS}{fig:algo_psmwm_ds} are outlined in \aref{Algorithm PS-MWM-PR}{fig:algo_psmwm_pr}. The new algorithm introduces two additional subroutines, \hyperref[fig:algo_mwm_reduce]{Reduce-To-Maximal} and \hyperref[fig:algo_mwm_augment]{Augment-Matching}, invoked in Steps~5(d) and 5(e), respectively, while retaining \hyperref[fig:algo_proc_edge_ds]{Process-Edge-DS} in Step~3. In this setup, Steps~1--2 constitute the preprocessing phase, Step~3 is the streaming phase, and Steps~4--6 comprise the post-processing phase.

Step~5 of \aref{PS-MWM-PR}{fig:algo_psmwm_pr} implements an adaptation of the $\br{2+\epsilon}$-approximate MWM algorithm from~\cite{lotker2015improved}. 
That algorithm invokes a black-box $\delta$-approximate MWM subroutine, instantiated with $\delta=5$ using the $\br{4+\epsilon}$-approximate algorithm of~\cite{lotker2007distributed}.
In our version, we replace this component with a simpler $\br{4+\epsilon}$-approximate algorithm, used as a subroutine in Step~5(d).

In each iteration of Step~5 in \aref{PS-MWM-PR}{fig:algo_psmwm_pr}, processor $\ell$ computes the gain $w^{\prime}_e$ for each non-matching edge $e \in R^{\ell}$, representing the weight improvement if $e$ replaces its incident matched edges. (Note that $w^{\prime}_e$ could be negative.) 
The processor then invokes the subroutine \hyperref[fig:algo_mwm_reduce]{Reduce-To-Maximal} using $w^{\prime}$ as the weight function.
These concurrent calls across all processors collectively yield a $\br{4+\epsilon}$-approximate MWM on the non-matching edges in $\D$.
In Step~5(e), processor $\ell$ augments the current matching using the result of its respective call.

The subroutine \hyperref[fig:algo_mwm_reduce]{Reduce-To-Maximal} is a standalone parallel algorithm for approximating an MWM, adapted from the sequential streaming algorithm of~\cite{crouch2014improved}. It defines $\bigO{\log_{1+\epsilon}W} = \bigO{\br{\log n}/\epsilon}$ geometrically decreasing thresholds for weight classes and assigns each edge to every class whose threshold it meets (Step~4(a)).
Processors concurrently invoke a maximal matching algorithm for each class, in decreasing order of thresholds (Step~4(b)). 
The resulting edges from each maximal matching are vertex-disjoint, so the current matching can be augmented concurrently using these edges (Step~4(c)).

The matching $\M:=\bigcup_{\ell \in [k]} \M^{\ell}$ returned by \aref{PS-MWM-PR}{fig:algo_psmwm_pr} is a $\br{4+\epsilon}$-approximate MMW in the input graph. This follows from the fact that \hyperref[fig:algo_mwm_reduce]{Reduce-To-Maximal} returns a $\br{4+\epsilon}$-approximate MWM, by an argument identical to that of Lemma~6 in~\cite{crouch2014improved}. Given this, Step~5 of \aref{PS-MWM-PR}{fig:algo_psmwm_pr} computes a $\br{2+\epsilon}$-approximate MWM on the edge set $\D$, using the same reasoning as in Theorem~4.5 in~\cite{lotker2015improved}. 
Since $\D$ contains the edges of a $\br{2+\epsilon}$-approximate MWM in the original graph, the final matching $\M$ inherits the $\br{4+\epsilon}$ approximation guarantee.

The subroutine \hyperref[fig:algo_mwm_reduce]{Reduce-To-Maximal} runs in $\bigOtilde{\lmax+n/k}$ time with high probability (w.h.p.). This holds because, for any constant $\epsilon > 0$, Step~4 performs $\bigO{\log n}$ iterations, and in each iteration, Step~4(b) runs in $\bigO{\lmax \cdot \log n + n/k}$ time w.h.p., by an argument identical to that of Lemma~3.8 in~\cite{lotker2007distributed}. 
Each iteration of Step~5(e) in \aref{PS-MWM-PR}{fig:algo_psmwm_pr} takes constant time, since any edge in the current matching can intersect with at most two edges in $\bigcup_{\ell \in [k]} \M^{\br{\ell, t}}$.
Because \hyperref[fig:algo_mwm_reduce]{Reduce-To-Maximal} is invoked a constant number of times in \aref{PS-MWM-PR}{fig:algo_psmwm_pr}, the total runtime of the algorithm is $\bigOtilde{\lmax+n/k}$ w.h.p.

\subsection{Proof of Theorem~\ref{thm_psmwm_ld}}
\label{subsec:thm_hd}

We extend the analysis of \hyperref[fig:algo_psmwm]{PS-MWM} to analyze \hyperref[fig:algo_psmwm]{PS-MWM-LD}. 
To do so, we modify the definition of superstep (\hyref{Definition}{def_ss}), referring to the modified notion as an \emph{ld-superstep} and denote the corresponding graph $G^{\br{t}}$ as $G_{L}^{\br{t}}$.

For a given edge, if the execution of \hyperref[fig:algo_proc_edge_ld]{Process-Edge-LD} does not encounter any contention, that is, each loop it executes (specifically, those in Step~2, Step~4, and Step~2 of the call to \hyperref[fig:algo_proc_edge]{Process-Edge}) is run at most once, then the processor is said to take one \emph{ld-superstep}. Each additional iteration of any of these loops, if executed, increases the processor's ld-superstep count by one.

\begin{lemma}
\label{lemma_free_ld}    
\aref{Algorithm PS-MWM-LD}{fig:algo_psmwm_ld} is free from deadlock, livelock, and starvation.
\end{lemma}
\begin{proof}
Only the delegates execute Steps~4--6 of \hyperref[fig:algo_proc_edge_ld]{Process-Edge-LD}.

For each component of $G_{L}^{\br{t}}$ in which a delegate participates in Step~5 of \hyperref[fig:algo_proc_edge_ld]{Process-Edge-LD}, \lemref{lemma_free} ensures that the concurrent executions of this step is free from deadlock, livelock, and starvation.

Building on this fact, we apply the argument from \lemref{lemma_free} to the components of $G_{L}^{\br{t}}$ involving the processors in a given group $j \in [r]$. This implies that Steps~1--6, when executed concurrently by the processors in group $j$, are also free from deadlock, livelock, and starvation.

Since the dependencies in Steps~1--4 are confined within each group, and the cross-group dependencies in Steps~5--6 are resolved by the delegates, the full execution of \hyperref[fig:algo_proc_edge_ld]{Process-Edge-LD} across all groups proceeds without deadlock, livelock, or starvation. 

By \lemref{lemma_free_pp}, Step~7 of \aref{PS-MWM-LD}{fig:algo_psmwm_ld} is likewise free from deadlock, livelock, and starvation.
\end{proof}

\begin{lemma}
\label{lemma_algo1_space_ld}
For any constant $\epsilon > 0$, per-edge processing time and space complexity of \aref{PS-MWM-LD}{fig:algo_psmwm_ld} are $\bigO{n \log n}$ and $\bigO{k + r \cdot n + n\log n}$, respectively.
\end{lemma}
\begin{proof}
The space bound follows from \lemref{lemma_algo1_space_ns}, with an additional $r\cdot n$ term accounting for the $r$ local copies of dual variables and the corresponding locks.

The per-edge processing time of \aref{PS-MWM-LD}{fig:algo_psmwm_ld} is the time spent in the subroutine \hyperref[fig:algo_proc_edge_ld]{Process-Edge-LD}.

Now, for $k=r$, the claim follows directly from \lemref{lemma_algo1_space_ns}, since Step~2 and Step~4 require at most one ld-superstep.

For $k>r$, consider an edge $e_{\ell}=\{u_{\ell}, v_{\ell}\} \in \bigcap_{t \in [a,b]} E_{L}^{\br{t}}$. Each iteration $t \in (a,b]$ in which processor $\ell$ executes Step~2 or Step~4 of \hyperref[fig:algo_proc_edge_ld]{Process-Edge-LD} is attributable to a delegate that was active in the $\br{t-1}$st effective iteration. These delegates execute Steps~4 and 5. If the delegates associated with $u_{\ell}$ or $v_{\ell}$ are executing only Step~4, then some delegate for a different vertex must be executing Step~5.

Since at most $\bigO{n \log n}$ edges are included in the stacks, the number of ld-supersteps during which delegates execute Step~5 is also $\bigO{n \log n}$. 

Each processor $\ell$ that executes Step~2 follows one of two execution paths: it either becomes a delegate or returns through Step~3. In both cases, the updates from Step~6 are propagated to processor $\ell$ within a single ld-superstep. Hence, in either path, processing the edge $e_{\ell}$ requires at most $\bigO{n \log n}$ ld-supersteps.
\end{proof}

By using $G_{L}^{\br{t}}$ in place of $G^{\br{t}}$ in the proof of \lemref{lemma_algo1_space_ns}, and modifying the argument to account for Step~2 and Step~4 of \hyperref[fig:algo_proc_edge_ld]{Process-Edge-LD}, we can show that processor $\ell$ takes $\bigO{|E^{\ell}| + n \log n}$ ld-supersteps to process $|E^{\ell}|$ edges. Therefore, \aref{PS-MWM-LD}{fig:algo_psmwm_ld} achieves the amortized per-edge processing time from \lemref{lemma_algo1_space_ns} and runtime from \lemref{lemma_algo1_time_ns}.

During the execution of Step~5 of \aref{PS-MWM-LD}{fig:algo_psmwm_ld}, the following invariant is maintained: for all $u \in V$ and for all $j \in [r]$, we have $\alpha_u \geq \alpha_u^j$. In Step~3 of \hyperref[fig:algo_proc_edge]{Process-Edge}, when an edge $e=\{u,v\}$ is added to a stack, we have $\alpha_u = \alpha_u^j$ and $\alpha_v = \alpha_v^j$, and for all $i \in [r]\backslash \{j\}$, it holds that $\alpha_u > \alpha_u^i$ and $\alpha_v > \alpha_v^i$. Step~3(b) of \hyperref[fig:algo_proc_edge]{Process-Edge} is the only step that updates the global dual variables. Local copies are only synchronized with their global counterpart in Step~6 \hyperref[fig:algo_proc_edge_ld]{Process-Edge-LD}, which maintains the invariant.

With the preceding invariant, it follows that any edge $e=\{u,v\}$ not included in a stack satisfies \[\br{1+\epsilon}\br{\alpha_u + \alpha_v} \geq \br{1+\epsilon}\br{\alpha_u^j + \alpha_v^j} \geq w_e.\] 
The edges included in the stacks also satisfy the dual constraint, as in \hyref{Proposition}{prop_algo1_dual}. 
Therefore, after all processors complete Steps~1-5 of \hyperref[fig:algo_psmwm_ld]{PS-MWM-LD}, the variables $\{\alpha_u\}_{u \in V}$, scaled by $(1+\epsilon)$, form a feasible solution of the dual LP in \fref{fig:lp}. Hence, \aref{PS-MWM-LD}{fig:algo_psmwm_ld} achieves the approximation ratio stated in \lemref{lemma_algo1_approx}.

\LemmaLD*
\begin{proof}
All processors execute the subroutine \hyperref[fig:algo_proc_edge_ld]{Process-Edge-LD} during the streaming phase. For each group $j \in [r]$, the updates in Step~6 of the subroutine increase the dual variables $\alpha_u^j$ and $\alpha_v^j$ by at least $\br{1+\epsilon}$. 
Thus, for each vertex $u$, the number of times some delegate executes Step~6 for any $\alpha_u^j$ is at most $r \cdot \log_{1+\epsilon} W = \bigO{r\cdot \log n}$. 
After this point, Step~1 and Step~3 ensure that no more delegates are created for $\alpha_u^j$. Summing over all $n$ vertices, we obtain a total $\bigO{r\cdot n\log n}$ updates in Step~6 where global variables are accessed.

At most $\bigO{n \log n}$ edges are included in the stacks, so the number of ld-supersteps during which delegates participate in Step~5 of \hyperref[fig:algo_proc_edge_ld]{Process-Edge-LD} is $\bigO{n \log n}$. Since at most $r$ delegates participate in each such superstep, the total number of times global variables are accessed in Steps~1-4 of \hyperref[fig:algo_proc_edge]{Process-Edge} is $\bigO{r \cdot n \log n}$.
\end{proof}

By \lemref{lemma_psmwm_ld}, Steps~1-6 of \aref{PS-MWM-LD}{fig:algo_psmwm_ld} access the global variables a total of $\bigO{r\cdot n \log n}$ times. 

In the post-processing phase, with $k$ processors executing Step~7 of \aref{PS-MWM-LD}{fig:algo_psmwm_ld} concurrently, we can only ensure that the total number of accesses to global variables is bounded by $\bigO{k \cdot n\log n}$.

For $k>r$, instead of maintaining one stack per processor, we can maintain a single stack per group, similar to a group lock, and allow all processors within a group to share their group's stack. A designated delegate from each group then participates in the execution of Step~7.

This modification ensures that the total number of accesses to global variables throughout the entire execution of \aref{PS-MWM-LD}{fig:algo_psmwm_ld} remains $\bigO{r \cdot n\log n}$, while preserving the bounds on other metrics. 

\hyref{Theorem}{thm_psmwm_ld} now follows from the preceding analysis and its extension to the \emph{deferrable strategy}.

\subsection{Distributed Implementations}
\label{subsec:dist}

Tightly coupled distributed-memory multiprocessors can be viewed as a generalization of NUMA architectures in terms of memory hierarchy. Consequently, memory-efficient algorithms for hierarchical architectures such as NUMA can be interpreted as communication-efficient algorithms for tightly coupled distributed-memory systems. This correspondence is especially clear in distributed architectures that support remote memory access. In systems based on explicit message passing (e.g., send/receive), remote memory access can be emulated by assigning processors to mediate access to shared locations via messages.

\phantomsection\label{para:ps_mwm_dm} 
From \hyref{Theorem}{thm_psmwm_ld}, we therefore obtain a \emph{single-pass} distributed streaming algorithm for computing a $\br{2+\epsilon}$-approximate MWM. 
For $r=k$, let \aref{PS-MWM-DM}{para:ps_mwm_dm} denote such a distributed implementation on a cluster with $r$ nodes, for example by implementing the deferrable strategy in a manner similar to \aref{PS-MWM-LD}{fig:algo_psmwm_ld}.
By \hyref{Theorem}{thm_psmwm_ld}, the total number of remote memory accesses by \aref{PS-MWM-DM}{para:ps_mwm_dm} is $\bigOtilde{r\cdot n}$; that is, its communication cost is $\bigOtilde{r\cdot n}$ bits. The algorithm runs in $\bigOtilde{\lmax + n}$ time and uses $\bigOtilde{n}$ space per node. (For $k>r$, we can use the non-deferrable strategy with the same performance guarantees.)

We now compare \aref{PS-MWM-DM}{para:ps_mwm_dm} with several distributed algorithms and highlight its advantages.
Multiple MPC algorithms have been developed for the MWM problem (see~\cite{assadi2019distributed, gamlath2019weighted} and the references therein).
Those with $\bigOtilde{n}$ space per node require a large number of rounds (see Table~1 in~\cite{assadi2019distributed}). 
Even if each round is treated as equivalent to a single pass over the input, these algorithms require significantly more passes than our single-pass algorithm.
Under this comparison, the only algorithm that comes close to matching \aref{PS-MWM-DM}{para:ps_mwm_dm} is the one by~\cite{assadi2019distributed}, which requires two rounds of computation and uses $\bigO{\sqrt{m/n}}$ machines, each with $\bigO{\sqrt{mn}}$ space.

\phantomsection\label{para:coreset_dm}
Let \aref{CORESET-DM}{para:coreset_dm} denote the following implementation of the algorithm from~\cite{assadi2019distributed} (see the paper for details). Distribute the edges by sending each edge to a constant number of nodes chosen uniformly at random. Then, each node runs the greedy algorithm on its local edge set; that is, it repeatedly selects the heaviest edge compatible with the current matching. The resulting matchings from all nodes are then sent to a single node, which runs the greedy algorithm again on the union of these edges. This yields a $\br{3+\epsilon}$-approximation in expectation.

Note that there exists another implementation of the algorithm from~\cite{assadi2019distributed} that, in expectation, achieves a $\br{2+\epsilon}$\allowbreak-approximation guarantee, but it is not comparable to \aref{PS-MWM-DM}{para:ps_mwm_dm} in terms of implementation complexity. This variant requires computing a near-optimal matching during post-processing, which involves intricate algorithms that may not be amenable to efficient implementations in practice.

In \aref{PS-MWM-DM}{para:ps_mwm_dm}, the edges can be deterministically distributed evenly across the nodes. 
By setting $r=\bigO{\sqrt{m/n}}$, we obtain the comparisons shown in \hyref{Table}{tab:dist}.

\begin{table}
    \centering
    \caption{Comparison of the distributed algorithms \aref{PS-MWM-DM}{para:ps_mwm_dm} and \aref{CORESET-DM}{para:coreset_dm}.}
    \begin{tabular}{lll}        
        \toprule        
        Metric & \aref{PS-MWM-DM}{para:ps_mwm_dm} & \aref{CORESET-DM}{para:coreset_dm} \\
        \midrule
        Streaming support  & Yes & No \\
        Space per node & $\bigOtilde{n}$ & $\bigO{n\sqrt{n}}$  \\        
        Approximation ratio & $2+\epsilon$ (worst-case) &  $3+\epsilon$ (expected) \\
        Computation time & $\bigOtilde{\sqrt{mn}+n}$ & $\bigOtilde{\sqrt{mn}+n}$  \\        
        Communication cost & $\bigOtilde{\sqrt{mn}}$ & $\bigOtilde{\sqrt{mn}}$ \\        
        \bottomrule
    \end{tabular}
    \label{tab:dist}
\end{table}

From \hyref{Table}{tab:dist}, \aref{PS-MWM-DM}{para:ps_mwm_dm} uses $\bigOtilde{n}$ space per node, whereas \aref{CORESET-DM}{para:coreset_dm} uses $\bigO{n\sqrt{n}}$ space per node.
Both algorithms require the same amount of computation, but \aref{PS-MWM-DM}{para:ps_mwm_dm} achieves a $\br{2+\epsilon}$-approximation guarantee in the worst case, while \aref{CORESET-DM}{para:coreset_dm} achieves a $\br{3+\epsilon}$-approximation guarantee in expectation.

In \hyref{Table}{tab:dist}, the \emph{communication cost} refers to the total number of bits communicated by all nodes during the execution of an algorithm. 
Although both algorithms achieve optimal communication cost (up to logarithmic factors)~\cite{huang2015communication}, the nature of communication differs. In \aref{PS-MWM-DM}{para:ps_mwm_dm}, communication is distributed across nodes throughout the execution.
In contrast, \aref{CORESET-DM}{para:coreset_dm} requires $\bigOtilde{\sqrt{mn}}$ bits to be sent to a single node, creating a potential bottleneck.

In \aref{PS-MWM-DM}{para:ps_mwm_dm}, the number of nodes $r$ is an adaptable parameter, independent of the number of edges in the graph. 
In contrast, reducing the number of nodes in \aref{CORESET-DM}{para:coreset_dm} necessitates a proportional increase in space per node to accommodate the $\Omega\br{m}$ total edges, since each edge is sent to a constant number of nodes. 
As a result, if the cluster lacks sufficient memory to collectively store these $\Omega\br{m}$ edges, \aref{CORESET-DM}{para:coreset_dm} becomes infeasible.
\aref{PS-MWM-DM}{para:ps_mwm_dm} does not face this limitation.

A coreset-based sequential streaming algorithm can compute a $\br{3+\epsilon}$-approximate MWM in a single pass but requires random edge arrival and $\bigOtilde{n\sqrt{n}}$ space~\cite{assadi2019coresets}. Even if one could afford $\bigOtilde{n\sqrt{n}}$ space per node, the random edge arrival assumption is fundamentally limiting in the poly-streaming setting, which allows arbitrary distribution of data across streams.

These comparisons underscore the advantages of \aref{PS-MWM-DM}{para:ps_mwm_dm} for distributed streaming computation. To the best of our knowledge, it is the first \emph{single-pass} distributed streaming algorithm for approximating a maximum weight matching.

We note that the instance sizes reported in the empirical study of~\cite{assadi2019distributed} appear to be inconsistent with publicly available data. The authors claim to have evaluated their algorithm on graphs with over 500 billion edges, attributing the largest instance to the publicly available Friendster graph from the SNAP dataset. However, the SNAP version of this graph contains fewer than two billion edges, and no publicly available variant is known that matches the reported size. As such, the scalability claims made in that study require further verification. In light of this discrepancy, we conjecture that our distributed streaming algorithm would outperform the algorithm of~\cite{assadi2019distributed} on truly massive graphs.

\subsection{Further Use of the Dual Formulation}
\label{subsec:app_duals}

We previously used the dual formulation of the MWM problem (\fref{fig:lp}) to establish the approximation guarantee of our algorithm (\hyref{Section}{subsec:mwm_analyses}) and to extend it to other settings (\hyref{Section}{sec:numa}).
We now consider two further applications: assessing solution quality and designing alternative algorithms.

We describe several dual update rules that produce feasible dual solutions, yielding empirical upper bounds on the weight of an MWM. These rules are also useful for improving solution quality in practice.

Recall the dual LP listed in \fref{fig:lp}. 
For a graph $G=(V,E)$, let $\M^{*}$ be an MWM, and let $\{y_u\}_{u \in V}$ be any feasible dual solution.
By LP duality, we have $w\br{\M^{*}} \leq \sum_{u \in V} y_u$, so the dual objective provides an upper bound on the optimal matching weight. 

We used this fact in our approximation analysis (\lemref{lemma_algo1_approx}). By \hyref{Proposition}{prop_algo1_dual}, the dual solution was $\{\br{1+\epsilon}\alpha_u\}_{u \in V}$. 
This solution was generated by the dual update rule used in Steps~3(a)--(b) of~\hyperref[fig:algo_proc_edge]{Process-Edge} (and Steps~4(a)--(b) of~\hyperref[fig:algo_proc_edge_ds]{Process-Edge-DS}). 
The rule is rooted in the local-ratio technique~\cite{bar2004local}, a general approximation framework adapted to matching by~\cite{paz20182+}.

We now explore additional dual update rules to empirically assess solution quality and consider implications of these rules.

We consider five simple dual update rules, each producing a feasible dual solution and thus an upper bound on $w\br{\M^{*}}$.
These rules are independent of any matching algorithm and can be used individually and jointly to derive tighter instance-specific upper bounds. The general procedure is as follows: initialize all dual variables $\{y_u\}_{u \in V}$ to zero. 
For any edge $e=\{u, v\}$, if $w_e \leq y_u + y_v$, then do nothing. 
Otherwise, compute $\delta_e \leftarrow w_e - \br{y_u + y_v}$, and update the dual variables using one of the rules in \hyref{Table}{tab:dual_rules}.

\begin{table}[h]
    \centering
    \caption{For an edge $e=\{u,v\}$, if $w_e > y_u + y_v$ then compute $\delta_e \leftarrow w_e - \br{y_u + y_v}$, and apply one of the following rules.}
    \begin{tabular}{ll}
    \toprule
    Identifier          &  Dual Update Rule\\
    \toprule
    \multirow{2}{*}{\emph{UniRelaxed}} & $y_u \leftarrow y_u + \delta_e$\\
                   & $y_v \leftarrow y_v + \delta_e$\\
    \midrule                   
    \multirow{2}{*}{\emph{UniTight}} & $y_u \leftarrow y_u + \delta_e/2$\\
                   & $y_v \leftarrow y_v + \delta_e/2$\\
    \midrule  
    \multirow{2}{*}{\emph{ArgMax}} & $x \leftarrow argmax_{\{u,v\}} \{y_u, y_v\}$\\
                   & $y_x \leftarrow y_x + \delta_e$ \\
    \midrule                  
    \multirow{2}{*}{\emph{ArgMin}} & $x \leftarrow argmin_{\{u,v\}} \{y_u, y_v\}$\\
                   & $y_x \leftarrow y_x + \delta_e$ \\
    \midrule
    \multirow{2}{*}{\emph{ArgRand}} & pick $x \in \{u,v\}$ uniformly at random \\
                   & $y_x \leftarrow y_x + \delta_e$ \\        
    \bottomrule
    \end{tabular}    
    \label{tab:dual_rules}
\end{table}

These rules are directly applicable in streaming settings.
One may apply multiple rules to the same instance or design additional variants, for example, selecting vertices based on degree (static or dynamic) instead of uniformly at random.

Each rule in \hyref{Table}{tab:dual_rules} yields a feasible dual solution satisfying the constraints in \hyref{Figure}{fig:lp}.
For a graph $G=(V,E)$, we compute the dual objective $Y=\sum_{u \in V} y_u$ for each rule and take $Y_{\min}$ as the smallest among them. Comparing $w\br{\M}$ with $Y_{\min}$ gives an empirical estimate of approximation quality without computing the true optimum. If $Y_{\min} \approx w\br{\M^{*}}$, then $Y_{\min}$ serves as a tight a posteriori upper bound on optimality.

Rule UniRelaxed mirrors the update rule used in \hyperref[fig:algo_proc_edge]{Process-Edge} (and \hyperref[fig:algo_proc_edge_ds]{Process-Edge-DS}), except that the $\epsilon$ factor is omitted. Rules UniRelaxed and UniTight distribute $\delta_e$ uniformly among endpoints, with UniTight doing so more conservatively. The remaining rules exploit structural asymmetries and edge orderings.

As shown in \hyref{Section}{sec:evals}, these rules reveal that our algorithm often produces matchings of significantly better quality than its worst-case approximation ratio suggests.

A natural question is whether these rules, like UniRelaxed, can lead to useful algorithms. 
The answer is yes: any of them can replace UniRelaxed in the \hyperref[fig:algo_proc_edge]{Process-Edge} (and \hyperref[fig:algo_proc_edge_ds]{Process-Edge-DS}) subroutine,
but only UniTight yields a provable approximation guarantee.

In the post-processing phase, we previously used edge gain values $g_e$ to reconstruct edge orderings. However, the analysis does not depend on the absolute values of $g_e$. An equivalent approach is to use auxiliary variables $\{z_u\}_{u \in V}$, initialized to zero. For a fixed constant $c_z > 0$, whenever an edge updates the duals, increment $z_u$ and $z_v$ by $c_z$ and store $z_e=z_u + z_v$ instead of $g_e$. An edge is considered \emph{tight} if $z_e=z_u+z_v$, and $\alpha_x$ and $g_e$ can be replaced by $z_x$ and $c_z$, respectively.

For any rule in \hyref{Table}{tab:dual_rules}, the number of updates to each $\alpha_u$ is at most $\log_{1+\epsilon}\br{nW} = \bigO{\frac{\log n}{\epsilon}}$, where $W = \bigO{poly(n)}$ is the normalized maximum edge weight. Hence, for any constant $\epsilon > 0$, the total number of dual updates is bounded by $\bigO{n \log n}$. Combined with the alternate edge-ordering approach, this ensures that the bounds in \lemref{lemma_algo1_space_ns}, \lemref{lemma_algo1_time_ns}, and \lemref{lemma_algo1_space_ds} hold for all listed rules.

Rule UniTight achieves the same approximation guarantee as UniRelaxed, as can be confirmed by adapting the proofs of \hyref{Proposition}{prop_algo1_aux} and \hyref{Proposition}{prop_algo1_dual}. The remaining rules do not yield provable approximation bounds, but often lead to better solutions in practice (\hyref{Section}{sec:evals}).

This raises the question of whether simple dual update rules can be extended to achieve stronger approximation guarantees, particularly in multi-pass settings. Substantial improvements in approximation are known to require multiple passes~\cite{kapralov2021space}. Although recent work has simplified the design of such algorithms (see~\cite{assadi2024simple}) and these techniques can be adapted to the poly-streaming setting, they remain well beyond the level of simplicity needed for efficient implementation in practice.

A recurring theme in the literature is that multiple rounds of adaptive computation using a simple algorithm can substantially amplify its approximation guarantee.
For example, \hyref{Appendix}{subsec:plog_time} demonstrates how a $\br{4+\epsilon}$-approximation can be amplified to $\br{2+\epsilon}$.
Similarly,~\cite{mcgregor2005finding} showed that a $6$-approximation can be improved to $\br{2+\epsilon}$ through adaptive passes.
Exploring such amplification effects within the framework of simple dual update rules remains a promising direction for future research.

%% file: 9_app3.tex
\section{Deferred Empirical Details}
\label{sec:app_evals}

\subsection{Detailed Datasets}
\label{subsec:app_datasets}

\hyref{Table}{tab:datasets_p1} and \hyref{Table}{tab:datasets_p2} describe the 46 graphs in our datasets. Two graphs appear in both the ER graphs and the ER-dv graphs (listed at the bottom of the tables).
In \hyref{Table}{tab:datasets_p1}, we refer to the first eight graphs, middle eight, and last eight as \emph{the SSW graphs}, \emph{the BA graphs}, and \emph{the ER graphs}, respectively. 
In \hyref{Table}{tab:datasets_p2}, the first eight, the middle eight, and the last eight are referred to as \emph{the UA-dv graphs}, \emph{the UA graphs}, and \emph{the ER-dv graphs}, respectively.

The first six graphs in the SSW graphs are the largest from \href{https://sparse.tamu.edu/}{SuiteSparse Matrix Collection} \cite{davis2011university}, while the last two are the largest publicly available graphs from \href{https://webdatacommons.org/hyperlinkgraph/index.html}{Web Data Commons} \cite{meusel2015graph}. 
The BA graphs are generated using the \emph{Barabási–Albert (BA) model} \cite{albert2002statistical}. 
The UA graphs are generated using the \emph{uniform attachment (UA) model} \cite{pekoz2013total}, which retains only the growth component of the BA model. 
The ER graphs and the ER-dv graphs are generated using the $G(n,p)$ variant of the \Erdos–\Renyi \ random graph model \cite{erdos1960evolution}.
The \emph{-dv} variants were included to examine how varying graph density affects algorithmic performance.

For the SSW graphs, the reported number of vertices excludes isolated vertices, and the number of edges excludes self-loops. 
For the BA graphs and the UA graphs, $x$ in $BA\_x$ or $UA\_x$ denotes the number of edges added for each new vertex. 
In $UA\_x\_y$, $x$ has the same meaning, and $y$ indicates the number of vertices in millions. The initial seed graphs used to generate the BA graphs, the UA graphs, and the UA-dv graphs are sampled from $G(n,p)$ with $n=262{,}144$ and $p=0.01$. 
Graphs labeled $ERx\_y$ have density $p=1/x$ and contain $y$ billion edges.

Three graphs from the \href{https://sparse.tamu.edu/}{SuiteSparse Matrix Collection} (GAP-kron, GAP-urand, MOLIERE\_2016) are weighted.
For all other graphs, edge weights are assigned uniformly at random from the range $[1,n^2]$, where $n$ is the number of vertices.
In the BA, UA, and UA-dv graphs, neighbors of each new vertex are sampled with replacement, potentially introducing multi-edges. Accordingly, many of these graphs are multigraphs. 
The two largest graphs in the SSW graphs are also multigraphs. 
The UA-dv and ER-dv graphs are specifically included to study the effects of density variation relative to the UA and ER graphs.

\subsection{Detailed Experimental Setup}
\label{subsec:app_setup}

\begin{table}
  \centering
  \caption{A description of the datasets}
  \begin{subtable}{0.48\textwidth}
    \centering
    \caption{Part 1.The first eight, middle eight, and last eight graphs are referred to as \emph{the SSW graphs}, \emph{the BA graphs}, and \emph{the ER graphs}, respectively. For the SSW graphs, the vertex and edge counts exclude isolated vertices and self-loops, respectively. For the BA graphs, $BA\_x$ denotes the number of edges added per new vertex. $ERx\_y$ denotes a graph with density $1/x$ and $y$ billion edges.}
    {\small
    \begin{tabular}{lrrr}
    \toprule
    \multirow{2}{*}{Graph} & \# of Vertices  & \# of Edges     & \multirow{2}{*}{Density} \\
                           & (in million) & (in billion) & \\
    \midrule
    mycielskian20  & 0.79                     & 1.36                  & 4.38E-03 \\
    com-Friendster & 65.61                    & 1.81                  & 8.39E-07 \\
    GAP-kron       & 63.07                    & 2.11                  & 1.06E-06 \\
    GAP-urand      & 134.22                   & 2.15                  & 2.38E-07 \\
    MOLIERE\_2016  & 30.22                    & 3.34                  & 7.31E-06 \\
    AGATHA\_2015   & 183.96                   & 5.79                  & 3.42E-07 \\
    WDC\_2014      & 1597.59                  & 64.15                 & 5.03E-08 \\
    WDC\_2012      & 3438.46                  & 127.38                & 2.15E-08 \\
                   &                          &                       &          \\
    BA\_512       & 8.65   & 4.64    & 1.24E-04 \\
    BA\_1024      & 8.65   & 8.93    & 2.39E-04 \\
    BA\_2048      & 8.65   & 17.52   & 4.68E-04 \\
    BA\_4096      & 8.65   & 34.70   & 9.27E-04 \\
    BA\_8192      & 8.65   & 69.06   & 1.85E-03 \\
    BA\_16384     & 8.65   & 137.78  & 3.68E-03 \\
    BA\_32768     & 8.65   & 275.22  & 7.36E-03 \\
    BA\_65536     & 8.65   & 550.10  & 1.47E-02 \\
                   &                          &                       &          \\
    ER2\_256        & 1.01                     & 256.00                & 5.00E-01 \\
    ER1\_512        & 1.01                     & 512.00                & 1.00E+00 \\
    ER2\_512        & 1.43                     & 512.00                & 5.00E-01 \\
    ER1\_1024       & 1.43                     & 1024.00               & 1.00E+00 \\
    ER2\_1024       & 2.02                     & 1024.00               & 5.00E-01 \\
    ER1\_2048       & 2.02                     & 2048.00               & 1.00E+00 \\
    ER2\_2048       & 2.86                     & 2048.00               & 5.00E-01 \\
    ER1\_4096       & 2.86                     & 4096.00               & 1.00E+00\\
    \bottomrule
    \end{tabular}
    }
   \label{tab:datasets_p1} 
  \end{subtable}
  \hfill
  \begin{subtable}{0.48\textwidth}
    \centering
    \caption{Part 2. The first eight, middle eight, and last eight graphs are referred to as \emph{the UA-dv graphs}, \emph{the UA graphs}, and \emph{the ER-dv graphs}, respectively. 
    For the UA graphs, $UA\_x$ denotes the number of edges added per new vertex. In the UA-dv graphs, $x$ has the same meaning and $y$ denotes the number of vertices in millions. $ERx\_y$ denotes a graph with density $1/x$ and $y$ billion edges.}
    {\small
    \begin{tabular}{lrrr}
    \toprule
    \multirow{2}{*}{Graph} & \# of Vertices  & \# of Edges     & \multirow{2}{*}{Density} \\
                           & (in million) & (in billion) & \\
    \midrule
    UA\_4096\_67  & 67.37  & 275.22  & 1.21E-04 \\
    UA\_8192\_67  & 67.37  & 550.10  & 2.42E-04 \\
    UA\_2048\_134 & 134.48 & 275.22  & 3.04E-05 \\
    UA\_4096\_134 & 134.48 & 550.10  & 6.08E-05 \\
    UA\_1024\_268 & 268.70 & 275.22  & 7.62E-06 \\
    UA\_2048\_268 & 268.70 & 550.10  & 1.52E-05 \\
    UA\_512\_537  & 537.13 & 275.22  & 1.91E-06 \\
    UA\_1024\_537 & 537.13 & 550.10  & 3.81E-06 \\
                   &                          &                       &          \\
    UA\_1024      & 8.65   & 8.93    & 2.39E-04 \\
    UA\_2048      & 8.65   & 17.52   & 4.68E-04 \\
    UA\_4096      & 8.65   & 34.70   & 9.27E-04 \\
    UA\_8192      & 8.65   & 69.06   & 1.85E-03 \\
    UA\_16384     & 8.65   & 137.78  & 3.68E-03 \\
    UA\_32768     & 8.65   & 275.22  & 7.36E-03 \\
    UA\_65536     & 8.65   & 550.10  & 1.47E-02 \\
    UA\_131072    & 8.65   & 1099.86 & 2.94E-02 \\
                   &                          &                       &          \\
    ER128\_32       & 2.86                     & 32.00                 & 7.81E-03 \\
    ER64\_64        & 2.86                     & 64.00                 & 1.56E-02 \\
    ER32\_128       & 2.86                     & 128.00                & 3.12E-02 \\
    ER16\_256       & 2.86                     & 256.00                & 6.25E-02 \\
    ER8\_512        & 2.86                     & 512.00                & 1.25E-01 \\
    ER4\_1024       & 2.86                     & 1024.00               & 2.50E-01 \\
    ER2\_2048       & 2.86                     & 2048.00               & 5.00E-01 \\
    ER1\_4096       & 2.86                     & 4096.00               & 1.00E+00\\
    \bottomrule
    \end{tabular}
    }
    \label{tab:datasets_p2}
  \end{subtable}
\end{table}

We evaluated the algorithms in a setting where edges become available in edge streams as soon as the algorithms are ready to process them. 
This setup reflects a practical worst-case scenario, independent of how edge streams are generated. 
In fact, delayed edge arrivals would only reduce contention on shared variables, making execution easier.

To simulate a steady flow of edges (and minimize delays in edge availability), both edge generation and processing are confined to the cores of a single node.
Each experiment proceeds in multiple synchronous rounds.
In every round, all processors within the node collaborate to generate $k$ edge streams, each containing a bounded number of edges.
Then $k$ processors, uniformly selected from distinct physical groups (e.g., sockets), simultaneously execute the streaming phase of the algorithm on those edge streams.

In each round, every processor generates a random portion of the graph. For example, for a BA\_x graph, processors collectively sample a random new vertex and attach $x$ edges based on the current degree distribution. 
Each processor maintains up to 8192 buffers and generates edges such that no buffer exceeds 8192 entries. 
Edges are randomly assigned to buffers, and each buffer is independently permuted to vary edge arrival order.
Once the generation is complete, the $k$ streaming processors process their assigned edge buffers in parallel.

For the SSW graphs, we partitioned the edges of each graph into multiples of 128 groups and stored each group in a separate file. (Isolated vertices and self-loops were removed during this process.) This enabled consistent buffering and streaming across all graphs. During execution, processors collectively read and buffer edges from these files using the same strategy described above.

Since streaming alternates with generation over several rounds, streaming time is measured as the sum of the critical-path durations across streaming rounds. All reported metrics, including runtime and effective iterations, are averaged over five runs. Observed variances were negligible, so we report only averages.

\subsection{Detailed Space Usage}
\label{subsec:app_space}

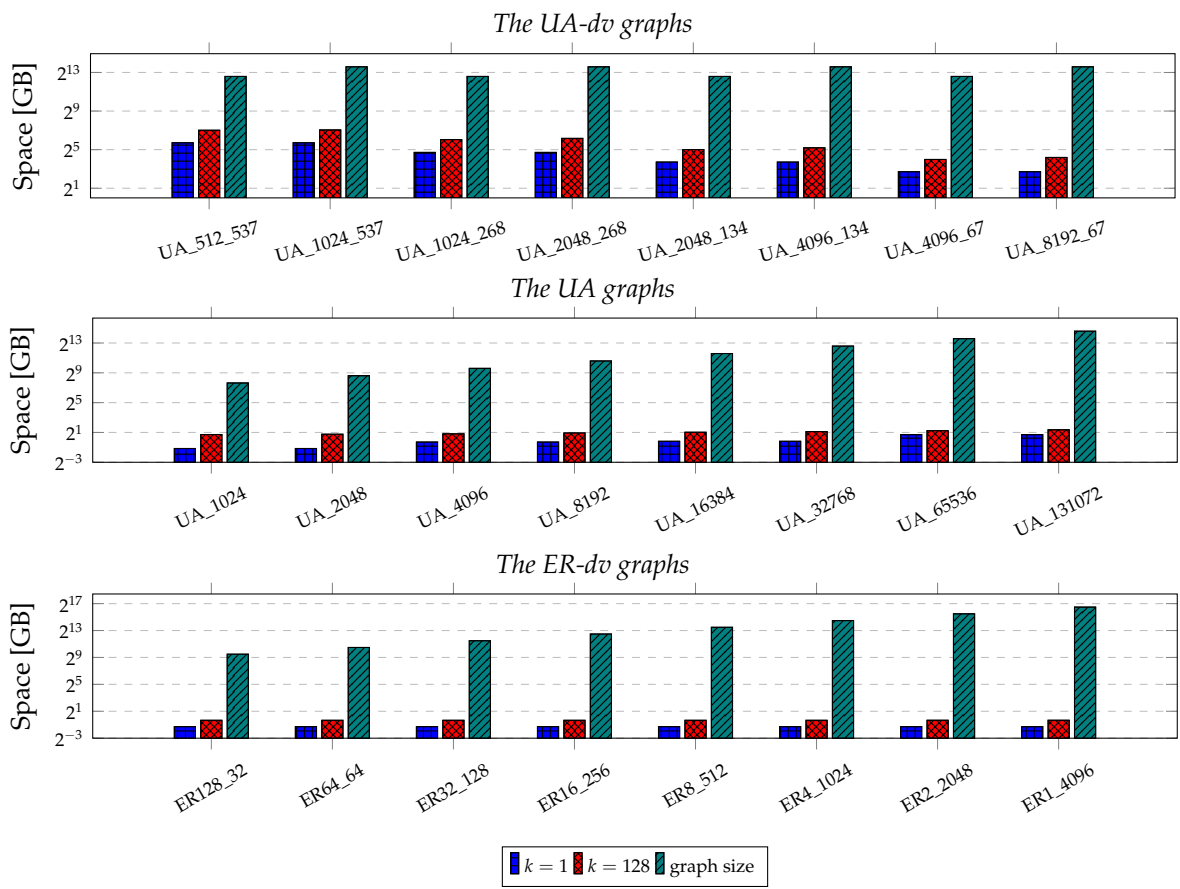
\begin{figure}[h]
\begin{subfigure}{\linewidth}{\emph{The UA-dv graphs}}
\centering
\input{Tikz/space/space_ua_dv}
\end{subfigure}
\begin{subfigure}{\linewidth}{\emph{The UA graphs}}
\centering
\input{Tikz/space/space_ua}
\end{subfigure}
\begin{subfigure}{\linewidth}{\emph{The ER-dv graphs}}
\centering
\input{Tikz/space/space_er_dv} 
\end{subfigure}
\caption{Memory used by the algorithm and the corresponding graph size (space needed to store the entire graph in CSR format). Note that the $y$-axes are in a logarithmic scale.}
\label{fig:space2}
\end{figure}

In \fref{fig:space} (\hyref{Section}{subsec:evals_space}), we presented the space usage of our algorithm alongside the corresponding graph sizes in compressed sparse row (CSR) format for the SSW, BA, and ER graphs. 
\fref{fig:space2} extends this by showing the same for the UA-dv, UA, and ER-dv graphs.

\begin{table}[h]
    \centering    
    \caption{Values of $r$ (number of groups) used in the evaluation.}
    \begin{tabular}{c|l}
    \toprule
        $r$ & Graphs \\
         \toprule
       $\min\{k,2\}$  & WDC\_2012\\
      \midrule
      \multirow{2}{*}{$\min\{k,4\}$} & mycielskian20, GAP-urand, MOLIERE\_2016, WDC\_2014 \\             
            & \emph{the BA graphs} (largest four)\\
      \midrule
      \multirow{1}{*}{$\min\{k,8\}$} & com-Friendster, GAP-kron, AGATHA\_2015\\            
      \midrule
      \multirow{2}{*}{$\min\{k,16\}$} &  \emph{the BA graphs} (smallest four)\\
      & \emph{the ER graphs}, \emph{the UA graphs}, \emph{the ER-dv graphs}\\      
    \bottomrule
    \end{tabular}
    \label{tab:rv}
\end{table}

We observed that using more than one group ($r>1$) with many processors ($k=128$) consistently yields better runtime speedups.
Thus, for $k\geq 16$, suitable values of $r$ lie in the range $(1,16]$, guided by our system's architecture, which includes two sockets and eight memory controllers \cite{Milan}. However, for many graphs, speedup gains plateau beyond a certain $r$, as different graph classes influence the algorithm's memory access patterns in distinct ways, depending on factors such as density, structures, and memory hierarchy interactions.

To avoid unnecessary space usage, we selected the largest value of $r$ (prior to the speedup plateau) from $\{2,4,8,16\}$ based on a small number of runs with $k=128$. 
The selected $r$ values for each graph are listed in \hyref{Table}{tab:rv}.
The effects of varying $r$ are demonstrated using the UA-dv graphs in \hyref{Appendix}{subsec:app_localize}.

\subsection{Detailed Solution Quality}
\label{subsec:app_quality}

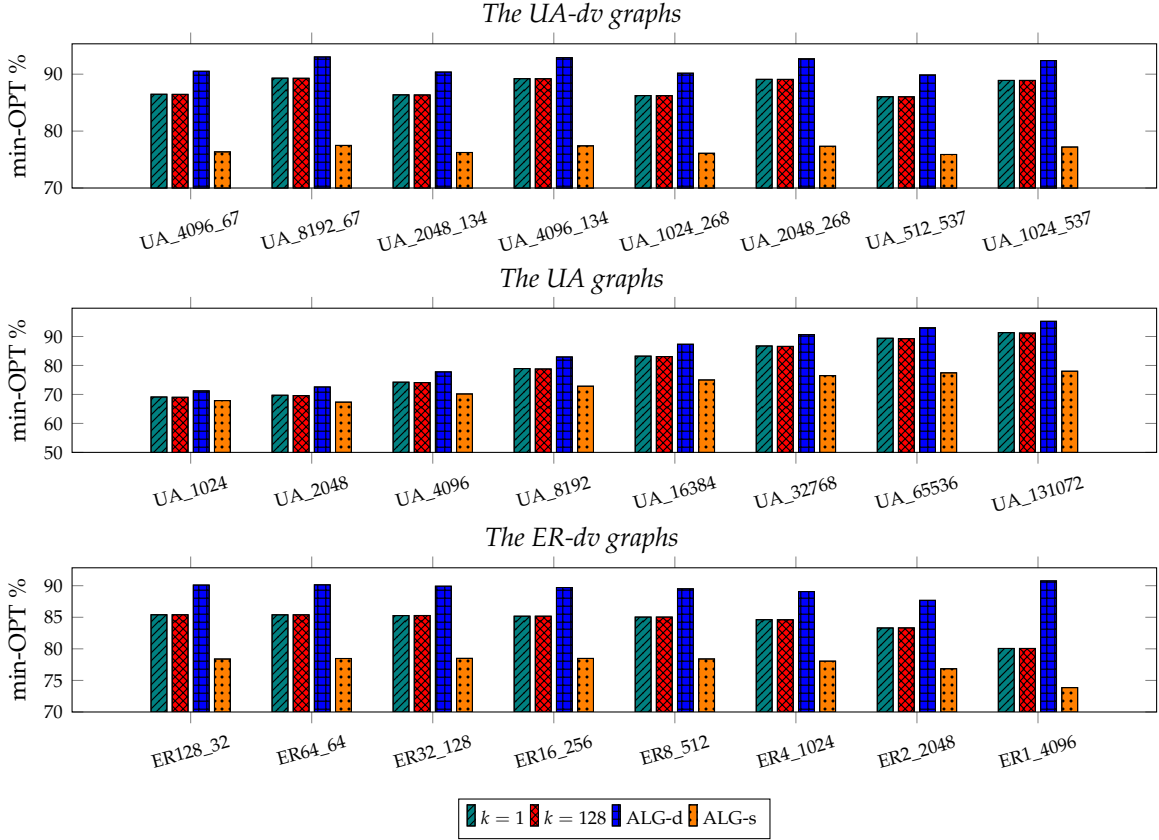
\begin{figure}
\begin{subfigure}{\linewidth}{\emph{The UA-dv graphs}}
\centering
\input{Tikz/quality/sq_ua_dv}
\end{subfigure}
\begin{subfigure}{\linewidth}{\emph{The UA graphs}}
\centering
\input{Tikz/quality/sq_ua}    
\end{subfigure}
\begin{subfigure}{\linewidth}{\emph{The ER-dv graphs}}
\centering
\input{Tikz/quality/sq_er_dv} 
\end{subfigure}
\caption{Comparisons of \emph{min-OPT percent} obtained by different algorithms. \emph{ALG-d} denotes the best results from four dual update rules described in \hyref{Appendix}{subsec:app_duals}, and \emph{ALG-s} denotes the algorithm of Feigenbaum et al.~\cite{feigenbaum2005graph}.}
\label{fig:sq2}
\end{figure}

\begin{figure}
\begin{subfigure}{\linewidth}
\centering
\input{Tikz/quality/sq_greedy}
\end{subfigure}
\begin{subfigure}{\linewidth}
\centering
\input{Tikz/quality/sq_greedy2}
\end{subfigure}
\caption{Comparisons of \emph{min-OPT percent} obtained by different algorithms. \emph{Greedy} was able to solve the 16 smallest graphs in our datasets using 1024 GB of space. For ER128\_32 and ER64\_64, \emph{UniRelaxed} and \emph{ALG-d} used less than 0.8 GB of space, whereas \emph{Greedy} required more than 476 GB and 954 GB of space, respectively.}
\label{fig:sq_greedy}
\end{figure}
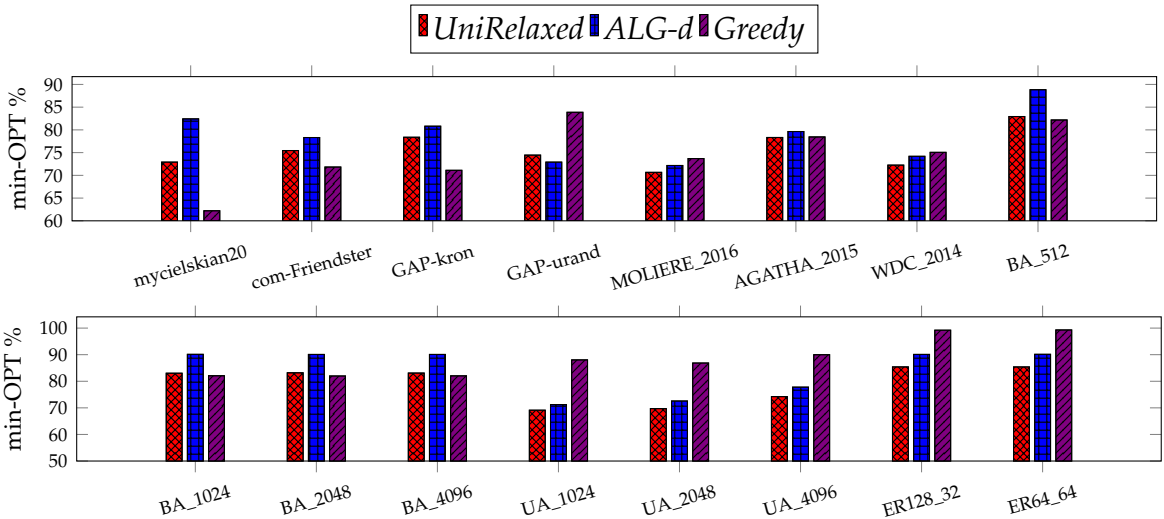

\begin{figure}
\centering
\begin{subfigure}{\linewidth}{\emph{The SSW graphs}}
\centering
\input{Tikz/quality/sqd_ssw}
\end{subfigure}
\begin{subfigure}{\linewidth}{\emph{The ER graphs}}
\centering
\input{Tikz/quality/sqd_er}
\end{subfigure}
\begin{subfigure}{\linewidth}{\emph{The UA graphs}}
\centering
\input{Tikz/quality/sqd_ua}
\end{subfigure}
\begin{subfigure}{\linewidth}{\emph{The BA graphs}}
\centering
\input{Tikz/quality/sqd_ba}
\end{subfigure}
\begin{subfigure}{\linewidth}{The UA-dv graphs}
\centering
\input{Tikz/quality/sqd_ua_dv}
\end{subfigure}
\begin{subfigure}{\linewidth}{The ER-dv graphs}
\centering
\input{Tikz/quality/sqd_er_dv}
\end{subfigure}
\caption{Comparisons of \emph{min-OPT percent} obtained by different dual update rules. \emph{ArgX} stands for the graph-wise best results obtained from \emph{ArgMax}, \emph{ArgMin}, and \emph{ArgRand}.}
\label{fig:sq_dual}
\end{figure}
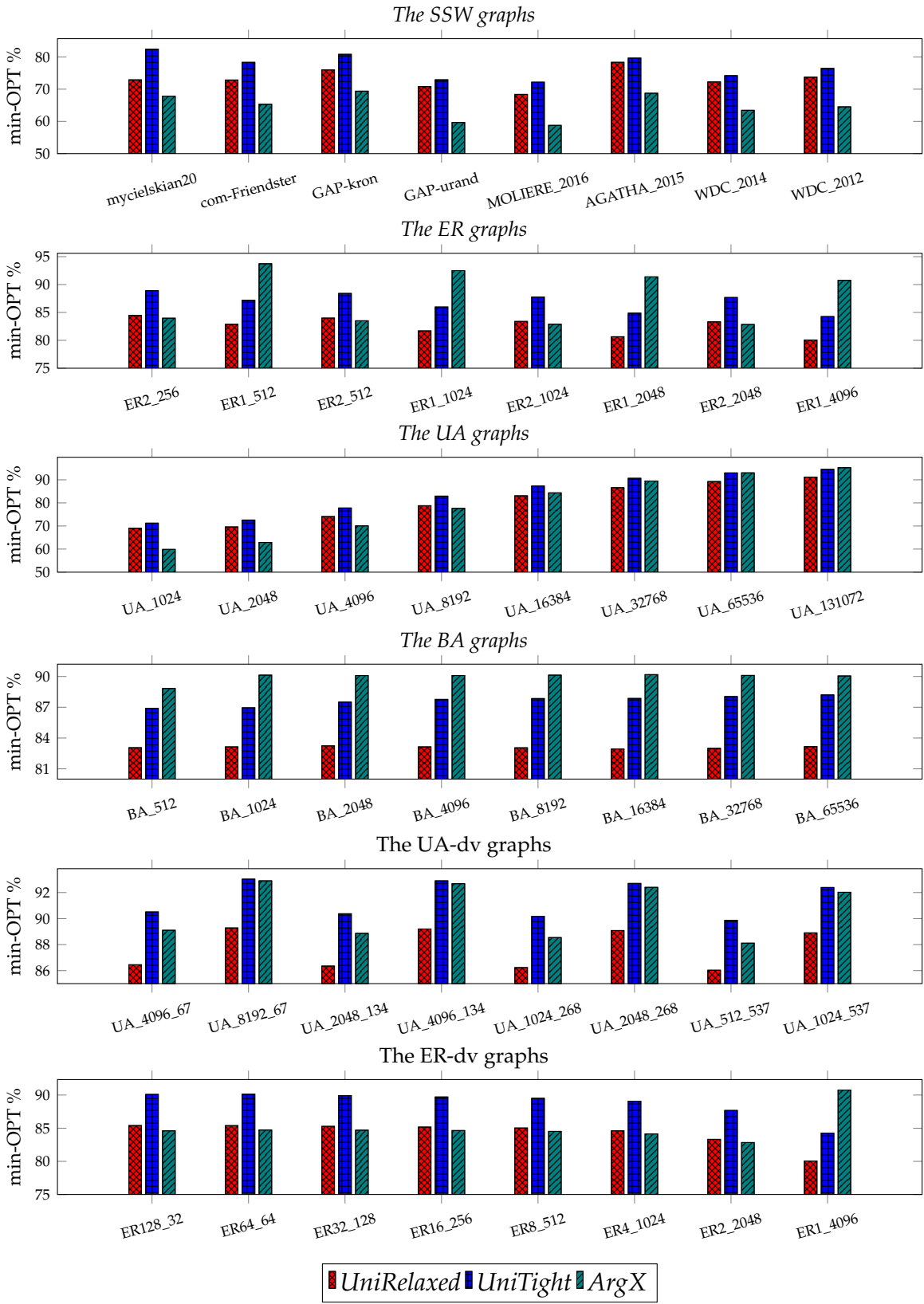

In \fref{fig:sq} (\hyref{Section}{subsec:evals_quality}), we showed the min-OPT percent obtained by different algorithms for the SSW, BA, and ER graphs. 
\fref{fig:sq2} extends this evaluation to the UA-dv, UA, and ER-dv graphs. For all runs of our algorithms, we use $\epsilon=10^{-6}$ (arbitrarily chosen). 

Although all algorithms perform better than their approximation guarantees, the streaming algorithm ALG-s is significantly outperformed by our algorithm on most graphs. 
ALG-s, designed by Feigenbaum et al. \cite{feigenbaum2005graph} (\hyref{Section}{subsec:mwm_related} contains a description), has the same runtime as the local-ratio algorithm by Paz and Schwartzman \cite{paz20182+} but provides only a $6$-approximation guarantee.

Among streaming algorithms for approximating MWM on general graphs (see \hyref{Appendix}{subsec:mwm_related}), the algorithms by \cite{feigenbaum2005graph} is the simplest. 
Algorithms achieving better than a $\br{2+\epsilon}$-approximation require multiple passes and are substantially more complex to implement in practice.
For example, the algorithm of \cite{assadi2024simple} computes a $\br{1+\epsilon}$-approximate solution with high probability but requires $\bigO{\frac{\log n}{\epsilon}}$ passes and computing MWMs in subgraphs, which is computationally expensive in practice.

In \fref{fig:sq_greedy}, we show comparisons with the offline \emph{Greedy} algorithm, which achieves a 2-approximation by repeatedly selecting the heaviest available edge to include in the matching.
This requires storing the entire graph in memory. 
Using 1024 GB of RAM, we were able to run Greedy on the 16 smallest graphs in our datasets.

On some graphs (such as mycielskian20) Greedy substantially underperforms compared to our algorithm. 
While Greedy achieves better solution weights on several graphs, it incurs substantially higher space cost.
For example, on ER128\_32 and ER64\_64, our algorithm used less than 0.8 GB of space, whereas Greedy required more than 476 GB and 954 GB of space, respectively.

In \fref{fig:sq_dual}, we compare min-OPT percent obtained by different dual update rules discussed in \hyref{Appendix}{subsec:app_duals}. 
\emph{ArgX} denotes the best result, graph-wise, among rules ArgMax, ArgMin, and ArgRand.
Unlike UniRelaxed and UniTight, these rules do not offer bounded approximation guarantees but exhibit strong empirical performance across many graphs. 
ArgRand performs comparably to UniRelaxed on most graphs.
ArgMax performs well primarily on the BA graphs, while ArgMin performs well on the UA, UA-dv, and four ER graphs.

These findings suggest that our algorithmic framework still has room for improving solution quality.
In particular, under reasonable assumptions, the dual formulation may enable significant improvement in the approximation ratio, even in the single-pass setting.

\subsection{Detailed Runtime}
\label{subsec:app_runtime}

\begin{table}
\centering
\caption{Breakdown of \aref{Algorithm PS-MWM-LD}{fig:algo_psmwm_ld}'s runtime (in seconds) into three phases.}
{\small
\begin{tabular}{l|rr|rr|rr}
\toprule
\multirow{3}{*}{Graph} & \multicolumn{2}{c|}{Steps 1-4} & 
\multicolumn{2}{c|}{Step 7} & \multicolumn{2}{c}{Steps 5-6}\\
        & \multicolumn{2}{c|}{(Preprocessing)} & 
        \multicolumn{2}{c|}{(Post-Processing)} & \multicolumn{2}{c}{(Streaming)}\\
                       & $k=1$           & $k=128$         & $k=1$            & $k=128$          & $k=1$         & $k=128$       \\
\midrule
mycielskian20          & 0.005           & 0.002           & 0.018            & 0.006            & 6.16          & 0.09          \\
com-Friendster         & 0.306           & 0.043           & 1.532            & 0.162            & 34.78         & 1.31          \\
GAP-kron               & 0.298           & 0.042           & 0.823            & 0.076            & 36.11         & 0.75          \\
GAP-urand              & 0.627           & 0.059           & 5.126            & 0.638            & 60.26         & 3.32          \\
MOLIERE\_2016          & 0.143           & 0.013           & 1.183            & 0.135            & 51.54         & 1.64          \\
AGATHA\_2015           & 0.857           & 0.120           & 4.369            & 0.687            & 114.40        & 4.19          \\
WDC\_2014              & 7.446           & 0.666           & 14.040           & 1.873            & 730.50        & 24.76         \\
WDC\_2012              & 15.760          & 1.021           & 76.870           & 7.981            & 1771.00       & 75.77         \\
                       &                 &                 &                  &                  &               &               \\
BA\_512                & 0.042           & 0.011           & 0.263            & 0.038            & 49.11         & 1.29          \\
BA\_1024               & 0.042           & 0.012           & 0.274            & 0.061            & 100.10        & 2.33          \\
BA\_2048               & 0.044           & 0.011           & 0.302            & 0.066            & 192.40        & 4.23          \\
BA\_4096               & 0.042           & 0.011           & 0.325            & 0.057            & 368.90        & 7.41          \\
BA\_8192               & 0.043           & 0.007           & 0.336            & 0.084            & 701.30        & 12.97         \\
BA\_16384              & 0.042           & 0.005           & 0.343            & 0.067            & 1250.00       & 19.01         \\
BA\_32768              & 0.042           & 0.008           & 0.350            & 0.092            & 1627.00       & 22.23         \\
BA\_65536              & 0.044           & 0.004           & 0.353            & 0.076            & 2084.00       & 29.70         \\
                       &                 &                 &                  &                  &               &               \\
ER2\_256                & 0.006           & 0.004           & 0.042            & 0.016            & 557.30        & 8.15          \\
ER1\_512                & 0.006           & 0.004           & 0.044            & 0.017            & 1002.00       & 12.73         \\
ER2\_512                & 0.008           & 0.004           & 0.063            & 0.021            & 1098.00       & 14.54         \\
ER1\_1024               & 0.008           & 0.004           & 0.066            & 0.023            & 1991.00       & 25.12         \\
ER2\_1024               & 0.010           & 0.008           & 0.086            & 0.030            & 2210.00       & 27.25         \\
ER1\_2048               & 0.010           & 0.005           & 0.087            & 0.038            & 3998.00       & 49.46         \\
ER2\_2048               & 0.014           & 0.005           & 0.125            & 0.047            & 4450.00       & 53.34         \\
ER1\_4096               & 0.014           & 0.005           & 0.127            & 0.047            & 8026.00       & 96.63\\
&                 &                 &                  &                  &               &               \\
UA\_4096\_67  & 0.311 & 0.088 & 6.487  & 0.742 & 3344.00  & 119.40 \\
UA\_8192\_67  & 0.315 & 0.083 & 7.641  & 0.906 & 6451.00  & 214.40 \\
UA\_2048\_134 & 0.626 & 0.155 & 14.320 & 1.813 & 3436.00  & 162.10 \\
UA\_4096\_134 & 0.623 & 0.156 & 16.070 & 2.030 & 6707.00  & 288.10 \\
UA\_1024\_268 & 1.249 & 0.323 & 28.710 & 4.032 & 3525.00  & 197.20 \\
UA\_2048\_268 & 1.259 & 0.309 & 32.690 & 4.369 & 6879.00  & 327.70 \\
UA\_512\_537  & 2.487 & 0.610 & 57.140 & 7.698 & 3681.00  & 224.20 \\
UA\_1024\_537 & 2.569 & 0.599 & 57.080 & 8.657 & 7098.00  & 355.60 \\
              &       &       &        &       &          &        \\
UA\_1024      & 0.042 & 0.013 & 0.308  & 0.070 & 108.30   & 2.67   \\
UA\_2048      & 0.042 & 0.013 & 0.334  & 0.077 & 208.10   & 4.67   \\
UA\_4096      & 0.043 & 0.011 & 0.353  & 0.160 & 392.40   & 8.70   \\
UA\_8192      & 0.043 & 0.019 & 0.419  & 0.124 & 773.80   & 15.93  \\
UA\_16384     & 0.043 & 0.022 & 0.523  & 0.114 & 1405.00  & 26.60  \\
UA\_32768     & 0.042 & 0.018 & 0.561  & 0.138 & 2663.00  & 43.64  \\
UA\_65536     & 0.043 & 0.011 & 0.597  & 0.132 & 5167.00  & 84.17  \\
UA\_131072    & 0.044 & 0.012 & 0.807  & 0.118 & 10180.00 & 168.10 \\
              &       &       &        &       &          &        \\
ER128\_32      & 0.014 & 0.006 & 0.124  & 0.039 & 97.05    & 2.42   \\
ER64\_64       & 0.014 & 0.006 & 0.128  & 0.046 & 180.30   & 4.28   \\
ER32\_128      & 0.014 & 0.006 & 0.126  & 0.044 & 364.10   & 8.08   \\
ER16\_256      & 0.014 & 0.006 & 0.128  & 0.043 & 697.60   & 11.43  \\
ER8\_512       & 0.014 & 0.006 & 0.127  & 0.047 & 1346.00  & 17.66  \\
ER4\_1024      & 0.014 & 0.005 & 0.126  & 0.044 & 2571.00  & 31.52  \\
ER2\_2048      & 0.014 & 0.005 & 0.125  & 0.047 & 4450.00  & 53.34  \\
ER1\_4096      & 0.014 & 0.005 & 0.127  & 0.047 & 8026.00  & 96.63 \\
\bottomrule
\end{tabular}
}
\label{tab:runtime}
\end{table}

\begin{figure}
\centering
\begin{subfigure}[t]{0.46\textwidth}{The ER-dv graphs}
\centering
\input{Tikz/time/sp_er_dv}
\end{subfigure}
\hfill
\begin{subfigure}[t]{0.46\textwidth}{The UA graphs}
\centering
\input{Tikz/time/sp_ua}
\end{subfigure}
\caption{Speedup in runtime vs.\ $k$. Note that both axes are on a logarithmic scale.}
\label{fig:speedup2}
\end{figure}
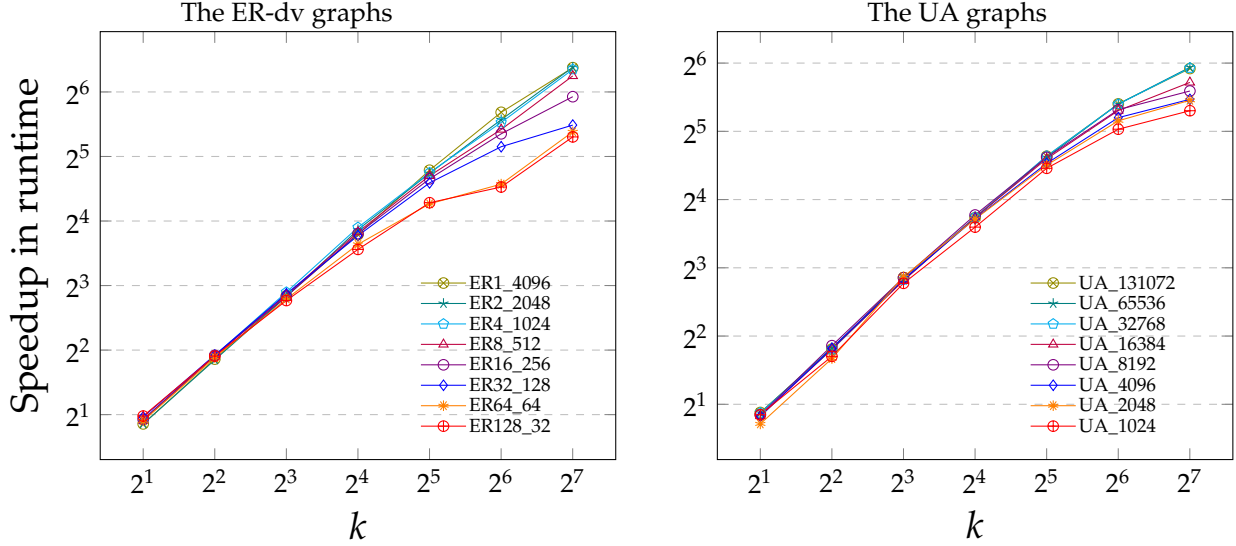

\hyref{Table}{tab:runtime} shows the breakdown of runtime into three phases, preprocessing, post-processing, and streaming, for $k=1$ and $k=128$.

In \fref{fig:speedup} (\hyref{Section}{subsec:evals_runtime}), we presented speedups computed w.r.t. runtime for the SSW, BA, and ER graphs. \fref{fig:speedup2} extends this evaluation to the UA and ER-dv graphs.
The UA-dv graphs are specifically included to assess the effects of localizing memory access for a fixed value of $k$ (\hyref{Appendix}{subsec:app_localize}). 
We also measured speedups w.r.t. effective iterations. For $k=128$, these speedup values lie in the range 112--128 across all graphs in our datasets.

A few remarks on runtime-based speedups are in order.
Using a simple memory read-write test (concurrent but local to each processor), we verified that the memory system of the nodes we used can support speedups of about 70--80 when using 128 processors, limited by eight memory controllers per node \cite{Milan}. 
Therefore, speedups beyond 80 are achievable only when a significant portion of the algorithm's working set fits in cache.

For $k>16$, the BA, UA, and ER-dv graphs show how decreasing graph density exacerbates the memory system's inability to efficiently serve large volumes of concurrent, random memory accesses. 
The SSW graphs further highlight two key memory bottlenecks: limited support for concurrent/random accesses (visible for $k>16$), and non-uniform memory access costs (visible even for $k=2$).
The small runtime speedups observed for $k=2$ on the SSW graphs are due to memory accesses across sockets being more than three times slower than accesses within a socket (we choose $k$ cores evenly from distinct physical groups, such as sockets \cite{Milan}).

We replicated the extreme bottlenecks observed for the SSW graphs using the UA-dv graphs (\hyref{Appendix}{subsec:app_localize}). 
These results indicate that the sparser graphs encounter more memory-related bottlenecks due to their greater reliance on random accesses. 
We further confirmed this by examining speedups w.r.t. effective iterations, which show that processors experience negligible contention when accessing shared variables. 
Thus, the algorithm could achieve even better runtime speedups on architectures with more memory controllers and/or stronger support for remote memory access.

\subsection{Per-Edge Processing Time}
\label{subsec:pept}

From the analyses in \lemref{lemma_algo1_space_ns} and \lemref{lemma_algo1_space_ld}, if $\lmin = \bigOmega{n}$, then our algorithm has $\bigO{\log n}$ amortized per-edge processing time. 
All graphs in our datasets satisfy this condition due to the edge stream generation procedure described in \hyref{Appendix}{subsec:app_setup}. 
Using the notion of effective iterations, we now show that, in practice, the algorithm achieves $\bigO{1}$ amortized per-edge processing time.

For each processor $\ell \in k$, we compute the ratio of the number of supersteps taken by processor $\ell$ to $\lmin$ (noting that $\lmin \leq |E^{\ell}|$). 
The maximum of these ratios over all $k$ processors serves as an upper bound on the amortized per-edge processing time. 
This is equivalent to the ratio of the effective iterations to $\lmin$.

For the SSW graphs, the maximum value of this ratio across all graphs and all values of $k$ is $1.15$.
For the BA and UA graphs, it remains below $1.05$; and for the ER and ER-dv graphs, it is below $1.003$. These results confirm that the amortized per-edge processing time of the algorithm is bounded by a small constant in practice.

\subsection{Effect of Localizing Memory Access}
\label{subsec:app_localize}

\begin{figure}
\centering
\begin{minipage}{0.23\textwidth}
\begin{subfigure}[t]{\linewidth}
\centering
\input{Tikz/time/sp_sparse_lma1}
\end{subfigure}
\begin{subfigure}[t]{\linewidth}
\centering
\input{Tikz/time/sp_sparse_lma2}
\end{subfigure}    
\end{minipage}
\hspace{0.5cm}
\begin{minipage}{0.23\textwidth}
\begin{subfigure}[t]{\linewidth}
\centering
\input{Tikz/time/sp_sparse_lma3}
\end{subfigure}
\begin{subfigure}[t]{\linewidth}
\centering
\input{Tikz/time/sp_sparse_lma4}
\end{subfigure}
\end{minipage}
\hspace{0.1cm}
\begin{minipage}{0.46\textwidth}
\begin{subfigure}[t]{\linewidth}
\centering
\input{Tikz/time/sp_dense_lma}
\end{subfigure}    
\end{minipage}
\caption{Effect of localizing memory access for four sparse and four dense graphs.}
\label{fig:sp_lma}
\end{figure}
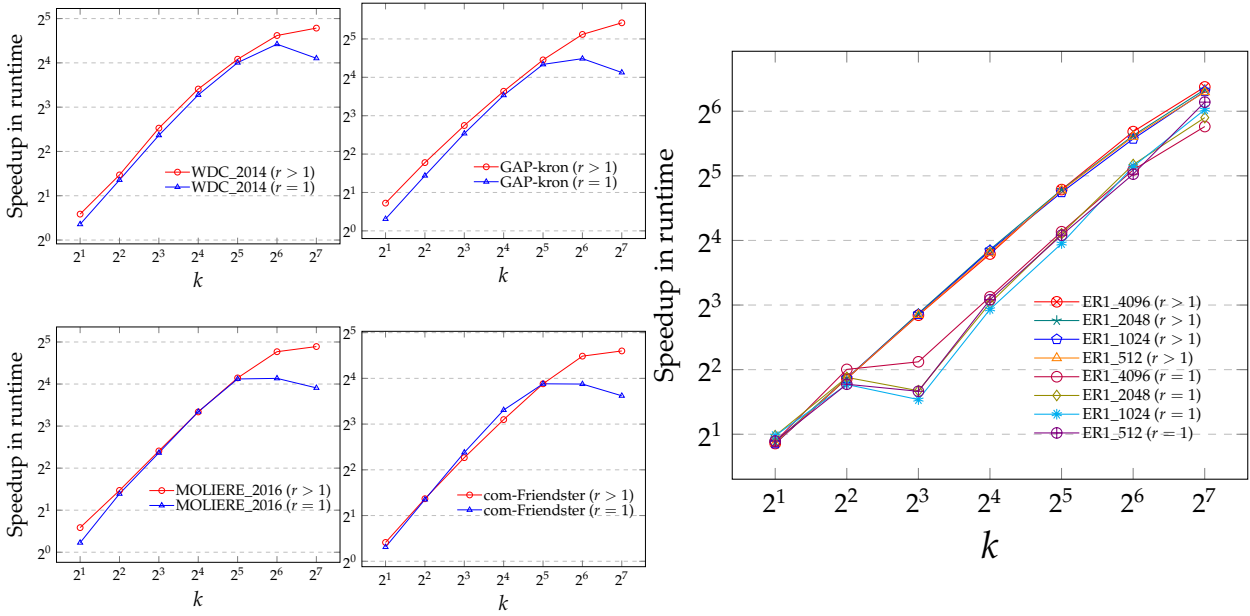

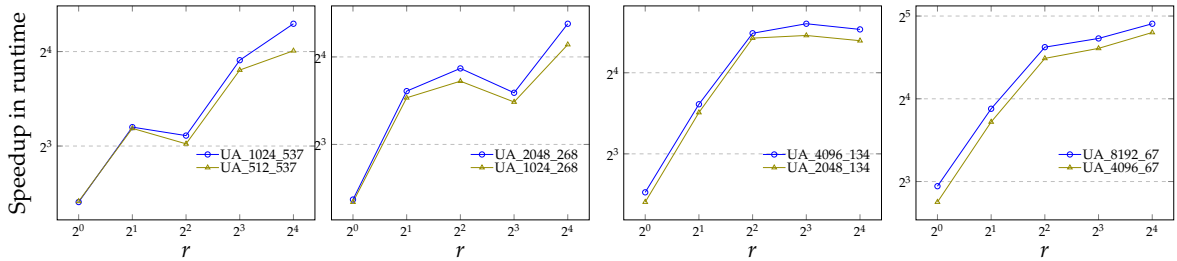
\begin{figure}
\centering
\begin{subfigure}[t]{0.23\textwidth}
\centering
\input{Tikz/time/sp_ua_537}
\end{subfigure}
\hspace{0.1cm}
\begin{subfigure}[t]{0.23\textwidth}
\centering
\input{Tikz/time/sp_ua_268}
\end{subfigure}
\begin{subfigure}[t]{0.23\textwidth}
\centering
\input{Tikz/time/sp_ua_134}
\end{subfigure}
\begin{subfigure}[t]{0.23\textwidth}
\centering
\input{Tikz/time/sp_ua_67}
\end{subfigure}
\caption{Speedup in runtime vs. $r$, for $k=128$. Each of the subplots shows the effect of density. From left to right, density increases.}
\label{fig:sp_ua_x_y}
\end{figure}

All graphs exhibit significant gains from memory access localization (recall \aref{Algorithm PS-MWM-LD}{fig:algo_psmwm_ld}). 
\fref{fig:sp_lma} illustrates this effect on four sparse and four dense graphs (see \hyref{Appendix}{subsec:app_space} for the corresponding values of $r$).
For $k=128$ and $r=1$, speedup decreases as the number of vertices (or the size of the working set) increases.
In contrast, with localized memory access ($r>1$), speedups increase steadily. 

The benefits of localization become more pronounced as the number of random accesses increases. 
This trend is further demonstrated in \fref{fig:sp_ua_x_y} using the UA-dv graphs. 
Each of the subplots compares two graphs of different densities. From left to right, graph density increases, and we observe a corresponding rise in runtime-based speedups.

%% file: Tikz/space/space_ua_dv.tex
\begin{tikzpicture}
\begin{axis}[   
    ymin=1,
    width=16cm,
    height=3.5cm,
    symbolic x coords={UA\_512\_537, UA\_1024\_537, UA\_1024\_268, UA\_2048\_268, UA\_2048\_134, UA\_4096\_134, UA\_4096\_67, UA\_8192\_67},
    ylabel={Space [GB]},
    enlargelimits=0.14,
    enlarge y limits=upper,    
    ybar, bar width=8pt,
    ymajorgrids=true,
    grid style=dashed,    
    ymode=log,
    log basis y={2},    
    log origin=infty,
    point meta=explicit,
    x tick label style={rotate=15},
    x tick label style={font=\scriptsize},
    y tick label style={font=\scriptsize},
    ytick distance=16
    ]

\addplot [fill=blue, postaction={pattern=grid}] table [meta=k1] \spaceuadv;

\addplot [fill=red, postaction={pattern=crosshatch}] table [y=k128, meta=k128] \spaceuadv;

\addplot [fill=teal, postaction={pattern=north east lines}] table [y=gsize, meta=gsize] \spaceuadv;
\end{axis}
\end{tikzpicture}

%% file: Tikz/space/space_ua.tex
\begin{tikzpicture}
\begin{axis}[  
    ymin=0.125,
    width=16cm,
    height=3.5cm,
    symbolic x coords={UA\_1024, UA\_2048, UA\_4096, UA\_8192, UA\_16384, UA\_32768, UA\_65536, UA\_131072},
    ylabel={Space [GB]},
    enlargelimits=0.14,
    enlarge y limits=upper,    
    ybar, bar width=8pt,
    ymajorgrids=true,
    grid style=dashed,
    ymode=log,
    log basis y={2},    
    log origin=infty,
    point meta=explicit,
    x tick label style={rotate=25},
    x tick label style={font=\scriptsize},
    y tick label style={font=\scriptsize},
    ytick distance=16
    ]

\addplot [fill=blue, postaction={pattern=grid}] table [meta=k1] \spaceua;

\addplot [fill=red, postaction={pattern=crosshatch}] table [y=k128, meta=k128] \spaceua;

\addplot [fill=teal, postaction={pattern=north east lines}] table [y=gsize, meta=gsize] \spaceua;
\end{axis}
\end{tikzpicture}

%% file: Tikz/space/space_er_dv.tex
\begin{tikzpicture}
\begin{axis}[        
    ymin=0.125,
    width=16cm,
    height=3.5cm,
    symbolic x coords={ER128\_32, ER64\_64, ER32\_128, ER16\_256, ER8\_512, ER4\_1024, ER2\_2048, ER1\_4096},
    ylabel={Space [GB]},
    enlargelimits=0.14,
    enlarge y limits=upper,    
    ybar, bar width=8pt,
    ymajorgrids=true,
    grid style=dashed, 
    legend style={        
        font=\scriptsize,
        at={(0.5,-0.75)},
        anchor=north,
        legend columns=-1
        },
    ymode=log,
    log basis y={2},  
    log origin=infty,
    point meta=explicit,
    x tick label style={rotate=30},
    x tick label style={font=\scriptsize},
    y tick label style={font=\scriptsize},
    ytick distance=16
    ]

\addplot [fill=blue, postaction={pattern=grid}] table [meta=k1] \spaceerdv;

\addplot [fill=red, postaction={pattern=crosshatch}] table [y=k128, meta=k128] \spaceerdv;

\addplot [fill=teal, postaction={pattern=north east lines}] table [y=gsize, meta=gsize] \spaceerdv;
\legend{$k=1$, $k=128$, graph size}
\end{axis}
\end{tikzpicture}

%% file: Tikz/quality/sq_ua_dv.tex
\begin{tikzpicture}
\begin{axis}[
    ymin=70,    
    width=16cm,
    height=3.5cm,    
    ylabel={min-OPT \%},
    ylabel style={font=\small},
    enlargelimits=0.14,
    enlarge y limits=upper,    
    ybar, bar width=6pt,
    point meta=explicit,
    symbolic x coords={UA\_4096\_67, UA\_8192\_67, UA\_2048\_134, UA\_4096\_134, UA\_1024\_268, UA\_2048\_268, UA\_512\_537, UA\_1024\_537},
    x tick label style={rotate=15},    
    x tick label style={font=\scriptsize},
    y tick label style={font=\scriptsize},
    ytick distance=10,
    nodes near coords={}    
]
\addplot [fill=teal, postaction={pattern=north east lines}] table [meta=k1] \squadv;
\addplot [fill=red, postaction={pattern=crosshatch}] table [y=k128, meta=k128] \squadv;
\addplot [fill=blue, postaction={pattern=grid}] table [y=algdu, meta=algdu] \squadv;
\addplot [fill=orange, postaction={pattern=dots}] table [y=algss, meta=algss] \squadv;
\end{axis}
\end{tikzpicture}

%% file: Tikz/quality/sq_ua.tex
\begin{tikzpicture}
\begin{axis}[
    ymin=50,    
    width=16cm,
    height=3.5cm,    
    ylabel={min-OPT \%},
    ylabel style={font=\small},
    enlargelimits=0.14,
    enlarge y limits=upper,    
    ybar, bar width=6pt,
    point meta=explicit,
    symbolic x coords={UA\_1024, UA\_2048, UA\_4096, UA\_8192, UA\_16384, UA\_32768, UA\_65536, UA\_131072},
    x tick label style={rotate=15},    
    x tick label style={font=\scriptsize},
    y tick label style={font=\scriptsize},
    ytick distance=10,
    nodes near coords={}    
]
\addplot [fill=teal, postaction={pattern=north east lines}] table [meta=k1] \squa;
\addplot [fill=red, postaction={pattern=crosshatch}] table [y=k128, meta=k128] \squa;
\addplot [fill=blue, postaction={pattern=grid}] table [y=algdu, meta=algdu] \squa;
\addplot [fill=orange, postaction={pattern=dots}] table [y=algss, meta=algss] \squa;
\end{axis}
\end{tikzpicture}

%% file: Tikz/quality/sq_er_dv.tex
\begin{tikzpicture}
\begin{axis}[
    ymin=70,    
    width=16cm,
    height=3.5cm,    
    ylabel={min-OPT \%},
    ylabel style={font=\small},
    enlargelimits=0.14,
    enlarge y limits=upper,
    legend style={
        font=\scriptsize,
        at={(0.5,-0.6)},
        anchor=north,
        legend columns=-1
    },
    ybar, bar width=6pt,
    point meta=explicit,
    symbolic x coords={ER128\_32, ER64\_64, ER32\_128, ER16\_256, ER8\_512, ER4\_1024, ER2\_2048, ER1\_4096},
    x tick label style={rotate=15},
    x tick label style={font=\scriptsize},
    y tick label style={font=\scriptsize},
    ytick distance=5,
    nodes near coords={}    
]
\addplot [fill=teal, postaction={pattern=north east lines}] table [meta=k1] \sqerdv;
\addplot [fill=red, postaction={pattern=crosshatch}] table [y=k128, meta=k128] \sqerdv;
\addplot [fill=blue, postaction={pattern=grid}] table [y=algdu, meta=algdu] \sqerdv;
\addplot [fill=orange, postaction={pattern=dots}] table [y=algss, meta=algss] \sqerdv;
\legend{$k=1$, $k=128$, ALG-d, ALG-s}
\end{axis}
\end{tikzpicture}

%% file: Tikz/quality/sq_greedy.tex
\begin{tikzpicture}
\begin{axis}[
    ymin=60,    
    width=16cm,
    height=3.5cm,    
    ylabel={min-OPT \%},
    ylabel style={font=\small},
    enlargelimits=0.14,
    enlarge y limits=upper,    
    legend style={
        font=\large,
        at={(0.5,1.5)},
        anchor=north,
        legend columns=-1
    },
    ybar, bar width=6pt,
    point meta=explicit,
    symbolic x coords={mycielskian20,com-Friendster,GAP-kron, GAP-urand, MOLIERE\_2016, AGATHA\_2015, WDC\_2014, BA\_512},
    x tick label style={rotate=15},
    x tick label style={font=\scriptsize},
    y tick label style={font=\scriptsize},
    ytick distance=5,
    nodes near coords={}    
]
\addplot [fill=red, postaction={pattern=crosshatch}] table [meta=relaxed] \sqgreedy;
 \addplot [fill=blue, postaction={pattern=grid}] table [y=algdu, meta=algdu] \sqgreedy;
\addplot [fill=violet, postaction={pattern=north east lines}] table [y=greedy, meta=greedy] \sqgreedy;
\legend{\emph{UniRelaxed}, \emph{ALG-d}, \emph{Greedy}}
\end{axis}
\end{tikzpicture}

%% file: Tikz/quality/sq_greedy2.tex
\begin{tikzpicture}
\begin{axis}[
    ymin=50,    
    width=16cm,
    height=3.5cm,    
    ylabel={min-OPT \%},
    ylabel style={font=\small},
    enlargelimits=0.14,
    enlarge y limits=upper,    
    ybar, bar width=6pt,
    point meta=explicit,
    symbolic x coords={BA\_1024, BA\_2048, BA\_4096, UA\_1024, UA\_2048, UA\_4096, ER128\_32, ER64\_64},
    x tick label style={rotate=15},
    x tick label style={font=\scriptsize},
    y tick label style={font=\scriptsize},
    ytick distance=10,
    nodes near coords={}    
]
\addplot [fill=red, postaction={pattern=crosshatch}] table [meta=relaxed] \sqgreedytwo;
 \addplot [fill=blue, postaction={pattern=grid}] table [y=algdu, meta=algdu] \sqgreedytwo;
\addplot [fill=violet, postaction={pattern=north east lines}] table [y=greedy, meta=greedy] \sqgreedytwo;
\end{axis}
\end{tikzpicture}

%% file: Tikz/quality/sqd_ssw.tex
\pgfplotstableread{
Graph    relaxed     tight    argx
mycielskian20	72.89	82.42	67.85
com-Friendster	72.84	78.32	65.33
GAP-kron	75.98	80.81	69.38
GAP-urand	70.78	72.91	59.67
MOLIERE\_2016	68.38	72.17	58.77
AGATHA\_2015	78.33	79.64	68.74
WDC\_2014	72.26	74.19	63.45
WDC\_2012	73.72	76.43	64.52
}\webgraphs

\makeatletter
\begin{tikzpicture}
\begin{axis}[
    ymin=50,    
    width=16cm,
    height=3.5cm,    
    ylabel={min-OPT \%},
    ylabel style={font=\small},
    enlargelimits=0.14,
    enlarge y limits=upper,        
    ybar, bar width=6pt,    
    point meta=explicit,
    symbolic x coords={mycielskian20,com-Friendster,GAP-kron, GAP-urand, MOLIERE\_2016, AGATHA\_2015, WDC\_2014, WDC\_2012},
    x tick label style={rotate=15},    
    x tick label style={font=\scriptsize},
    y tick label style={font=\scriptsize},
    ytick distance=10,
    nodes near coords={}    
]
\addplot [fill=red, postaction={pattern=crosshatch}] table [meta=relaxed] \webgraphs;
 \addplot [fill=blue, postaction={pattern=grid}] table [y=tight, meta=tight] \webgraphs;
\addplot [fill=teal, postaction={pattern=north east lines}] table [y=argx, meta=argx] \webgraphs;
\end{axis}
\end{tikzpicture}

%% file: Tikz/quality/sqd_er.tex
\pgfplotstableread{
Graph    relaxed     tight    argx
ER2\_256	84.46	88.90	83.97
ER1\_512	82.87	87.20	93.73
ER2\_512	83.99	88.41	83.51
ER1\_1024	81.70	85.98	92.50
ER2\_1024	83.39	87.77	82.90
ER1\_2048	80.65	84.87	91.38
ER2\_2048	83.32	87.68	82.84
ER1\_4096	80.06	84.26	90.76
}\webgraphs

\makeatletter
\begin{tikzpicture}
\begin{axis}[
    ymin=75,    
    width=16cm,
    height=3.5cm,    
    ylabel={min-OPT \%},
    ylabel style={font=\small},
    enlargelimits=0.14,
    enlarge y limits=upper,        
    ybar, bar width=6pt,
    point meta=explicit,
    symbolic x coords={ER2\_256, ER1\_512, ER2\_512, ER1\_1024, ER2\_1024, ER1\_2048, ER2\_2048, ER1\_4096},
    x tick label style={rotate=15},    
    x tick label style={font=\scriptsize},
    y tick label style={font=\scriptsize},
    ytick distance=5,
    nodes near coords={}    
]
\addplot [fill=red, postaction={pattern=crosshatch}] table [meta=relaxed] \webgraphs;
 \addplot [fill=blue, postaction={pattern=grid}] table [y=tight, meta=tight] \webgraphs;
\addplot [fill=teal, postaction={pattern=north east lines}] table [y=argx, meta=argx] \webgraphs;
\end{axis}
\end{tikzpicture}

%% file: Tikz/quality/sqd_ua.tex
\pgfplotstableread{
Graph    relaxed     tight    argx
UA\_1024	69.02	71.23	59.86
UA\_2048	69.60	72.58	62.88
UA\_4096	74.09	77.79	70.07
UA\_8192	78.79	82.94	77.66
UA\_16384	83.08	87.32	84.35
UA\_32768	86.58	90.65	89.48
UA\_65536	89.25	92.98	93.01
UA\_131072	91.19	94.52	95.24
}\webgraphs

\makeatletter
\begin{tikzpicture}
\begin{axis}[
    ymin=50,    
    width=16cm,
    height=3.5cm,    
    ylabel={min-OPT \%},
    ylabel style={font=\small},
    enlargelimits=0.14,
    enlarge y limits=upper,      
    ybar, bar width=6pt,    
    point meta=explicit,
    symbolic x coords={UA\_1024, UA\_2048, UA\_4096, UA\_8192, UA\_16384, UA\_32768, UA\_65536, UA\_131072},
    x tick label style={rotate=15},    
    x tick label style={font=\scriptsize},
    y tick label style={font=\scriptsize},
    ytick distance=10,
    nodes near coords={}    
]
\addplot [fill=red, postaction={pattern=crosshatch}] table [meta=relaxed] \webgraphs;
 \addplot [fill=blue, postaction={pattern=grid}] table [y=tight, meta=tight] \webgraphs;
\addplot [fill=teal, postaction={pattern=north east lines}] table [y=argx, meta=argx] \webgraphs;
\end{axis}
\end{tikzpicture}

%% file: Tikz/quality/sqd_ba.tex
\pgfplotstableread{
Graph    relaxed     tight    argx
BA\_512	83.06	86.89	88.83
BA\_1024	83.14	86.95	90.13
BA\_2048	83.24	87.50	90.09
BA\_4096	83.14	87.77	90.09
BA\_8192	83.05	87.85	90.13
BA\_16384	82.94	87.86	90.18
BA\_32768	83.00	88.04	90.10
BA\_65536	83.16	88.20	90.05
}\webgraphs

\makeatletter
\begin{tikzpicture}
\begin{axis}[
    ymin=80,    
    width=16cm,
    height=3.5cm,    
    ylabel={min-OPT \%},
    ylabel style={font=\small},
    enlargelimits=0.14,
    enlarge y limits=upper,    
    ybar, bar width=6pt,    
    point meta=explicit,
    symbolic x coords={BA\_512, BA\_1024, BA\_2048, BA\_4096, BA\_8192, BA\_16384, BA\_32768, BA\_65536},
    x tick label style={rotate=15},    
    x tick label style={font=\scriptsize},
    y tick label style={font=\scriptsize},
    ytick distance=3,
    nodes near coords={}    
]
\addplot [fill=red, postaction={pattern=crosshatch}] table [meta=relaxed] \webgraphs;
 \addplot [fill=blue, postaction={pattern=grid}] table [y=tight, meta=tight] \webgraphs;
\addplot [fill=teal, postaction={pattern=north east lines}] table [y=argx, meta=argx] \webgraphs;
\end{axis}
\end{tikzpicture}

%% file: Tikz/quality/sqd_ua_dv.tex
\pgfplotstableread{
Graph    relaxed     tight    argx
UA\_4096\_67	86.45	90.51	89.10
UA\_8192\_67	89.28	93.03	92.90
UA\_2048\_134	86.35	90.37	88.86
UA\_4096\_134	89.20	92.90	92.68
UA\_1024\_268	86.23	90.17	88.54
UA\_2048\_268	89.08	92.70	92.40
UA\_512\_537	86.03	89.86	88.11
UA\_1024\_537	88.90	92.39	92.02
}\webgraphs

\makeatletter
\begin{tikzpicture}
\begin{axis}[
    ymin=85,    
    width=16cm,
    height=3.5cm,    
    ylabel={min-OPT \%},
    ylabel style={font=\small},
    enlargelimits=0.14,
    enlarge y limits=upper,    
    ybar, bar width=6pt,
    point meta=explicit,
    symbolic x coords={UA\_4096\_67, UA\_8192\_67, UA\_2048\_134, UA\_4096\_134, UA\_1024\_268, UA\_2048\_268, UA\_512\_537, UA\_1024\_537},
    x tick label style={rotate=15},    
    x tick label style={font=\scriptsize},
    y tick label style={font=\scriptsize},
    ytick distance=2,
    nodes near coords={}    
]
\addplot [fill=red, postaction={pattern=crosshatch}] table [meta=relaxed] \webgraphs;
 \addplot [fill=blue, postaction={pattern=grid}] table [y=tight, meta=tight] \webgraphs;
\addplot [fill=teal, postaction={pattern=north east lines}] table [y=argx, meta=argx] \webgraphs;
\end{axis}
\end{tikzpicture}

%% file: Tikz/quality/sqd_er_dv.tex
\pgfplotstableread{
Graph    relaxed     tight    argx
ER128\_32	85.40	90.10	84.61
ER64\_64	85.39	90.14	84.73
ER32\_128	85.28	89.92	84.70
ER16\_256	85.17	89.71	84.65
ER8\_512	85.03	89.51	84.52
ER4\_1024	84.61	89.07	84.14
ER2\_2048	83.32	87.68	82.84
ER1\_4096	80.06	84.26	90.76
}\webgraphs

\makeatletter
\begin{tikzpicture}
\begin{axis}[
    ymin=75,    
    width=16cm,
    height=3.5cm,    
    ylabel={min-OPT \%},
    ylabel style={font=\small},
    enlargelimits=0.14,
    enlarge y limits=upper, 
    legend style={
        font=\large,
        at={(0.5,-0.6)},
        anchor=north,
        legend columns=-1
    },
    ybar, bar width=6pt,    
    point meta=explicit,
    symbolic x coords={ER128\_32, ER64\_64, ER32\_128, ER16\_256, ER8\_512, ER4\_1024, ER2\_2048, ER1\_4096},
    x tick label style={rotate=15},    
    x tick label style={font=\scriptsize},
    y tick label style={font=\scriptsize},
    ytick distance=5,
    nodes near coords={}    
]
\addplot [fill=red, postaction={pattern=crosshatch}] table [meta=relaxed] \webgraphs;
 \addplot [fill=blue, postaction={pattern=grid}] table [y=tight, meta=tight] \webgraphs;
\addplot [fill=teal, postaction={pattern=north east lines}] table [y=argx, meta=argx] \webgraphs;
\legend{\emph{UniRelaxed}, \emph{UniTight}, \emph{ArgX}}
\end{axis}
\end{tikzpicture}

%% file: Tikz/time/sp_er_dv.tex
\begin{tikzpicture}
\begin{axis}[        
    xlabel={$k$},
    xlabel style={font=\Large},
    ylabel={Speedup in runtime},
    ylabel style={font=\Large},
    ymajorgrids=true,
    grid style=dashed,
    legend style={        
        font=\Large,
        nodes={scale=0.5, transform shape},
        cells={anchor=west},
        at={(0.75,0.45)},
        anchor=north,
        draw=none,
        fill=none},
    xmode=log,
    ymode=log,
    log basis y={2},
    log basis x={2},
    ]

\addplot [mark=otimes,olive] table [x={nstreams}, y={G1_4096}] \sperdv;

\addplot [mark=star,teal] table [x={nstreams}, y={G2_2048}] \sperdv;

\addplot [mark=pentagon,cyan] table [x={nstreams}, y={G4_1024}] \sperdv;

\addplot [mark=triangle,purple] table [x={nstreams}, y={G8_512}] \sperdv;

\addplot [mark=o,violet] table [x={nstreams}, y={G16_256}] \sperdv;

\addplot [mark=diamond,blue] table [x={nstreams}, y={G32_128}] \sperdv;

\addplot [mark=10-pointed star,orange] table [x={nstreams}, y={G64_64}] \sperdv;

\addplot [mark=oplus,red] table [x={nstreams}, y={G128_32}] \sperdv;

\legend{ER1\_4096, ER2\_2048, ER4\_1024, ER8\_512, ER16\_256, ER32\_128, ER64\_64, ER128\_32};

\end{axis}
\end{tikzpicture}

%% file: Tikz/time/sp_ua.tex
\begin{tikzpicture}
\begin{axis}[        
    xlabel={$k$},
    xlabel style={font=\Large},
    ylabel={},
    ymajorgrids=true,
    grid style=dashed,
    legend style={   
        font=\Large,
        nodes={scale=0.5, transform shape},
        cells={anchor=west},
        at={(0.75,0.45)},
        anchor=north,
        draw=none,
        fill=none},
    xmode=log,
    ymode=log,
    log basis y={2},
    log basis x={2},
    ]

\addplot [mark=otimes,olive] table [x={nstreams}, y={UA_131072}] \spua;

\addplot [mark=star,teal] table [x={nstreams}, y={UA_65536}] \spua;

\addplot [mark=pentagon,cyan] table [x={nstreams}, y={UA_32768}] \spua;

\addplot [mark=triangle,purple] table [x={nstreams}, y={UA_16384}] \spua;

\addplot [mark=o,violet] table [x={nstreams}, y={UA_8192}] \spua;

\addplot [mark=diamond,blue] table [x={nstreams}, y={UA_4096}] \spua;

\addplot [mark=10-pointed star,orange] table [x={nstreams}, y={UA_2048}] \spua;

\addplot [mark=oplus,red] table [x={nstreams}, y={UA_1024}] \spua;

\legend{UA\_131072, UA\_65536, UA\_32768, UA\_16384, UA\_8192, UA\_4096, UA\_2048, UA\_1024};

\end{axis}
\end{tikzpicture}

%% file: Tikz/time/sp_sparse_lma1.tex
\begin{tikzpicture}[scale=0.55]
\begin{axis}[        
    xlabel={$k$},
    xlabel style={font=\Large},
    ylabel={Speedup in runtime},
    ylabel style={font=\Large},
    ymajorgrids=true,
    grid style=dashed,
    legend style={  
        font=\huge,
        nodes={scale=0.5, transform shape},
        cells={anchor=west},
        at={(0.65,0.35)},
        anchor=north,
        draw=none,
        fill=none},
    xmode=log,
    ymode=log,
    log basis y={2},
    log basis x={2},
    ]

\addplot [mark=o,red] table [x={nstreams}, y={WDC_2014}] \spNLD;

\addplot [mark=triangle,blue] table [x={nstreams}, y={WDC_2014(NLD)}] \spNLD;

\legend{WDC\_2014 ($r > 1$), WDC\_2014 ($r=1$)};

\end{axis}
\end{tikzpicture}

%% file: Tikz/time/sp_sparse_lma2.tex
\begin{tikzpicture}[scale=0.55]
\begin{axis}[        
    xlabel={$k$},
    xlabel style={font=\Large},
    ylabel={Speedup in runtime},
    ylabel style={font=\Large},
    ymajorgrids=true,
    grid style=dashed,
    legend style={
        font=\huge,
        nodes={scale=0.5, transform shape},
        cells={anchor=west},
        at={(0.65,0.35)},
        anchor=north,
        draw=none,
        fill=none},
    xmode=log,
    ymode=log,
    log basis y={2},
    log basis x={2},
    ]

\addplot [mark=o,red] table [x={nstreams}, y={MOLIERE_2016}] \spNLD;

\addplot [mark=triangle,blue] table [x={nstreams}, y={MOLIERE_2016(NLD)}] \spNLD;

\legend{MOLIERE\_2016 ($r>1$), MOLIERE\_2016 ($r=1$)};

\end{axis}
\end{tikzpicture}

%% file: Tikz/time/sp_sparse_lma3.tex
\begin{tikzpicture}[scale=0.55]
\begin{axis}[        
    xlabel={$k$},
    xlabel style={font=\Large},
    ylabel={},
    ymajorgrids=true,
    grid style=dashed,
    legend style={    
        font=\huge,
        nodes={scale=0.5, transform shape},
        cells={anchor=west},
        at={(0.65,0.35)},
        anchor=north,
        draw=none,
        fill=none},
    xmode=log,
    ymode=log,
    log basis y={2},
    log basis x={2},
    ]

\addplot [mark=o,red] table [x={nstreams}, y={GAP-kron}] \spNLD;

\addplot [mark=triangle,blue] table [x={nstreams}, y={GAP-kron(NLD)}] \spNLD;

\legend{GAP-kron ($r>1$), GAP-kron ($r=1$)};

\end{axis}
\end{tikzpicture}

%% file: Tikz/time/sp_sparse_lma4.tex
\begin{tikzpicture}[scale=0.55]
\begin{axis}[        
    xlabel={$k$},
    xlabel style={font=\Large},
    ylabel={},
    ymajorgrids=true,
    grid style=dashed,
    legend style={  
        font=\huge,
        nodes={scale=0.5, transform shape},
        cells={anchor=west},
        at={(0.65,0.35)},
        anchor=north,
        draw=none,
        fill=none},
    xmode=log,
    ymode=log,
    log basis y={2},
    log basis x={2},
    ]

\addplot [mark=o,red] table [x={nstreams}, y={com-Friendster}] \spNLD;

\addplot [mark=triangle,blue] table [x={nstreams}, y={com-Friendster(NLD)}] \spNLD;

\legend{com-Friendster ($r>1$), com-Friendster ($r=1$)};

\end{axis}
\end{tikzpicture}

%% file: Tikz/time/sp_dense_lma.tex
\begin{tikzpicture}
\begin{axis}[        
    xlabel={$k$},
    xlabel style={font=\Large},
    ylabel={Speedup in runtime},
    ymajorgrids=true,
    grid style=dashed,
    legend style={
        font=\large,
        nodes={scale=0.5, transform shape},
        cells={anchor=west},
        at={(0.75,0.45)},
        anchor=north,
        draw=none,
        fill=none},
    xmode=log,
    ymode=log,
    log basis y={2},
    log basis x={2},
    ]

\addplot [mark=otimes,red] table [x={nstreams}, y={G1_4096}] \spGNLD;

\addplot [mark=star,teal] table [x={nstreams}, y={G1_2048}] \spGNLD;

\addplot [mark=pentagon,blue] table [x={nstreams}, y={G1_1024}] \spGNLD;

\addplot [mark=triangle,orange] table [x={nstreams}, y={G1_512}] \spGNLD;

\addplot [mark=o,purple] table [x={nstreams}, y={G1_4096(NLD)}] \spGNLD;

\addplot [mark=diamond,olive] table [x={nstreams}, y={G1_2048(NLD)}] \spGNLD;

\addplot [mark=10-pointed star,cyan] table [x={nstreams}, y={G1_1024(NLD)}] \spGNLD;

\addplot [mark=oplus,violet] table [x={nstreams}, y={G1_512(NLD)}] \spGNLD;

\legend{ER1\_4096 ($r>1$), ER1\_2048 ($r>1$), ER1\_1024 ($r>1$), ER1\_512 ($r>1$), ER1\_4096 ($r=1$), ER1\_2048 ($r=1$), ER1\_1024 ($r=1$), ER1\_512 ($r=1$)};

\end{axis}
\end{tikzpicture}

%% file: Tikz/time/sp_ua_537.tex
\begin{tikzpicture}[scale=0.5]
\begin{axis}[        
    xlabel={$r$},
    xlabel style={font=\LARGE},
    ylabel={Speedup in runtime},
    ylabel style={font=\LARGE},
    ymajorgrids=true,
    grid style=dashed,
    legend style={  
        font=\huge,
        nodes={scale=0.5, transform shape},
        cells={anchor=west},
        at={(0.75,0.35)},
        anchor=north,
        draw=none,
        fill=none},
    xmode=log,
    ymode=log,
    log basis y={2},
    log basis x={2},
    ]
    
\addplot [mark=o,blue] table [x={ngroups}, y={UA_1024_537}] \spuaxy;

\addplot [mark=triangle,olive] table [x={ngroups}, y={UA_512_537}] \spuaxy;

\legend{UA\_1024\_537, UA\_512\_537};

\end{axis}
\end{tikzpicture}

%% file: Tikz/time/sp_ua_268.tex
\begin{tikzpicture}[scale=0.5]
\begin{axis}[        
    xlabel={$r$},
    xlabel style={font=\LARGE},
    ylabel={},
    ymajorgrids=true,
    grid style=dashed,
    legend style={   
        font=\huge,
        nodes={scale=0.5, transform shape},
        cells={anchor=west},
        at={(0.75,0.35)},
        anchor=north,
        draw=none,
        fill=none},
    xmode=log,
    ymode=log,
    log basis y={2},
    log basis x={2},
    ]
    
\addplot [mark=o,blue] table [x={ngroups}, y={UA_2048_268}] \spuaxy;

\addplot [mark=triangle,olive] table [x={ngroups}, y={UA_1024_268}] \spuaxy;

\legend{UA\_2048\_268, UA\_1024\_268};

\end{axis}
\end{tikzpicture}

%% file: Tikz/time/sp_ua_134.tex
\begin{tikzpicture}[scale=0.5]
\begin{axis}[        
    xlabel={$r$},
    xlabel style={font=\LARGE},
    ylabel={},
    ymajorgrids=true,
    grid style=dashed,
    legend style={ 
        font=\huge,
        nodes={scale=0.5, transform shape},
        cells={anchor=west},
        at={(0.75,0.35)},
        anchor=north,
        draw=none,
        fill=none},
    xmode=log,
    ymode=log,
    log basis y={2},
    log basis x={2},
    ]
    
\addplot [mark=o,blue] table [x={ngroups}, y={UA_4096_134}] \spuaxy;

\addplot [mark=triangle,olive] table [x={ngroups}, y={UA_2048_134}] \spuaxy;

\legend{UA\_4096\_134, UA\_2048\_134};

\end{axis}
\end{tikzpicture}

%% file: Tikz/time/sp_ua_67.tex
\begin{tikzpicture}[scale=0.5]
\begin{axis}[        
    xlabel={$r$},
    xlabel style={font=\LARGE},
    ylabel={},
    ymajorgrids=true,
    grid style=dashed,
    legend style={
        font=\huge,
        nodes={scale=0.5, transform shape},
        cells={anchor=west},
        at={(0.75,0.35)},
        anchor=north,
        draw=none,
        fill=none},
    xmode=log,
    ymode=log,
    log basis y={2},
    log basis x={2},
    ]
    
\addplot [mark=o,blue] table [x={ngroups}, y={UA_8192_67}] \spuaxy;

\addplot [mark=triangle,olive] table [x={ngroups}, y={UA_4096_67}] \spuaxy;

\legend{UA\_8192\_67, UA\_4096\_67};

\end{axis}
\end{tikzpicture}